\documentclass[11pt]{article}
\usepackage{epsf}
\usepackage{epsfig}
\usepackage{graphicx}
\usepackage{color}
\usepackage{amsmath}
\usepackage{mathrsfs}

\usepackage{amsthm}
\usepackage{latexsym}
\usepackage{amssymb}
\usepackage{amsmath}

\usepackage{stmaryrd}

\usepackage{epsf}
\usepackage{verbatim}

\usepackage[margin=1in]{geometry}
\usepackage{xspace}
\usepackage{xcolor}
\usepackage{setspace}

\newtheorem{theorem}{Theorem}[section]
\newtheorem{lemma}[theorem]{Lemma}
\newtheorem{corollary}[theorem]{Corollary}
\newtheorem{definition}[theorem]{Definition}

\newtheorem{fact}{Fact}

\newcommand{\methsepahr}{domain separated holographic reduction }
\newcommand{\methpoly}{the polynomial argument }
\newcommand{\methps}{the partial symmetry argument }

\begin{document}

\title{{\bf Dichotomy for Holant* Problems with a Function on Domain Size 3}}

\vspace{0.3in}
\author{Jin-Yi Cai\thanks{University of Wisconsin-Madison.
 {\tt jyc@cs.wisc.edu}. Supported by NSF CCF-0914969.}
\and Pinyan Lu\thanks{Microsoft Research Asia. {\tt
pinyanl@microsoft.com}}
\and Mingji Xia\thanks{Max-Planck Institut f\"ur Informatik,
 {\tt xmjljx@gmail.com}. Supported by NSFC 61003030. }}

\date{}
\maketitle

\bibliographystyle{plain}

\begin{abstract}
Holant problems are a general framework to study
the algorithmic complexity of counting problems.
Both counting constraint satisfaction
problems and graph homomorphisms are special cases.
All previous results of Holant problems are over the Boolean domain.
In this paper, we give the first dichotomy theorem for Holant problems
for domain size $>2$.
We discover unexpected tractable families of counting problems,
by giving new polynomial time algorithms.
This paper also initiates holographic reductions in domains of
size $>2$. This is our main algorithmic technique,
and is used for both tractable families and hardness reductions.
The dichotomy theorem is the following:
For any complex-valued symmetric function ${\bf F}$ with arity 3
on domain size 3, we give an explicit criterion on ${\bf F}$,
such that if ${\bf F}$ satisfies
the criterion then the problem ${\rm Holant}^*({\bf F})$ is computable
in polynomial time, otherwise ${\rm Holant}^*({\bf F})$ is \#P-hard.
\end{abstract}


\section{Introduction}

The study of computational complexity of counting problems has been a very
active research area recently.  Three related frameworks in which
counting problems can be expressed as partition functions have
received the most attention:
Graph Homomorphisms (GH), Constraint Satisfaction Problems (CSP)
and Holant Problems.

Graph Homomorphism was first defined by Lov\'{a}sz~\cite{lovasz67}.
It captures a wide variety of graph properties. Given any fixed
$k \times k$ symmetric matrix ${\bf A}$ over $\mathbb{C}$,
the partition function
$Z_{\bf A}$ maps any input graph $G = (V, E)$
to $Z_{\bf A}(G) =  \sum_{\xi:V\rightarrow [k]}\hspace{0.07cm}
\prod_{(u,v)\in E} {\bf A}_{\xi(u),\xi(v)}$.
When  ${\bf A}$ is a 0-1 matrix, then
the product $\prod_{(u,v)\in E}$ is essentially a Boolean {\sc And} function.
The product value $\prod_{(u,v)\in E} {\bf A}_{\xi(u),\xi(v)} =0$ or 1,
and it is
1 iff every edge $(u,v)\in E$ is mapped to an edge in the graph $H$ whose
adjacency matrix is ${\bf A}$. Hence for a 0-1 matrix
${\bf A}$, $Z_{\bf A}(G)$ counts the
number of ``homomorphisms'' from $G$ to $H$. For example,
if ${\bf A} = {\scriptsize \begin{bmatrix}  1 & 1 \\ 1 & 0 \end{bmatrix}}$
then $Z_{\bf A}(G)$ counts the number of {\sc Independent Sets} in $G$.
If ${\bf A} = {\tiny \begin{bmatrix} 0 & 1 & 1 \\ 1 & 0 & 1 \\ 1 & 1 & 0
\end{bmatrix}}$ then $Z_{\bf A}(G)$ is the number of valid {\sc 3-colorings}.
When ${\bf A}$ is not 0-1, $Z_{\bf A}(G)$ is a weighted sum of
homomorphisms. Each ${\bf A}$
defines  a graph property on graphs $G$.  Clearly if $G$ and $G'$ are
isomorphic then $Z_{\bf A}(G) = Z_{\bf A}(G')$.
While individual graph properties are fascinating to study,
 Lov\'{a}sz's intent is to study a wide class of graph
 properties representable as graph homomorphisms.
The use of more general matrices ${\bf A}$ brings us into contact with
another tradition, called {\it partition functions
of spin systems}
 from statistical physics~\cite{baxter1982exactly,mccoy1973two}.
The case of a $2 \times 2$ matrix ${\bf A}
= {\scriptsize \begin{bmatrix}  \beta & 1 \\ 1 & \gamma \end{bmatrix}}$
is called a 2-spin system, and the special case $\beta = \gamma$
is the Ising model~\cite{ising1925beitrag, JS93, GJPaterson}. The Potts model with
interaction strength $\gamma$ is defined by a
$k \times k$ matrix ${\bf A}$ where all off-diagonal entries equal to
1 and all diagonal entries equal to $1 + \gamma$~\cite{inapp_GJ10}.
In classical physics, the matrix ${\bf A}$ is always
real-valued.
However, in a generic quantum system for which complex numbers
are the right language, the partition function is in general
complex-valued \cite{feynman1970}. In particular, if the physics model is
over a discrete graph and is non-orientable, then the edge weights are
  given by a symmetric complex matrix.
We will see that the use of complex numbers is not just a modeling issue,
it provides an inner unity in the algorithmic
theory of partition functions.


A more general framework than GH is called counting CSP. Let ${\cal F}$
be any finite set of (complex-valued) constraint functions
defined on some domain set $D$.
It defines a counting CSP problem ${\rm \#CSP}({\cal F})$:
An input consists of a bipartite graph $G = (X, Y, E)$,
each $x \in X$ is a variable on $D$, each $y \in Y$
is labeled by a constraint function $f \in {\cal F}$, and
the edges in $E$ indicate how each constraint function is applied.
The output is the sum of product of evaluations of the constraint functions
over all assignments for the
variables~\cite{CreignouH96,BulatovD03,weightedCSP,Bulatov08,STOC09,Dyer-Rich,ccl-csp}.
Again if all constraint functions in ${\cal F}$
are 0-1 valued then it counts the number of solutions.
In general, this {\it sum of product}  a.k.a.~{\it partition
function} is a weighted sum of solutions, and has occupied a central
position. It reaches many areas
ranging from  AI, machine learning, tensor networks,
statistical physics and coding theory.
Note that GH is the special case where ${\cal F}$ consists of
a single binary symmetric function.

The strength of these frameworks derives from the fact that
they can express many problems of interest and simultaneously
it is possible to achieve a complete
classification of its worst case complexity.
%


While GH (or spin systems) can express a great variety of natural
counting problems, Freedman, Lov\'{a}sz and Schrijver \cite{freedman-l-s}
showed that  GH cannot express the problem of counting {\sc Perfect Matchings}.
It is well known that the FKT algorithm~\cite{Kasteleyn1961,TF1961}
 can count the
number of perfect matchings in a planar graph in polynomial time.
This is one basic component of
holographic algorithms recently introduced by Valiant~\cite{HA_J,AA_FOCS}.
(The second basic component is holographic reduction.)
To capture this extended class of problems typified by
{\sc Perfect Matchings}, the framework of Holant problems
was introduced~\cite{FOCS08,STOC09,holant}.
Briefly, an input instance of a Holant problem is a graph $G = (V, E)$
where each edge represents a variable and each vertex is labeled by
a constraint function. The partition function is again the sum of
product of the constraint function evaluations,
 over all edge assignments.
 E.g., if edges are Boolean variables (i.e., domain size 2),
and the constraint function at every vertex is the {\sc Exact-One}
function which is 1 if exactly one incident edge is assigned true
and 0 otherwise, then the partition function  counts the
number of perfect matchings.  If each vertex has the {\sc At-Most-One}
function then it counts all (not necessarily perfect) matchings.
It can be shown easily that the Holant framework can simulate
spin systems but, as shown by \cite{freedman-l-s},
the converse is not true.
The Holant framework turns out to be a very natural setting
and captures many interesting problems.
E.g., it was independently discovered in coding theory, where
it is called Normal Factor Graphs or Forney Graphs~\cite{Forney01,FV11,Al-BashabshehMV11,Al-BashabshehM11}.

A complexity dichotomy theorem for counting
problems classifies every problem within a class to be
 either in P or \#P-hard. For GH,  this is proved
for $Z_{\bf A}$ for all symmetric complex matrices
 ${\bf A}$~\cite{Homomorphisms}.
This is a culmination of a long series of
results~\cite{DyerG00,BulatovG05,GGJT09}.
The proof of~\cite{Homomorphisms} is difficult,
but the tractability criterion is very
explicit: $Z_{\bf A}$ is in polynomial time if ${\bf A}$
is a suitable rank-one modification of a tensor product of Fourier
matrices, and is \#P-hard otherwise.
Explicit dichotomy theorems were also proved for counting CSP
on the Boolean domain (i.e., $|D|=2$): unweighted~\cite{CreignouH96}, non-negative weighted~\cite{weightedCSP}, real weighted~\cite{CSP-real}, and finally complex weighted~\cite{STOC09},
where holographic reductions
played an important role in the final result.
Complex numbers make their appearance naturally as eigenvalues,
and provide an internal logic to the theory,
even if one is only interested in 0-1 valued constraint functions.

When we go from the Boolean domain to domain size $> 2$,
there is a huge increase in difficulty to prove dichotomy theorems.
 This is already
seen in decision CSP, where the dichotomy (i.e., any decision CSP
is either in P or NP-complete) for the Boolean domain
is Schaefer's theorem~\cite{Schaefer},
but the dichotomy for domain size 3 is a major achievement by
Bulatov~\cite{Bulatov06}.
A long standing conjecture by Feder and Vardi~\cite{Feder-Vardi} states that
a dichotomy for decision CSP holds for all domain size,
but this is open for domain size $>3$.
The assertion that every decision CSP is either solvable
in polynomial time or NP-complete is by no means obvious,
since assuming P $\not =$ NP, Ladner showed that NP contains
problems that are neither in P nor NP-complete~\cite{Lad75}.
This is also valid for P versus \#P.

With respect to counting problems, for  any finite set of 0-1 valued  functions $\cal F$
over a general domain, Bulatov~\cite{Bulatov08}
proved a dichotomy theorem for ${\rm \#CSP}({\cal F})$, which uses
deep results from Universal Algebra.
Dyer and Richarby~\cite{Dyer-Rich,Dyer-Rich2} gave a more direct proof which has the advantage
that their tractability criterion is decidable.
Decidable dichotomy theorems are more desirable since they
tell us not only every ${\cal F}$
belongs to either one or the other class, but also how to decide
for a given ${\cal F}$ which class it belongs to.
A decidable dichotomy theorem for ${\rm \#CSP}({\cal F})$,
where all functions in $\cal F$ take non-negative values,
is given in \cite{ccl-csp}. Finally a dichotomy theorem for
all complex-valued ${\rm \#CSP}({\cal F})$ is proved in~\cite{caichen12}.
This last dichotomy is not known to be decidable.

More than giving a formal classification, the deeper meaning of
a dichotomy theorem is to provide  a  comprehensive
 structural understanding
as to what makes  a problem easy and what makes it hard.
This deeper understanding goes beyond the validity of a dichotomy,
and even  more than
decidability, which is:
Given $\cal F$, decide whether it satisfies the tractability criterion
so that  ${\rm \#CSP}({\cal F})$ is in P.
 Ideally we hope for dichotomy theorems
that are {\it explicit} in the sense that the tractability criteria provide
a mathematical characterization that can be applied
symbolically to an arbitrary $\cal F$.
An explicit dichotomy can also be readily used to prove broader
dichotomy theorems, as we will see in this paper.
The known dichotomy theorems for GH~\cite{Homomorphisms} and
for CSP on general domains have very different flavors.
 Dichotomy theorems for ${\rm \#CSP}({\cal F})$ for all domain size
$>2$ are not explicit. The tractability criterion is infinitary.
This is in marked contrast with the dichotomy theorems for GH.
%
For Holant problems all previous results are over the Boolean domain
and are mostly explicit.
In this paper, we give the first dichotomy theorem for
Holant problems for domain size $>2$, and it is explicit.

Our main theorem can be stated as follows:
For any complex-valued symmetric function ${\bf F}$ with arity 3 on domain size
3, we give an explicit criterion on ${\bf F}$, such that if ${\bf F}$ satisfies
the criterion then the problem ${\rm Holant}^*({\bf F})$ is computable
in polynomial time, otherwise ${\rm Holant}^*({\bf F})$ is \#P-hard.
(Formal definitions will be given in Section~\ref{sec:notation}.)
It is known that in the Holant framework any set of binary functions
is tractable.
A ternary function is the basic setting in the Holant framework
where both tractable and intractable cases occur.
A single ternary function in the Holant framework is
the analog of GH as the basic setting in the CSP framework
with a single binary function.
 Therefore this
case is interesting in its own right.
Furthermore, as demonstrated many times in the Boolean domain~\cite{STOC09,holant,CHL09,parity,holant-real},
a dichotomy for a single ternary function serves as the starting
point for more general dichotomies in the Holant framework.

In order to prove this dichotomy theorem, we have to discover
new tractable classes of Holant problems, and design new polynomial
time algorithms. Many intricacies of the interplay between
tractability and intractability do not occur in the Boolean
domain.  However these new algorithms actually provide fresh insight
to our previous dichotomy theorems for the Boolean domain.
They offer a deeper and more complete understanding
of what makes  a problem easy and what makes it hard.

Our main algorithmic innovation is to initiate
the theory of holographic reductions in domains of
size $>2$.
It is a recurring theme in our proof techniques here.
This is a new development; all previous work
on holographic algorithms and reductions have been
on the Boolean domain.
Holographic transformation offers a perspective on internal
connections and equivalences between different looking problems, that is
unavailable by any other means. In particular since it naturally
uses eigenvalues and eigenvectors, the  field of
 complex numbers $\mathbb{C}$
is the natural setting to formulate the class of problems, even if one is only
interested in 0-1 valued or non-negative valued constraint functions.
Using complex-valued constraints in defining Holant problems
we can see the internal logical connections between various problems.
Completely different looking problems can be seen as one and  the same
problem under holographic transformations.
The  proof of our dichotomy theorem would be impossible
without working over $\mathbb{C}$. Even the
dichotomy criterion would be impossible to state without it.
To quote Jacques Hadamard:
{\it ``The shortest path between two truths on the real line
passes through the complex plane."}

Suppose our domain set is $\{B, G, R\}$, named for the three colors
Blue, Green and Red.
We isolate several classes of tractable cases of ${\bf F}$.
One of them is a generalization of Fibonacci signatures from the Boolean
domain, under an orthogonal transformation. Another involves
a concept called isotropic vectors, which self-annihilates under
dot product. The third type involves a more intricate interplay
between an isotropic vector in some dimension and another function
primarily ``living'' in
the other dimensions. This last type was only discovered after
we failed to push through certain  hardness proofs.

For hardness proofs,
the first main idea  is to construct a binary
function which acts as an {\sc Equality} function when
restricted to $\{G, R\}$, and
is zero elsewhere.  This construction allows us to restrict
a function on  $\{B, G, R\}$ to a domain of size 2, and  employ
the known (and explicit) dichotomy theorems for the Boolean domain.
The plan is to use it to restrict  ${\bf F}$ to  $\{G, R\}$
and, assuming it is non-degenerate, to {\it anchor} the entire
hardness proof on that. Here it is crucial that
the known Boolean domain dichotomy is explicit. This part of the proof is quite
demanding and heavily depends on holographic reductions.
A central motif  is to show that after a holographic
reduction, ${\bf F}$ must possess {\it fantastic} regularity to escape
\#P-hardness.

What perhaps took us by surprise is that when ${\bf F}$ restricted
to  $\{G, R\}$
is degenerate, there is still considerable technical difficulty
remaining.  These are eventually overcome by using
unsymmetric functions.

This work has been a marathon for us.
During the process, repeatedly,
we failed to clinch the hardness proof for some subclasses
of functions
and then new tractable cases were found. So we had to reformulate
 the final dichotomy  several times. The discovery process
 is mutually
reinforcing between new algorithms and hardness proofs.  On many occasions
we believed that we had overcome one last hurdle, only to be
stymied by yet another. However the struggle has also
paid handsome dividends. For example, our SODA paper two years
ago~\cite{SODA11} was obtained as part of the program to achieve
this dichotomy. We realized we needed a dichotomy for unsymmetric
functions over the Boolean domain, and indeed that is used
to overcome a major difficulty in the proof here.

\section{Preliminary}\label{sec:notation}

\subsection{Definitions}

Definitions of Holant problem and gadget are introduced in this subsection. The readers who are familiar with the definitions in \cite{holant, SODA11} may skip.

Let $D$ be a finite domain set, and $\cal F$ be a finite
set of constraint functions called signatures.
Each $\mathbf{F} \in {\cal F}$ is a mapping from $D^k \to \mathbb{C}$
for some arity $k$. We assume signatures take complex algebraic numbers.

A \emph{signature grid} $\Omega = (G, {\cal F}, \pi)$ consists of
a graph $G = (V,E)$ where each vertex $v \in V$ is labeled by a
function $\mathbf{F}_v \in \mathbb{C}$, and $\pi$ is the labeling.
The Holant problem on instance $\Omega$ is to evaluate
\begin{equation} \label{equ:Holant-def}
{\rm Holant}_\Omega = \sum_{\sigma} \prod_{v \in V}
\mathbf{F}_v(\sigma \mid_{E(v)}),
\end{equation}
 a sum over all edge assignments $\sigma: E \rightarrow D$,
where $E(v)$ denotes the incident edges at $v$.

A function $\mathbf{F}_v$ is listed by its values lexicographically
 as a truth table, or as a tensor in $(\mathbb{C}^{|D|})^{\otimes \deg(v)}$.
We can identify a unary function $\mathbf{F}(x) : D \rightarrow \mathbb{C}$
with  a  vector  ${\bf u} \in \mathbb{C}^{|D|}$.
Given two vectors ${\bf u}$ and ${\bf v}$ of dimension $|D|$,
the tensor product  ${\bf u} \otimes {\bf v}$ is a vector in
 $\mathbb{C}^{|D|^2}$, with entries $u_i v_j$ ($1 \le i, j \le |D|$).
For matrices $A = (a_{i,j})$ and $B=(b_{k, l})$,
the tensor product (or Kronecker product) $A \otimes B$ is defined similarly;
it has  entries $a_{i,j} b_{k, l}$ indexed
by  $((i, k), (j, l))$ lexicographically.
We write ${\bf u}^{\otimes k}$ for ${\bf u}
\otimes \ldots \otimes {\bf u}$ with $k$ copies of ${\bf u}$.
$A^{\otimes k}$ is similarly defined.
We have
$(A \otimes B) (A' \otimes B') = (AA' \otimes BB')$ whenever
the matrix products are defined.
In particular,
$A^{\otimes k} ({\bf u}_1 \otimes \ldots \otimes {\bf u}_k)
= A{\bf u}_1 \otimes \ldots \otimes A{\bf u}_k$
when the matrix-vector products $A{\bf u}_i$ are defined.

A signature $\mathbf{F}$ of arity $k$ is  \emph{degenerate}
if $\mathbf{F} = {\bf u}_1 \otimes {\bf u}_2 \otimes \ldots \otimes {\bf u}_k$
for some vectors ${\bf u}_i$. Equivalently there are unary
functions $\mathbf{F}_i$ such that $\mathbf{F}(x_1, \ldots, x_k)
= \mathbf{F}_1(x_1) \cdots \mathbf{F}_k(x_k)$.
Such a signature is very weak; there is no interaction between
the variables.  If every function in ${\cal F}$
is degenerate, then ${\rm Holant}_\Omega$ for
any $\Omega = (G, {\cal F}, \pi)$  is computable in polynomial time
in a trivial way: Simply split every vertex $v$ into $\deg(v)$ many
vertices each assigned a unary $\mathbf{F}_i$ and connected to the incident
edge. Then ${\rm Holant}_\Omega$ becomes a product over each
component of a single edge.  Thus degenerate signatures
are weak and should be properly understood as made up by
unary signatures.  To concentrate on the essential features that
differentiate tractability from intractability, we introduced
Holant$^*$ problems~\cite{STOC09,holant}. These  are Holant problems
where unary signatures are assumed to be present.


We consider a type of graphs $G = (V,I,E)$ with
 two kinds of edges. Edges in $I$ are ordinary
 internal edges with two endpoints. Edges in $E$ are external edges (also called dangling edges) which have only one endpoint in $V$. Such a graph
can be made into a part of a larger graph as follows.
Given a graph  $G'$,
we can replace a vertex $v$ of $G'$ by a graph $G$ with
 external edges, merging the external edges with the incident
edges of  $v$.
Reversely, when some edges are cut from a graph, the cut edges become external edges on both sides. When two external edges are connected, they merge to
become one edge.

A \emph{gadget} consists of a graph $G = (V,I,E)$
and a labeling $\pi$, where each vertex $v \in V$ is labeled by a
function $\mathbf{F}_v \in \mathbb{C}$. A gadget can be a part of a signature grid. For example, in a signature grid, a single vertex of degree $d$  constitutes a gadget. It has
the single vertex,  together with its function, an empty $I$ set,
and its $d$ incident edges  as external edges. It can be replaced by a gadget $G$ with $|E|=d$ and vice versa. In a signature grid, when we want to replace a gadget $G$ with $|E|=d$  by a vertex $v$ of degree $d$, what is the right function $\mathbf{F}_v$ that keeps the value of the signature grid unchanged?
 The function of a gadget is defined to have this property, and it is also a natural generalization of Holant.

On an assignment $\tau: E \rightarrow D$, the function $\mathbf{F}_G$ of a gadget $G$ has value
\[
\mathbf{F}_G (\tau) = \sum_{\sigma } \prod_{v \in V}
\mathbf{F}_v(\tau \sigma \mid_{E(v)}),
\]
 a sum over all edge assignments $\sigma: I \rightarrow D$,
and $\tau \sigma$ is the combined assignment on $E \cup I$.


Suppose one gadget is the disjoint union of two parts, each has
two external edges. Suppose
the binary functions (on $x_1,x_2$ and $x_3,x_4$ respectively)
 in matrix form are $\mathbf{A}$
and $\mathbf{B}$.
Then the function of this gadget is $\mathbf{F}_{x_1x_3,x_2x_4}=\mathbf{A}_{x_1,x_2} \otimes \mathbf{B}_{x_3,x_4}$,  where $\mathbf{F}_{x_1x_3,x_2x_4}$ denotes the matrix with two indices $x_1x_3$ and $x_2x_4$, and the value of this entry is just $\mathbf{F}(x_1,x_2,x_3,x_4)$.

Another example is the following.
There are two binary functions $\mathbf{A}$ and
$\mathbf{B}$. They share an internal edge $x_2$. Other two edges $x_1,x_3$ are external. The function of this gadget is $\mathbf{F}(x_1,x_3)=\sum_{x_2} \mathbf{A}(x_1,x_2) \mathbf{A}(x_2,x_3)$, that is, $\mathbf{F}_{x_1, x_3}=\mathbf{A}_{x_1,x_2} \mathbf{A}_{x_2,x_3}$,
the matrix product.

\subsection{Holographic Reduction}

To introduce the idea of holographic reductions, it is convenient to consider bipartite graphs.
For a general graph, we can always transform it into a bipartite graph
preserving the Holant value, as follows:
For each edge in the graph, we replace it by a path of length $2$, and assign to
 the new vertex the binary \textsc{Equality} function $(=_2)$.

We use ${\rm Holant}(\mathcal{R} \mid \mathcal{G})$ to denote the Holant problem on bipartite graphs $H = (U,V,E)$, where each signature for a vertex in $U$ or $V$ is from $\mathcal{R}$ or $\mathcal{G}$, respectively.
An input instance for the bipartite Holant problem is a bipartite signature grid and is denoted as $\Omega = (H;\ \mathcal{R} \mid \mathcal{G};\ \pi)$.
Signatures in $\mathcal{R}$ are considered as row vectors (or covariant tensors); signatures in $\mathcal{G}$ are considered as column vectors (or contravariant tensors)~\cite{DP91}.

For a $|D| \times |D|$ matrix $T$ and a signature set $\mathcal{F}$, define $T \mathcal{F} = \{\mathbf{G} \mid \exists \mathbf{F} \in \mathcal{F}$ of arity $n,~\mathbf{G} = T^{\otimes n} \mathbf{F}\}$, similarly for $\mathcal{F} T$.
Whenever we write $T^{\otimes n} \mathbf{F}$ or $T \mathcal{F}$, we view the signatures as column vectors; similarly for $\mathbf{F} T^{\otimes n} $ or $\mathcal{F} T$ as row vectors.
A holographic transformation by $T$ is the following operation: given a signature grid $\Omega = (H;\  \mathcal{R} \mid \mathcal{G};\ \pi)$, for the same graph $H$, we get a new grid $\Omega' = (H;\  \mathcal{R} T \mid T^{-1} \mathcal{G};\  \pi)$ by replacing each signature in $\mathcal{R}$ or $\mathcal{G}$ with the corresponding signature in $\mathcal{R} T$ or $T^{-1} \mathcal{G}$.

\begin{theorem}[Valiant's Holant Theorem~\cite{HA_J}]
 If there is a holographic transformation mapping signature grid $\Omega$ to $\Omega'$, then ${\rm Holant}_\Omega = {\rm Holant}_{\Omega'}$.
\end{theorem}

Therefore, an invertible holographic transformation does not change the complexity of the Holant problem in the bipartite setting.
We illustrate the power of holographic transformation by an example.
Let $\mathbf{F} = [\frac{3}{2}, 0, \frac{1}{2}, 0, \frac{3}{2}]$.
Consider  ${\rm Holant}(\mathbf{F})$ on the Boolean domain.
For a 4-regular graph $G$,
${\rm Holant}(\mathbf{F})$
 is a sum over all 0-1 edge assignments of products of local evaluations.
Each vertex contributes a factor $\frac{3}{2}$
 if all incident edges are assigned the same truth value,
 a factor $\frac{1}{2}$ if exactly half are assigned 1 and the other
half 0.
Before anyone consigns this problem to be artificial and unnatural,
consider a
holographic transformation by
$Z = \frac{1}{\sqrt{2}} {\scriptsize
\begin{bmatrix} 1 & 1 \\ i & -i \end{bmatrix}}$.
Then ${\rm Holant}(\mathbf{F}) = {\rm Holant}( =_2 \mid \mathbf{F}) =
{\rm Holant}( (=_2) Z^{\otimes 2} \mid (Z^{-1})^{\otimes 4} \mathbf{F})$.
Let $\hat{\mathbf{F}} = [0,0,1,0,0]$, and writing it as a symmetrized sum of
tensor products, then
\begin{align*}
 Z^{\otimes 4} \hat{\mathbf{F}}
 &=
 Z^{\otimes 4}
 {\scriptsize \left\{
 \begin{bmatrix} 1 \\ 0 \end{bmatrix}
 \otimes \begin{bmatrix} 1 \\ 0 \end{bmatrix}
 \otimes \begin{bmatrix} 0 \\ 1 \end{bmatrix}
 \otimes \begin{bmatrix} 0 \\ 1 \end{bmatrix}
 +
 \begin{bmatrix} 1 \\ 0 \end{bmatrix}
 \otimes \begin{bmatrix} 0 \\ 1 \end{bmatrix}
 \otimes \begin{bmatrix} 1 \\ 0 \end{bmatrix}
 \otimes \begin{bmatrix} 0 \\ 1 \end{bmatrix}
 +
 \dotsb
 +
 \begin{bmatrix} 0 \\ 1 \end{bmatrix}
 \otimes \begin{bmatrix} 0 \\ 1 \end{bmatrix}
 \otimes \begin{bmatrix} 1 \\ 0 \end{bmatrix}
 \otimes \begin{bmatrix} 1 \\ 0 \end{bmatrix}
 \right\} } \\
 &=
 \tfrac{1}{4}
 {\scriptsize \left\{
 \begin{bmatrix} 1 \\ i \end{bmatrix}
 \otimes \begin{bmatrix} 1 \\ i \end{bmatrix}
 \otimes \begin{bmatrix} 1 \\ -i \end{bmatrix}
 \otimes \begin{bmatrix} 1 \\ -i \end{bmatrix}
 +
 \begin{bmatrix} 1 \\ i \end{bmatrix}
 \otimes \begin{bmatrix} 1 \\ -i \end{bmatrix}
 \otimes \begin{bmatrix} 1 \\ i \end{bmatrix}
 \otimes \begin{bmatrix} 1 \\ -i \end{bmatrix}
 +
 \dotsb
 +
 \begin{bmatrix} 1 \\ -i \end{bmatrix}
 \otimes \begin{bmatrix} 1 \\ -i \end{bmatrix}
 \otimes \begin{bmatrix} 1 \\ i \end{bmatrix}
 \otimes \begin{bmatrix} 1 \\ i \end{bmatrix}
 \right\} } \\
 &= \tfrac{1}{2} [3, 0, 1, 0, 3] = \mathbf{F};
\end{align*}
Hence the contravariant transformation
$(Z^{-1})^{\otimes 4} \mathbf{F} = \hat{\mathbf{F}}$.
Meanwhile, a covariant transformation by $Z$ transforms $(=_2)$ to the binary \textsc{Disequality} function $(\neq_2)$
\[
 (=_2) Z^{\otimes 2}
= {\scriptsize \begin{pmatrix} 1 & 0 & 0 & 1 \end{pmatrix}
 } Z^{\otimes 2}
 =              {\scriptsize \left\{ \begin{pmatrix} 1 & 0 \end{pmatrix}^{\otimes 2} + \begin{pmatrix} 0 &  1 \end{pmatrix}^{\otimes 2} \right\} } Z^{\otimes 2}
 = \tfrac{1}{2} {\scriptsize \left\{ \begin{pmatrix} 1 & 1 \end{pmatrix}^{\otimes 2} + \begin{pmatrix} i & -i \end{pmatrix}^{\otimes 2} \right\} }
 = [0, 1, 0]
 = (\neq_2).
\]
So
 ${\rm Holant}(\mathbf{F}) = {\rm Holant}((\neq_2) \mid [0,0,1,0,0])$; they are
really one and the same problem.
A moment's reflection shows that this latter formulation
is counting the number of
Eulerian orientations on 4-regular graphs, an eminently natural problem!

Furthermore,
holographic transformation by an orthogonal matrix $T$
 preserves the binary equality
and thus can be used freely in the standard setting.

\begin{theorem} \label{thm:orthogoal}
 Suppose $T$ is an orthogonal matrix $(T T^{\tt T} = I)$ and let $\Omega = (G, \mathcal{F}, \pi)$ be a signature grid.
 Under a holographic transformation by $T$, we get a new grid $\Omega' = (G, T \mathcal F, \pi)$ and ${\rm Holant}_\Omega = {\rm Holant}_{\Omega'}$.
\end{theorem}

%


 When $T$ has a special $\{B\}$ and $\{G,R\}$ domain separated form, we observe that
each $\{G,R\}$-line in the table for $T^{\otimes 3} \mathbf{F}$
and $\mathbf{F}$,
which correspond to a fixed number of $B$ assigned,
are closely related by the $\{G,R\}$-block of $T$, as stated in
the following Fact. We call this a {\it \methsepahr}.

\begin{fact} \label{fact-domain-sepa-HR}
Suppose
$T$ is in the $\{B\}$ and $\{G,R\}$ domain separated form,
$\left ( \begin{matrix} e & 0 & 0 \\
   0 & a & b \\
   0 & c & d
\end{matrix} \right )$.
Let $M = \left ( \begin{matrix}  a & b \\   c & d
\end{matrix} \right )$.  We have,
\begin{eqnarray*}
 (T^{\otimes 3} \mathbf{F})^{*\rightarrow \{G,R\}}
&=&M^{\otimes 3} (\mathbf{F}^{*\rightarrow \{G,R\}}), \\
 (T^{\otimes 3} \mathbf{F})^{1=B,2, 3 \rightarrow \{G,R\}}
&=&e M^{\otimes 2} (\mathbf{F}^{1=B,2,3\rightarrow \{G,R\}}), \\
(T^{\otimes 3} \mathbf{F})^{1=B, 2=B, 3 \rightarrow \{G,R\}}
&=&e^2 M (\mathbf{F}^{1=B, 2=B,3\rightarrow \{G,R\}}).
\end{eqnarray*}
\end{fact}

\begin{proof}
We prove the second formula as an example. Other formulae can be proved similarly.

$(T^{\otimes 3} \mathbf{F})^{1=B,2,3 \rightarrow \{G,R\}}$ is the line $[(T^{\otimes 3} \mathbf{F})_{BGG},(T^{\otimes 3} \mathbf{F})_{BGR},(T^{\otimes 3} \mathbf{F})_{BRR}]$ in the triangular table form of $T^{\otimes 3} \mathbf{F}$, and $\mathbf{F}^{1=B,2,3\rightarrow \{G,R\}}$ is the corresponding
 line of $\mathbf{F}$.

Because one input of $T^{\otimes 3} \mathbf{F}$ is fixed to $B$, it is
equivalent to connecting one unary function $(1,0,0)$
to $T^{\otimes 3} \mathbf{F}$.
By associativity this unary can be combined with a copy of
$T$ in the gadget  $T^{\otimes 3} \mathbf{F}$. This combination
results in a unary function $\langle (1,0,0), T \rangle= (T_{B,B}, 0, 0)
= (e,0,0)$,
which is
then  connected to $\mathbf{F}$. This creates a binary function
$e \mathbf{F}^{1=B}$.  Now, we get $(T^{\otimes 2}(e \mathbf{F}^{1=B}))^{*\rightarrow \{G,R\}}$. The two external edges of the
gadget $T^{\otimes 2}(e \mathbf{F}^{1=B})$ are restricted to $\{G,R\}$. Because the domain of $T$ is separated into $\{B\}$ and $\{G,R\}$, they force the two internal edges to take values in $\{G,R\}$. Since all 4 edges take values in $\{G,R\}$, this turns $T^{\otimes 2}(e\mathbf{F}^{1=B})$ into $e M^{\otimes 2} \mathbf{F}^{1=B,2,3\rightarrow \{G,R\}}$.
\end{proof}

\subsection{Notations}

A signature $\mathbf{F}$ on $r$ variables is symmetric
if $\mathbf{F}(x_1, \ldots, x_r) = \mathbf{F}(x_{\sigma (1)},  \ldots, x_{\sigma (r)})$
for all $\sigma \in \mathfrak{S}_k$, the \emph{symmetric group}.
It can be shown easily that a symmetric signature $\mathbf{F}$  is degenerate iff
$\mathbf{F} = {\bf u}^{\otimes r}$ for some unary ${\bf u}$.

We use $\rm{Sym}(\mathbf{F})$ to denote the symmetrization of $\mathbf{F}$ as follows:
%
%
For $i_1,i_2,\ldots, i_r \in \{B, G, R\}$,
\[{(\rm{Sym}(\mathbf{F}))}_{(i_1i_2\ldots i_r)}
= \sum_{\sigma\in \mathfrak{S}_r}
F_{i_{\sigma 1}i_{\sigma 2} \cdots i_{\sigma r}},\]
where the summation is over the \emph{symmetric group}
$\mathfrak{S}_r$ on $r$ symbols. \footnote{Usually, there is a normalization factor $\frac{1}{r!}$  in front of the
summation, however a global factor does not change the complexity and we ignore this factor for notational simplicity.}

If $\mathbf{F}$ is degenerate, given as a simple tensor product
\[\mathbf{F}=v_1\otimes v_2\otimes\cdots \otimes v_r, \]
then the  symmetrization of $\mathbf{F}$ is the symmetric product of the factors:

\[{\rm{Sym}(\mathbf{F})}= \sum_{\sigma\in\mathfrak{S}_r}
v_{{\sigma 1}}\otimes v_{{\sigma 2}}\otimes\cdots\otimes v_{{\sigma r}}.\]

We consider a function $\mathbf{F}$ and its nonzero multiple $c\mathbf{F}$
  as the same function,
as $c\mathbf{F}$ only introduces a easily computable global factor.

A symmetric signature $\mathbf{F}$ on $r$ Boolean variables
can be expressed as $[f_0,f_1,\dotsc,f_k]$, where $f_j$ is the value of $\mathbf{F}$ on inputs of Hamming weight $j$.
In the following, we focus on symmetric signatures over domain $[3]$.
We use three symbols $\{B,G,R\}$ to denote the domain elements.

Let $\mathbf{F}$ be a symmetric signatures of arity $3$ over domain $\{B,G,R\}$. We use the following notation for $\mathbf{F}$.
\[ \mathbf{F}=[F_{BBB}; F_{BBG}, F_{BBR}; F_{BGG}, F_{BGR}, F_{BRR}; F_{GGG}, F_{GGR}, F_{GRR}, F_{RRR}].\]

Alternatively we also use the following notation:

\begin{equation}\label{domain-3-sig}
\begin{tabular}{*{11}{c}c}
{ }     & { }   & { }   &$F_{BBB}$& { }     & { }   & { }  \\
{ }     & { }   &$F_{BBG}$& { } &$F_{BBR}$& { }     & { }  \\
{ }     &$F_{BGG}$& { } &  $F_{BGR}$& { }   &$F_{BRR}$& { }  \\
 $F_{GGG}$& { } &$F_{GGR}$& { } & $F_{GRR}$& { } &   $F_{RRR}$  \\
\end{tabular}
\end{equation}


This notation can be extended to other arities. For a signature with arity two, we also use a symmetric $q \times q$ matrix to represent it.
\begin{eqnarray*}
 \mathbf{F} &=& [F_{BB}; F_{BG}, F_{BR}; F_{GG}, F_{GR}, F_{RR}]\\
   &=& \begin{bmatrix} F_{BB} & F_{BG} & F_{BR} \\
   F_{BG} & F_{GG} & F_{GR} \\
   F_{BR} & F_{GR} & F_{RR}
   \end{bmatrix}.
\end{eqnarray*}

For a binary signature, the rank of the signature is
the rank of its $q \times q$ matrix.

A unary function can be represented as $[F_B; F_G, F_R]$ in symmetric notation, or simply $(F_B, F_G, F_R)$ in full version.

We use $\mathbf{F}^{i= A}$, where $i\in [r]$ and $A \in \{B,G,R\}$, to denote a signature of arity $r-1$ by fixing the $i$-th input of $\mathbf{F}$ to $A$.
For example for the $\mathbf{F}$ in (\ref{domain-3-sig})
\begin{equation}\label{domain-3-sig-1-is-set-to-B}
\mathbf{F}^{1 = B} =   \begin{bmatrix} F_{BBB} & F_{BBG} & F_{BBR} \\
   F_{BBG} & F_{BGG} & F_{BGR} \\
   F_{BBR} & F_{BGR} & F_{BRR}
   \end{bmatrix}.
\end{equation}

Sometimes, we also restrict the
$i$-th input of  $\mathbf{F}$ to $S$, a subset of $\{B,G,R\}$, and we use
$\mathbf{F}^{i
\rightarrow S}$ (for example $\mathbf{F}^{2 \rightarrow \{B,R\}}$) to denote it.
We use $\mathbf{F}^{* \rightarrow {S}}$ to denote the case
when we restrict  all inputs of  $\mathbf{F}$ to $S$. For example
\[ \mathbf{F}^{* \rightarrow \{G,R\}}= [ F_{GGG}, F_{GGR}, F_{GRR}, F_{RRR}].\]

The above notation can be combined, for example
\[\mathbf{F}^{1 = B; 2,3  \rightarrow \{G,R\}}= [ F_{BGG}, F_{BGR}, F_{BRR}]. \]

We also use $F_{a,b,c}$, ($a,b,c \in \mathbf{N}, a+b+c=r$) to denote the
value of $\mathbf{F}$ when the numbers of $B$'s, $G$'s and $R$'s
among the inputs are respectively $a$, $b$ and $c$.
For example, $F_{1,2,0} =F_{BGG}$.

\begin{definition} \label{def-domain-separated}
A symmetric function $F$ of arity $r \geq 2$, gives a $r$-uniform hyper graph $G$ whose vertex set is the domain of variables. We say two disjoint subsets of domain are separated, if they are contained in different connected components of $G$.
\end{definition}

For example, if a ternary function has the form
\begin{center}
\begin{tabular}{*{11}{c}c}
{ }     & { }   & { }   &$F_{BBB}$& { }     & { }   & { }  \\
{ }     & { }   &$0$& { } &$0$& { }     & { }  \\
{ }     &$0$& { } &  $0$& { }   &$0$& { }  \\
 $F_{GGG}$& { } &$F_{GGR}$& { } & $F_{GRR}$& { } &   $F_{RRR}$  \\
\end{tabular}
\end{center}
We say that $B$ is separated from $\{G,R\}$.

\subsection{A Calculus with Symmetric Signatures}\label{section:calculus}


In order to follow the proofs in
this paper, it would be helpful to familiarize oneself
with a certain calculus that lets us reason about these symmetric signatures
on domain size 3.  We will mainly illustrate it with signatures of arity 2 or 3.
It is easy to generalize it to any higher arities.

For any symmetric signature
$\mathbf{F}$ of arity 2 on domain $\{B,G,R\}$,
we make the following identification of
the notation
\begin{center}
\begin{tabular}{*{11}{c}c}
{ }     & { }   & { }   &$F_{BB}$& { }     & { }   & { }  \\
{ }     & { }   &$F_{BG}$& { } &$F_{BR}$& { }     & { }  \\
{ }     &$F_{GG}$& { } &  $F_{GR}$& { }   &$F_{RR}$& { }  \\
\end{tabular}
\end{center}
with its matrix form
\begin{eqnarray*}
   \begin{bmatrix} F_{BB} & F_{BG} & F_{BR} \\
   F_{BG} & F_{GG} & F_{GR} \\
   F_{BR} & F_{GR} & F_{RR}
   \end{bmatrix}.
\end{eqnarray*}
We note that the three corners in counterclock-wise order $B, G, R$
are listed on the main diagonal in  the matrix. Then the off-diagonal
entries are filled by the corresponding color pairs, e.g., the entry
$F_{BG}$ between $B$ and $G$ are filled at the $(B,G)$ and $(G,B)$ entry
of the matrix.

Let $\mathbf{F}$ be a ternary symmetric signature,
and let $\mathbf{u} = (\alpha, \beta, \gamma)$ be a unary signature, both on
domain $\{B,G,R\}$, we can form a binary symmetric signature
by connecting one input of  $\mathbf{F}$ with $\mathbf{u}$.
Since $\mathbf{F}$ is symmetric, connecting to any one of the
input wires defines the same symmetric signature on the
other input wires. We denote this signature by
$\langle \mathbf{u}, \mathbf{F} \rangle$.
By symmetry, for $\mathbf{F}$ of arity at least 2,
 $\langle \mathbf{v}, \langle \mathbf{u}, \mathbf{F} \rangle \rangle
= \langle \mathbf{u}, \langle \mathbf{v}, \mathbf{F} \rangle \rangle$.

Suppose $\mathbf{F}$ is  given in (\ref{domain-3-sig}).
Then $\langle \mathbf{u}, \mathbf{F} \rangle$ is the following
\begin{center}
\begin{tabular}{*{11}{c}c}
{ }     & { }   & { }   &$F'_{BB}$& { }     & { }   & { }  \\
{ }     & { }   &$F'_{BG}$& { } &$F'_{BR}$& { }     & { }  \\
{ }     &$F'_{GG}$& { } &  $F'_{GR}$& { }   &$F'_{RR}$& { }  \\
\end{tabular}
\end{center}
where each entry $F'_{XY}$ is obtained by a linear combination
$\alpha F_{XYB} +  \beta F_{XYG} + \gamma F_{XYR}$;
i.e., we start at any entry on the first three rows in
the triangular table  for $\mathbf{F}$, and then form a
linear combination with coefficients $\alpha, \beta, \gamma$ in
 a counterclock-wise order involving the three entries forming
a {\it small triangle}. E.g., start with entry
$F_{BBG}$, we get $F'_{BG} = \alpha F_{BBG} + \beta F_{BGG} + \gamma F_{BGR}$.

Suppose $\mathbf{F}$ is a symmetric  ternary signature,
and $\mathbf{u} = ( 1, i, 0)$. Then we see immediately
that $\langle \mathbf{u}, \mathbf{F} \rangle = \mathbf{0}$
(the zero binary function) iff  $\mathbf{F}$ has the following form
\begin{equation}\label{annihilated-by-1-i}
\begin{tabular}{*{11}{c}c}
{ }     & { }   & { }   &$x$& { }     & { }   & { }  \\
{ }     & { }   &$x i$& { } &$y$& { }     & { }  \\
{ }     &$- x$& { } &  $y i $& { }   &$z$& { }  \\
 $- x i $& { } &$- y$& { } & $z i$& { } &   $w$  \\
\end{tabular}
\end{equation}




Fixing some variables to $R$ and restrict others to $\{B,G\}$, we get
$\mathbf{F}^{1,2,3  \rightarrow \{B,G\}}=[x,xi,-x,-xi]$, $\mathbf{F}^{1 = R; 2,3  \rightarrow \{B,G\}}=[y,yi,-y]$ and
$\mathbf{F}^{1 = R, 2=R; 3  \rightarrow \{B,G\}}=[z,zi]$. They all become the zero function, after connecting with the unary function $(1,i)$.

Suppose $\mathbf{F}$  has the property that when we fix the number of $R$'s,
the restricted signatures on domain $\{B, G\}$ all satisfy
a single linear recurrence,
then, viewed in terms of those {\it small triangles},
it follows that the $\{B, G\}$-restricted signatures
of $\langle \mathbf{u}, \mathbf{F} \rangle$ also
satisfy the same linear recurrence.

Let us suppose we are given a
symmetric  ternary signature  $\mathbf{F}$
with $F_{GGR} = F_{GRR}=0$, thus
\begin{center}
\begin{tabular}{*{11}{c}c}
{ }     & { }   & { }   &$F_{BBB}$& { }     & { }   & { }  \\
{ }     & { }   &$F_{BBG}$& { } &$F_{BBR}$& { }     & { }  \\
{ }     &$F_{BGG}$& { } &  $F_{BGR}$& { }   &$F_{BRR}$& { }  \\
 $F_{GGG}$& { } &$0$& { } & $0$& { } &   $F_{RRR}$  \\
\end{tabular}
\end{center}
By connecting a unary function $(1,t,0)$ to $\mathbf{F}$
we will obtain a binary function whose triangular table
has the third row being $[F_{BGG}+ t F_{GGG}, F_{BGR}, F_{BRR}]$.
If we further
connect both dangling edges of this binary function with
$(=_{G,R})
   =
\begin{bmatrix}
   0 & 0 & 0 \\
   0 & 1 & 0 \\
   0 & 0 & 1
   \end{bmatrix}
 $,
we get a symmetric binary signature whose restriction on $\{G, R\}$
is $[F_{BGG}+ t F_{GGG}, F_{BGR}, F_{BRR}]$, and zero elsewhere.

Now suppose further that  the ternary function $\mathbf{F}$
satisfies $F_{BGR}= F_{GGR} = F_{GRR}=0$, i.e., it has the form
\begin{center}
\begin{tabular}{*{11}{c}c}
{ }     & { }   & { }   &$g$& { }     & { }   & { }  \\
{ }     & { }   &$y$& { } &$w$& { }     & { }  \\
{ }     &$x$& { } &  $0$& { }   &$z$& { }  \\
 $a$& { } &$0$& { } & $0$& { } &   $b$  \\
\end{tabular}
\end{center}
Let us consider the gadget as depicted in the following Figure
to construct another binary function, where both vertices of
degree 3 are given the function $\mathbf{F}$.
\begin{figure}[hbtp]
   \begin{center}
                \includegraphics[width=3 in]{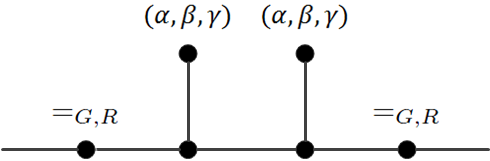}
\caption{A binary gadget.}        \label{calculus1}
   \end{center}
\end{figure}
We calculate its signature $S$ as follows:
It will be the matrix product of 4 matrices.
The first matrix is
$(=_{G,R})
   =
\begin{bmatrix}
   0 & 0 & 0 \\
   0 & 1 & 0 \\
   0 & 0 & 1
   \end{bmatrix} $.
The second is the  matrix  form $M$ of
$\langle \mathbf{u}, \mathbf{F} \rangle$,
where $\mathbf{u} = (\alpha, \beta, \gamma)$,
and we get
\begin{center}
\begin{tabular}{*{11}{c}c}
{ }     & { }   & { }   &$\alpha g + \beta y + \gamma w$& { }     & { }   & { }  \\
{ }     & { }   &$\alpha  y + \beta  x$& { } &$\alpha  w + \gamma z$& { }     & { }  \\
{ }     &$\alpha x + \beta  a$& { } &  $0$& { }   &$\alpha z + \gamma  b$& { }  \\
\end{tabular}
\end{center}
Thus the matrix form is
\begin{eqnarray*}
M=    \begin{bmatrix} \alpha g + \beta y + \gamma w & \alpha  y + \beta  x &
\alpha  w + \gamma z \\
   \alpha  y + \beta  x & \alpha x + \beta  a & 0 \\
   \alpha  w + \gamma z & 0 & \alpha z + \gamma  b
   \end{bmatrix}.
\end{eqnarray*}
The third matrix will be $M$ as well, and we note that
$M$ is symmetric, $M^{\tt T} = M$.
The fourth will be $=_{G,R}$ again, which is also symmetric.
We can calculate the $2 \times 2$ matrix for the signature
$S^{* \rightarrow \{G, R\}}$
as a function on the restricted domain $\{G, R\}$ to be
$\begin{bmatrix}
   0 & 1 & 0 \\
   0 & 0 & 1
   \end{bmatrix} M M^{\tt T}
\begin{bmatrix}
    0 & 0 \\
    1 & 0 \\
    0 & 1
   \end{bmatrix} $.
Thus the  signature  $S^{* \rightarrow \{G, R\}}$
can be computed  as follows, picking only the second and third rows of $M$:
\begin{eqnarray*}
 \begin{bmatrix}
   \alpha  y + \beta  x & \alpha x + \beta  a & 0 \\
   \alpha  w + \gamma z & 0 & \alpha z + \gamma  b
   \end{bmatrix}
 \begin{bmatrix}
   \alpha  y + \beta  x & \alpha  w + \gamma z \\
   \alpha x + \beta  a & 0 \\
   0 &  \alpha z + \gamma  b
   \end{bmatrix}.
\end{eqnarray*}
Written in the symmetric signature notation on domain size 2
we have
\begin{equation}\label{sig-2-weed-gadget}
[(\alpha y + \beta x)^2 + (\alpha x +\beta a)^2,
(\alpha y + \beta x)(\alpha w + \gamma z),
(\alpha w + \gamma z)^2 + (\alpha z + \gamma b)^2 ].
\end{equation}

\subsection{Symmetry and Decomposition}

We say a function is decomposable, if it has arity at least $2$, and is a product of two functions (of arity at least $1$) applied to its two disjoint variable subsets respectively.

\begin{fact} \label{fact-cut1}
If a symmetric function  $\mathbf{F}=\mathbf{A}(x_1)\mathbf{B}(x_2, \cdots, x_r)$,
then there  is a constant $c$, such that
$\mathbf{F} = c \prod_{i=1}^r \mathbf{A}(x_i)$.

\end{fact}

\begin{proof}
If $\mathbf{A} \equiv 0$ then it is trivial. We can assume $\mathbf{A}(a) \neq 0$.

If $r=2$, then  by $\mathbf{F}=\mathbf{A}\mathbf{B}$, $\mathbf{A}(a)\mathbf{B}(x_2)=\mathbf{A}(x_2)\mathbf{B}(a)$. Hence, $\mathbf{A}(a)\mathbf{F}=\mathbf{A}(a)\mathbf{A}(x_1)\mathbf{B}(x_2)=\mathbf{B}(a)\mathbf{A}(x_1)\mathbf{A}(x_2)$, and we can set $c = \mathbf{A}(a)^{-1} \mathbf{B}(a)$.

If $r>2$, restricting to $x_1 =a$, we see that
$\mathbf{B}$ is symmetric.
From $\mathbf{F}=\mathbf{A}(x_1)\mathbf{B}(x_2, \ldots, x_r)$,
we get $\mathbf{A}(a) \mathbf{B}(x_2,\ldots,x_r)=\mathbf{A}(x_2)\mathbf{B}(a,x_3, \ldots, x_r)$. This means  that
$\mathbf{B}(x_2,\ldots,x_r)$
 is a product of $\mathbf{A}(x_2)$ and a function  on $(x_3, \ldots,x_r)$.
 By induction hypothesis, the conclusion holds.
\end{proof}

For the general case, the idea is similar. If
a symmetric $F$ is decomposed into  $\mathbf{A}$ and $\mathbf{B}$, we utilize this to cut $\mathbf{A}$ and $\mathbf{B}$ into smaller pieces.

\begin{fact} \label{fact-cut2}
If a symmetric function  $\mathbf{F}=\mathbf{A}(x_1 \cdots, x_r )\mathbf{B}(x_{r+1}, \cdots, x_{r+s})$,
that is, it is decomposable,  then  for some constant $c$,
 and unary function $C$, $\mathbf{F}= c \prod_{i=1}^{r+s} \mathbf{C}(x_i)$.
\end{fact}

\begin{proof}
For convenience we write $(y_1, \cdots, y_s) =
(x_{r+1}, \cdots, x_{r+s})$.
If $r =1$ or $s=1$, we are done by Fact \ref{fact-cut1}. Let $r >1$ and $s >1$.
If $\mathbf{A} \equiv 0$ it is done. We can assume $\mathbf{A}(a_1,\cdots, a_r) \neq 0$.

By symmetry, we get
$\mathbf{A}(a_1,\cdots, a_r) \mathbf{B}(y_1, \cdots, y_s)
= \mathbf{A}(y_1,a_2,\cdots, a_r) \mathbf{B}(a_1, y_2 \cdots, y_s)$.
Thus $\mathbf{B}$ satisfies the assumption of Fact \ref{fact-cut1}.
So $\mathbf{B}(y_1, \cdots, y_s)$ has the form $c' \prod_{j=1}^{s} \mathbf{C}(y_j)$.
Then
$\mathbf{F}$ has the form $\mathbf{C}(y_1) (c' \mathbf{A}(x_1 \cdots, x_r ) \prod_{j=2}^{s} \mathbf{C}(y_j))$.
By Fact \ref{fact-cut1} again, we get the conclusion.

\end{proof}

By Fact \ref{fact-cut2}, if a symmetric function is decomposable, then it
is a tensor power of a unary function.
It is in $\langle \mathcal{U} \rangle - \mathcal{U}$, and degenerate.

We have seen symmetry can help to decompose a decomposable function into smaller parts. Next fact shows that some ``partial symmetry" property also helps.

\begin{fact} \label{fact-half-sym-cut1}
Suppose $\mathbf{F}$ satisfies $\mathbf{F}(x_1,x_2,y_1,y_2)=\mathbf{F}(x_2,x_1,y_1,y_2)=\mathbf{F}(x_1,x_2,y_2,y_1)$.
If $\mathbf{F}= \mathbf{A}(x_1) \mathbf{B}(x_2,y_1,y_2)$, then there are binary functions $C,D$,
such that $\mathbf{F}=\mathbf{C}(x_1,x_2)\mathbf{D}(y_1,y_2)$.
\end{fact}

The proof is similar to Facts \ref{fact-cut1} and \ref{fact-cut2}.

\begin{fact} \label{fact-half-sym-cut2}
Suppose $\mathbf{F}$ satisfies $\mathbf{F}(x_1,x_2,y_1,y_2)=\mathbf{F}(x_2,x_1,y_1,y_2)=\mathbf{F}(x_1,x_2,y_2,y_1)$.
If $\mathbf{F}$ is decomposed into two binary functions, then there are binary functions $\mathbf{A},\mathbf{B}$, either $\mathbf{F}=\mathbf{A}(x_1,x_2)\mathbf{B}(y_1,y_2)$  or $\mathbf{F}=\mathbf{A}(x_1,y_1)\mathbf{B}(x_2,y_2)$.
\end{fact}

The proof is straightforward. If $\mathbf{F}$ is decomposed into two binary functions of other forms, just utilize the  ``partial symmetry" property to rotate it into one of the two forms.

In our hardness proofs, we will
 need to use some gadget with this  ``partial symmetry" property
to realize a function $\mathbf{F}$ of arity 4 that can not be decomposed into two binary functions ($\mathbf{F} \not \in \langle \mathcal{T} \rangle$). By Fact \ref{fact-half-sym-cut2}, we only need to show
that  it cannot be decomposed into these two forms. We will
call this  \methps.

\subsection{Known Dichotomy Theorems}

We say a function set $\mathcal{F}$ is closed under tensor product,
if for any  $\mathbf{A}, \mathbf{B} \in \mathcal{F}$, $\mathbf{A }\otimes \mathbf{B} \in
\mathcal{F}$. Tensor closure $\langle \mathcal{F} \rangle$ of a set
$\mathcal{F}$ is the minimum set containing $\mathcal{F}$, closed
under tensor product.
This closure clearly exists, being  the set of all functions
obtained by performing a finite sequence of tensor products from $\mathcal{F}$.

We use $\mathcal{U}$ to denote the set of all unary functions.
$\mathcal{E}$ is the set  of all functions $F$ such that $F$
is zero except on two inputs $(a_1,\ldots,a_n)$ and
$(\bar{a}_1,\ldots, \bar{a}_n)=(1-a_1,\ldots,1-a_n)$.
In other words, $F \in \mathcal{E}$ iff its support is
contained in a pair of complementary points.
We think of $\mathcal{E}$ as a generalized form of {\sc Equality} function.
Equivalently, these are obtained by connecting some subset of
variables of {\sc Equality} with binary {\sc Disequality} $\not =_2$.
We use $\mathcal{M}$ to denote
 the set of all functions $F$ such that $F$ is zero
except on $n+1$ inputs whose Hamming weight is at most $1$, where $n$ is the
arity of $F$. The name $\mathcal{M}$ is given for {\it matching}.
Finally
$\mathcal{T}$ is the set of all functions of arity at most $2$.
Note that $\mathcal{U}$ is a subset of $\mathcal{E}$, $\mathcal{M}$ and
$\mathcal{T}$.


Suppose $\mathcal{F}$ is a function set and $M$ is a $2 \times 2$
matrix. We use $M \mathcal{F}$  to denote the set
$\{M^{\otimes r_F}F | F \in
\mathcal{F}, r_F={\rm arity}(F)\}$, the set consisting  of all
functions in $\mathcal{F}$ transformed by a matrix $M$.
$Z_1=\left (
\begin{array}{cc} 1 & 1 \\ i & -i \end{array} \right )$
and $Z_2=\left (
\begin{array}{cc} 1 & 1 \\ -i & i \end{array} \right )$.
Note that $Z_1 \mathcal{E} = Z_2 \mathcal{E}$.

\begin{theorem} \label{thm:dich-asym-Boolean} \cite{SODA11}
Let  $\mathcal{F}$ be  any set
 of complex valued  functions in Boolean variables.
The problem Holant$^*(\mathcal{F})$ is
polynomial time computable, if
\begin{enumerate}
\item
$\mathcal{F}  \subseteq \langle \mathcal{T} \rangle$, or
\item
for some orthogonal matrix $H$,  $\mathcal{F} \subseteq \langle H \mathcal{E} \rangle$, or
\item
$\mathcal{F} \subseteq \langle Z_1 \mathcal{E} \rangle$, or
\item
for some $Z \in \{Z_1, Z_2\}$,  $\mathcal{F} \subseteq \langle Z \mathcal{M} \rangle$.
\end{enumerate}
In all other cases,
Holant$^*(\mathcal{F})$ is \#P-hard.
\end{theorem}


This theorem is a generalization to not necessarily symmetric
function sets from the following
theorem which only applies to symmetric function sets.
It is also very conceptual;
however the following theorem is very easy to apply.

\begin{theorem}\label{thm:dich-sym-Boolean} \cite{STOC09}
 Let ${\cal F}$ be any set of non-degenerate, symmetric, complex-valued
signatures in Boolean variables.
If ${\cal F}$ is of one of the following types, then
${\rm Holant}^*({\cal F})$ is in P, otherwise it is \#P-hard.
 \begin{enumerate}  \item Any signature in ${\cal F}$ is of arity at most 2;
   \item There exist two constants $a$ and $b$
($b\neq \pm 2 i a$, depending only on ${\cal F}$), such that
  for all signatures $[f_0,f_1,\ldots,f_n]$
in ${\cal F}$ one of the two conditions is
  satisfied:  (1) for every $k=0,1,\ldots,n-2$, we have $a f_k + b f_{k+1} -a
  f_{k+2}=0$; (2) $n=2$ and the signature $[f_0, f_1, f_2]$ is of the form
  $[2a \lambda, b \lambda, -2a \lambda] $.
  \item For every signature $[f_0,f_1,\ldots,f_n]\in {\cal F}$ one of the two conditions is
  satisfied: (1) For every $k=0,1,\ldots,n-2$, we have $ f_k +
  f_{k+2}=0$; (2) $n=2$ and the signature $[f_0, f_1, f_2]$ is of the form
  $[ \lambda, 0,  \lambda] $.
  \item There exists $\alpha \in \{2 i, -2 i\}$,
such that for any signature $f \in \mathcal{F}$
 of arity $n$, for $0 \le k \le n-2$, we have $f_{k+2} = \alpha f_
{k+1} + f_k$. \label{case:holant_star:exceptional}
 \end{enumerate}\end{theorem}

In Holant$^*$ problems, unary functions are  freely available.
There is no difference between Holant$^*(\mathcal{F}- \langle \mathcal{U} \rangle)$ and Holant$^*(\mathcal{F} \cup \langle \mathcal{U} \rangle)$. Theorem \ref{thm:dich-sym-Boolean}
is stated for $\mathcal{F}- \langle \mathcal{U} \rangle$.

We give the correspondence between Theorem \ref{thm:dich-asym-Boolean} and  \ref{thm:dich-sym-Boolean}.
Consider the symmetric subset of the
first tractable class $\langle \mathcal{T} \rangle$ in Theorem \ref{thm:dich-asym-Boolean}.  If a symmetric function in $\langle \mathcal{T} \rangle$ has arity larger than 2, it is decomposable and degenerate.

The function sets in Theorem \ref{thm:dich-sym-Boolean}
in forms 2 to 4 can be described by,
\begin{equation}\label{Fibonacci-P-ab-and-P}
\mathcal{P}_{a,b}=\{[f_0,f_1,\cdots, f_n ]| n \in \mathbb{N}, a x_k + b x_{k+1} -a
  x_{k+2}=0\} \cup
\{\lambda [2a,b,-2a] \mid \lambda \in \mathbb{C}\},
\end{equation}
\begin{equation*}
\mathcal{P}=\{[f_0,f_1,\cdots, f_n]| n \in \mathbb{N}, f_k+f_{k+2}=0\}\cup \{\lambda [1,0,1]
\mid \lambda  \in \mathbb{C}\}.
\end{equation*}
Form 2 and 4 are described by $\mathcal{P}_{a,b}$ with
$(a, b)$ not both zero, with Form 4 corresponding to
$\mathcal{P}_{1, \pm 2i}$. Form 3 is described by $\mathcal{P}$.
Note that for $\alpha = \pm 2 i$, a binary $f$ with
$f_2 = \alpha f_1 + f_0$ is degenerate.
In $\mathcal{P}_{a,b}$, we always require $(a,b) \neq (0,0)$,
and $(a, b)$ is equivalent to any non-zero multiple of it.
When we say all $\mathcal{P}_{a,b}$,
we let $(a,b)$ range over all $\mathbb{C}^2-\{(0,0)\}$ (equivalently
 the projective
line $\mathbb{P}_{\mathbb{C}}^1$).

By Fact \ref{fact-cut2} a non-degenerate symmetric function  must not be decomposable.
It is in a set of tractable case $j$ in Theorem \ref{thm:dich-asym-Boolean}, iff it is in the corresponding set of tractable case $j$ in Theorem \ref{thm:dich-sym-Boolean}.
For example, suppose $H= \left (  \begin{array}{cc} u & v \\ s & t \end{array} \right )$ is an orthogonal matrix. $ H \mathcal{E}$  corresponds to the set $\mathcal{P}_{a,b}$, where the corresponding relation is that 3 vectors $(u^2, us, s^2), (v^2, vt, t^2), (a, b, -a)$ form an orthogonal independent vector set.
One $\mathcal{P}_{a,b}$ corresponds to two $(H\mathcal{E} )^{S}$, given by $H$ and $H \tau$, where $\tau = \left (  \begin{array}{cc} 0 & 1 \\ 1 & 0 \end{array} \right )$ exchanges the two columns of $H$.

\subsection{Polynomial Argument }


\begin{fact}\label{fact-nonzeropoly}
The product of two non-zero polynomials is a non-zero polynomial.
\end{fact}

It is a simple fact that a polynomial ring (in any number of
indeterminants and over any field)
 is an integral domain, and thus has
no zero divisor. The way we will use this fact is as follows.
When we design some gadget, usually there are some unary functions
 $(\alpha, \beta, \gamma)$ in this gadget, which work as parameters
in order  for the signature realized by the gadget
to satisfy some conditions (for example, it should have full rank).
Usually a condition can be described by a polynomial
$P(\alpha, \beta, \gamma)$ in these parameters,
such that when  $P(\alpha_0, \beta_0, \gamma_0) \neq 0$,
the signature realized by the gadget using the unary
function $(\alpha_0, \beta_0, \gamma_0)$ satisfies this condition.

By Fact \ref{fact-nonzeropoly}, when there are several such conditions to satisfy, we only need to show each polynomial $P_i$ is not zero, usually by finding
 some point $(\alpha_i, \beta_i, \gamma_i)$ for each $P_i$.
This guarantees the existence of some common parameter value
$(\alpha^*, \beta^*, \gamma^*)$ such that
$\prod_i P_i(\alpha^*, \beta^*, \gamma^*) \not =0$.
The value $(\alpha^*, \beta^*, \gamma^*)$ is
 implicit and not important; it has no direct connection to
the choice of each  $(\alpha_i, \beta_i, \gamma_i)$. This method is already used in \cite{SODA11}. In proof, we quote it as \methpoly.

\section{Statement of the Dichotomy Theorem}

\begin{theorem}\label{thm:ternary}
Let $\mathbf{F}$ be a symmetric ternary function over domain $\{B,G,R\}$.
Then Holant$^*(\mathbf{F})$ is \#P-hard unless $\mathbf{F}$ is of one of the following three forms,
in which case the problem is in polynomial time.
\begin{enumerate}
  \item There exist three vectors
 $\mbox{\boldmath $\alpha$, $\beta$, and $\gamma$}$
of dimension 3 such that
  they are mutually orthogonal to each other,
i.e. $\langle \mbox{\boldmath $\alpha$, $\beta$} \rangle=0$,
     $\langle \mbox{\boldmath $\alpha$, $\gamma$} \rangle=0$
 and $\langle \mbox{\boldmath $\beta$,  $\gamma$} \rangle=0$,  and
  \[ \mathbf{F}=\mbox{\boldmath $\alpha$}^{\otimes 3} +
                \mbox{\boldmath $\beta$}^{\otimes 3} +
                \mbox{\boldmath $\gamma$}^{\otimes 3};\]
  \item There exist three vectors
$\mbox{\boldmath $\alpha$, $\beta_1$, and $\beta_2$}$
of dimension 3 such that
  $\langle \mbox{\boldmath $\alpha$, $\beta_1$} \rangle=0$,
  $\langle \mbox{\boldmath $\alpha$, $\beta_2$} \rangle=0$,
  $\langle \mbox{\boldmath $\beta_1$, $\beta_1$} \rangle=0$,
  $\langle \mbox{\boldmath $\beta_2$, $\beta_2$} \rangle=0$
and
  \[ \mathbf{F}=\mbox{\boldmath $\alpha$}^{\otimes 3} +
                 \mbox{\boldmath $\beta_1$}^{\otimes 3} +
                \mbox{\boldmath $\beta_2$}^{\otimes 3};\]
    \item There exist two vectors $\mbox{\boldmath $\beta$ and $\gamma$}$
of dimension 3 and a function
${\mathbf{F}}_{\beta}$
 of arity three,
such that $\mbox{\boldmath $\beta$} \neq \mathbf{0}$,
$\langle \mbox{\boldmath $\beta$, $\beta$} \rangle=0$,
$\langle \mathbf{F}_{\beta}, \mbox{\boldmath $\beta$}
 \rangle = \mathbf{0}$ and
  \[ \mathbf{F}= \mathbf{F}_{\bf \beta}
+ \mbox{\boldmath $\beta$}^{\otimes 2} \otimes \mbox{\boldmath $\gamma$}
+ \mbox{\boldmath $\beta$} \otimes \mbox{\boldmath $\gamma$} \otimes
\mbox{\boldmath $\beta$}
+ \mbox{\boldmath $\gamma$}  \otimes  \mbox{\boldmath $\beta$}^{\otimes 2}.\]
\end{enumerate}
\end{theorem}

\noindent {\bf Remarks:} 1. In the forms above,
the vectors $\mbox{\boldmath $\alpha$, $\beta$, $\gamma$,
$\beta_1$, $\beta_2$}$ can be the zero vector (except
$\mbox{\boldmath $\beta$}$ in form {\it 3}.)\\
2. In form {\it 3}, $\mathbf{F}$ is
the sum of $\mathbf{F}_{\beta}$ with
($1/2$ of) the symmetrization of
$\mbox{\boldmath $\beta$}^{\otimes 2} \otimes \mbox{\boldmath $\gamma$}$.
The constant factor $1/2$ doesn't matter, and
  can be absorbed in $\mbox{\boldmath $\gamma$}$.\\
3. Let $T$ be an orthogonal $3 \times 3$ matrix, then
$\mathbf{F}$ is of one of the three forms above
iff $T^{\otimes 3} \mathbf{F}$ is.

\subsection{Canonical Forms for Tractable Cases}

Theorem~\ref{thm:ternary} gives a complete list of tractable cases
for Holant$^*(\mathbf{F})$. Before we give the proof of
tractability we need to discuss these tractable forms in some detail,
and give various canonical forms
of these tractable cases, under an orthogonal transformation $T$.
We note that for an orthogonal $T$, the arity 2 {\sc Equality} gate $(=_2)$
(on any domain  size) is invariant, the unary signatures are
transformed to unary signatures, and the formal description of
the three forms of
$\mathbf{F}$ is also invariant,
i.e., $\mathbf{F}$ is of one of the  three forms iff
$T^{\otimes 3} \mathbf{F}$ is.

In terms of the canonical forms, Theorem~\ref{thm:ternary}
can be restated as follows. We will write $T \mathbf{F}$
for $T^{\otimes 3} \mathbf{F}$ for simplicity.

\begin{theorem}\label{thm:normal-forms}
Let $\mathbf{F}$ be a symmetric ternary function over domain $\{B,G,R\}$.
Then Holant$^*(\mathbf{F})$ is \#P-hard unless under an orthogonal
transformation $T$, the function
$T \mathbf{F}$ is of one of the following forms,
in which case the problem is in P.
\begin{enumerate}
  \item For some $a, b, c \in \mathbb{C}$,
\[ T \mathbf{F} = a \mbox{\boldmath $e_1$}^{\otimes 3} + b \mbox{\boldmath $e_2$}^{\otimes 3} + c \mbox{\boldmath $e_3$}^{\otimes 3}.\]
  \item For some $c \not =0$ and $\lambda \in  \mathbb{C}$,
\[ c  T \mathbf{F} = \mbox{\boldmath $\beta_0$}^{\otimes 3} +
\overline{\mbox{\boldmath $\beta_0$}}^{\otimes 3} +
 \lambda  \mbox{\boldmath $e_3$}^{\otimes 3},\]
where $\mbox{\boldmath $\beta_0$} = \frac{1}{\sqrt{2}}(1, i, 0)^{\tt T}$,
and $\overline{\mbox{\boldmath $\beta_0$}}$ is its
conjugate $\frac{1}{\sqrt{2}}(1, -i, 0)^{\tt T}$.
\item For $\epsilon \in \{0, 1\}$,
\[T \mathbf{F} = \mathbf{F}_0 + \epsilon \rm{Sym}(
\mbox{\boldmath $\beta_0$} \otimes \mbox{\boldmath $\beta_0$}  \otimes
\overline{\mbox{\boldmath $\beta_0$}}), \]
where $\mathbf{F}_0$
satisfies the annihilation condition
\(\langle \mathbf{F}_0,
\mbox{\boldmath $\beta_0$} \rangle = \mathbf{0}.\)
\end{enumerate}\end{theorem}


%
%

We start by defining the complex version of rotations.
For any $z \in \mathbb{C}$, let $c = \cos z = \frac{e^{iz} + e^{-iz}}{2}$
and $s = \sin z = \frac{e^{iz} - e^{-iz}}{2i}$, and
$T_2 = \begin{bmatrix} c & s \\
			-s & c
\end{bmatrix}$. Then $c^2 + s^2 = 1$ and $T_2$ is a $2 \times 2$
orthogonal matrix.
If $\begin{bmatrix} a \\ b \end{bmatrix} \in \mathbb{C}^2$ is not isotropic,
then $T_2 \begin{bmatrix} a \\ b \end{bmatrix}
= \begin{bmatrix} ca + sb \\ -sa + cb \end{bmatrix}$
is also not isotropic $(ca + sb)^2 + (-sa + cb)^2 = a^2 + b^2 \not = 0$.
Let $\eta = \cot z = i \frac{e^{iz} + e^{-iz}}{e^{iz} - e^{-iz}}
= i \frac{e^{2iz} +1 }{e^{2iz} -1}$,
we want a suitable $z \in \mathbb{C}$, such that $-sa + cb =0$.
The M\"{o}bius map $\xi \mapsto i \frac{\xi+1}{\xi -1}$
is a one-to-one onto
 map on the extended Riemann  complex plane $\mathbb{C} \cup
\{\infty\}$.  As $z \mapsto e^{2z}$ maps $\mathbb{C}$ onto
$\mathbb{C} - \{0\}$, the mapping $z \mapsto \eta = \cot z$
from $\mathbb{C}$
has image
$\mathbb{C} \cup \{\infty\} - \{i, -i\}$.
This proves that we can find an orthogonal $T_2$
such that $T_2 \begin{bmatrix} a \\ b \end{bmatrix}
= \begin{bmatrix} a' \\ 0 \end{bmatrix}$, where $a'^2 = a^2 + b^2$,
for any non-isotropic $\begin{bmatrix} a \\ b \end{bmatrix}$.

Suppose $v = (a_1, a_2, \ldots, a_d)^{\tt T}
\in \mathbb{C}^d$ is non-isotropic, $d \ge 2$.
Suppose $d'$ is the number of non-zero entries $a_i$. Then $d' \ge 1$.
By a permutation matrix (which is orthogonal) we may assume
they are $a_1, \ldots, a_{d'}$.
Suppose $d' \ge 2$.
There exist $1 \le i < j \le d'$, such that
$(a_i, a_j)^{\tt T}$ is non-isotropic.
Otherwise, summing $a_i^2 + a_j^2$ over all distinct pairs $(i,j)$ among the
non-zero entries $1 \le i < j \le d'$
we get $(d'-1) \sum_{i=1}^{d'} a_i^2 = 0$ and $v$ is
isotropic.
Hence,
we can use a permutation matrix (which is orthogonal)
to map $v$ such that $a_1^2 + a_2^2 \not = 0$.
By a rotation described above, we may use an orthogonal
matrix of the form ${\rm diag}( T_2, I_{d-2} )$
to transform $v$, such that it has one fewer non-zero entries
but with the same value $\langle  v, v \rangle
= \sum_{i=1}^d a_i^2$. By induction, we have proved

\begin{lemma}\label{non-isotropic-orth-transf}
For any non-isotropic $v = (a_1, a_2, \ldots, a_d)^{\tt T}
\in \mathbb{C}^d$, $d \ge 1$, there exists an orthogonal matrix $T$
such that $T v = (\pm \sqrt{\langle  v, v \rangle}, 0, \ldots, 0)^{\tt T}$.
(Both $\pm$ are feasible.)
\end{lemma}


Now suppose
$v \in \mathbb{C}^d$ is a non-zero isotropic vector.
Certainly $d \ge 2$.
We want to show that there is an orthogonal matrix $T$
transforming $v$ to $\mbox{\boldmath $\beta_0$} =
\frac{1}{\sqrt{2}}(1, i, 0, \ldots, 0)^{\tt T}$.
First suppose $d=2$. Then $v = (a, b)^{\tt T}$
and $b = \pm a i$, and $v =  a \begin{bmatrix} 1 \\ \pm i \end{bmatrix}$.
As $v \not = 0$, we have $a \not = 0$.
We may use $\begin{bmatrix} 1 & 0 \\ 0 & -1 \end{bmatrix}$,
to get $v = a \begin{bmatrix} 1 \\ i \end{bmatrix}$.
  Use a complex rotation $T_2$ defined
above we get
$T_2 v = a \begin{bmatrix} c + si \\ -s + ci \end{bmatrix}
= a (c+si)
\begin{bmatrix} 1 \\  i \end{bmatrix}$.
As $c+si = e^{iz}$ can be an arbitrary nonzero
complex number, we may choose $z$ such that $e^{iz} = \frac{1}{\sqrt{2} a}$.
This gives us $T_2 v = \frac{1}{\sqrt{2}} \begin{bmatrix} 1 \\ i \end{bmatrix}$.
It is clear that we could also go to
any non-zero multiple of $\begin{bmatrix} 1 \\ i \end{bmatrix}$,
as well as $\begin{bmatrix} 1 \\ -i \end{bmatrix}$.

Now suppose $d >2$. Let $v = (a_1, a_2, \ldots, a_d)^{\tt T}
\not = 0$ be isotropic.
If $a_1 =0$, then $(a_2, \ldots, a_d)^{\tt T} \not = 0$ is isotropic.
By induction there exists an order  $d-1$
orthogonal matrix $T'$ such that
${\rm diag}(1, T') v
= \frac{1}{\sqrt{2}} (0, 1, i, 0, \ldots, 0)^{\tt T}$. Then
we complete the induction by a permutation matrix, obtaining
an order  $d$ orthogonal matrix $T$
such that $T v = \mbox{\boldmath $\beta_0$}$.
Next we assume $a_1 \not =0$.
Then $v' = (a_2, \ldots, a_d)^{\tt T}$ is {\it not} isotropic and non-zero.
By Lemma~\ref{non-isotropic-orth-transf},
there exists an order  $d-1$
orthogonal matrix $T'$ such that
${\rm diag}(1, T') v
= (a_1, \sqrt{\langle v', v'\rangle}, 0, \ldots, 0)^{\tt T}$.
Since $v$ is isotropic, we have
$\sqrt{\langle v', v'\rangle} = \pm a_1 i$.
So we have ${\rm diag}(1, T') v
= a_1 (1, \pm i, 0, \ldots, 0)^{\tt T}$.
And by the above discussion  we get an orthogonal
$T$ such that $T v = \mbox{\boldmath $\beta_0$}$.
We have proved

\begin{lemma}\label{isotropic-orth-transf}
For any non-zero isotropic $v = (a_1, a_2, \ldots, a_d)^{\tt T}
\in \mathbb{C}^d$, $d \ge 2$, there exists an orthogonal matrix $T$
such that $T v = \mbox{\boldmath $\beta_0$}
= \frac{1}{\sqrt{2}}(1, i, 0, \ldots, 0)^{\tt T}$.
(Both $(1, i, 0, \ldots, 0)^{\tt T}$ and
$(1, -i, 0, \ldots, 0)^{\tt T}$, and all non-zero
multiples of them are feasible.)
\end{lemma}

Now set $d=3$.
Our next task is to describe the set of all
order 3 orthogonal matrices $T$ which fixes
$\mbox{\boldmath $\beta_0$}$.

Let the first two columns of $T$ be denoted by
$u = (a_1, a_2, a_3)^{\tt T}$ and
$v = (b_1, b_2, b_3)^{\tt T}$.
We can derive $a_1 = 1 - a_2 i$, $b_1 = a_2$,  $b_2 = 1 + a_2 i$,
and $a_3 = - b_3 i$.
It follows that the first two columns are of the form
$\begin{bmatrix} 1 - ix & x \\
	x & 1+ix \\
		iy & -y
 \end{bmatrix}$.
Moreover, the columns are unit vectors, and so
$x = i y^2 /2$.
If we form the cross-product of these two vectors,
we obtain
$(-iy, y, 1)^{\tt T}$. This and its negation $(iy, -y, -1)^{\tt T}$ can be
the third column vector of $T$.
Thus the orthogonal matrix $T$ has the form
\begin{eqnarray}\label{T-form-orth-fixing-beta}
T =
\begin{bmatrix} 1 + y^2/2 & i y^2/2 & iy \\
                i y^2/2  & 1 - y^2/2 & -y \\
                iy & -y & -1
 \end{bmatrix},
\end{eqnarray}
or changing the last column to its negative.
This is a complete description of the set of $3 \times 3$
orthogonal matrices $T$ such that $T
\mbox{\boldmath $\beta_0$} = \mbox{\boldmath $\beta_0$}$.

Our next task is to determine what canonical form
a vector $v$ can take, under the mapping of such an orthogonal
matrix $T$ which fixes $\mbox{\boldmath $\beta_0$}$.
First we prove a simple lemma.

\begin{lemma}\label{orth-one-beta0-another-beta}
If  $\mathbf{\beta}_1, \mathbf{\beta}_2 \in \mathbb{C}^3$
are isotropic, and linearly independent.
Then $\langle \mathbf{\beta}_1, \mathbf{\beta}_2 \rangle \not =0$,
and
there exists an orthogonal matrix $T$
such that
$T \mathbf{\beta}_1 = \mbox{\boldmath $\beta_0$}$
and $T \mathbf{\beta}_2 =
\langle \mathbf{\beta}_1, \mathbf{\beta}_2 \rangle
\overline{\mbox{\boldmath $\beta_0$}}$.
Let $\lambda = 1/\sqrt{\langle \mathbf{\beta}_1, \mathbf{\beta}_2 \rangle}$,
there exists an orthogonal matrix $T$
such that
$\lambda T \mathbf{\beta}_1 = \mbox{\boldmath $\beta_0$}$
and $\lambda T \mathbf{\beta}_2 =
\overline{\mbox{\boldmath $\beta_0$}}$.
\end{lemma}

\proof
By Lemma~\ref{isotropic-orth-transf}, we have
an orthogonal $T_1$, such that
  $T_1 \mathbf{\beta}_1 =\mbox{\boldmath $\beta_0$}$.
Let $\mathbf{\gamma} = T_1 \mathbf{\beta}_2$.
Write $\mathbf{\gamma} =
\begin{bmatrix}
a \\
b \\
c
\end{bmatrix}$.
If $\langle \mathbf{\beta}_1, \mathbf{\beta}_2 \rangle=0$,
then, since $T_1$ preserves inner product,
 $a+bi =0$ and $c^2 = -(a^2 + b^2) = 0$. Hence,
$\mathbf{\gamma}$ is linearly dependent on $\mbox{\boldmath $\beta_0$}$, and thus
$\mathbf{\beta}_2$ is linearly dependent on $\mathbf{\beta}_1$,
a contradiction. Hence
$\langle \mathbf{\beta}_1, \mathbf{\beta}_2 \rangle \not =0$.

Now we may as well assume the given vectors are
$\mbox{\boldmath $\beta_0$}$ and $\mathbf{\gamma}$.
Consider those orthogonal matrices $T$ in (\ref{T-form-orth-fixing-beta})
fixing $\mbox{\boldmath $\beta_0$}$.
Let ${u} = \mathbf{\gamma} /
\langle {\mathbf{\gamma}}, {\mbox{\boldmath $\beta_0$}} \rangle$.
Then $\langle u, {\mbox{\boldmath $\beta_0$}} \rangle = 1$.
We want a $T$ such that $T \overline{\mbox{\boldmath $\beta_0$}} =
{u}$.
Write
${v} = \frac{1}{\sqrt{2}} {u}
=\begin{bmatrix}
a\\
b\\
c
\end{bmatrix}$,
then
$\langle {v}, (1, i, 0)^{\tt T} \rangle =
\langle {u}, \mbox{\boldmath $\beta_0$} \rangle = 1$,
and it follows that
$a+bi =1$ and so $-c^2 = a^2 + b^2 = 1 - 2bi$.
Hence ${v} = (\frac{1-c^2}{2}, \frac{1+c^2}{2i}, c)^{\tt T}$.
On the other hand, from (\ref{T-form-orth-fixing-beta}),
 $T (1, -i, 0)^{\tt T}
= (1+y^2, (1-y^2)/i, 2yi)^{\tt T}$.
Then by setting $y=c/i$ we get
$T (1, -i, 0)^{\tt T} = 2 {v}$.
Hence
$T \overline{\mbox{\boldmath $\beta_0$}}
= \frac{1}{\sqrt{2}} T (1, -i, 0)^{\tt T}
= \sqrt{2} {v} ={u}$.

The last conclusion  of Lemma~\ref{orth-one-beta0-another-beta}
 follows from what has been
proved applied to the pair
$\lambda  \mathbf{\beta}_1$
and $\lambda  \mathbf{\beta}_2$.
\qed

\begin{lemma}\label{orth-one-beta0-another-e3}
Suppose  $\mathbf{\beta} \in \mathbb{C}^3$
is isotropic, $\mathbf{\gamma} \in \mathbb{C}^3$
is not isotropic, $\{\mathbf{\beta}, \mathbf{\gamma}\}$
are linearly independent, and
$\langle \mathbf{\beta}, \mathbf{\gamma} \rangle =0$.
Then
there exists an orthogonal matrix $T$
such that
$T \mathbf{\beta} = \mbox{\boldmath $\beta_0$}$
and $T \mathbf{\gamma} =
\sqrt{\langle \mathbf{\gamma}, \mathbf{\gamma} \rangle} \mbox{\boldmath $e_3$}$.
For $\lambda = 1/\sqrt{\langle \mathbf{\gamma}, \mathbf{\gamma} \rangle}$,
there exists an orthogonal matrix $T$
such that
$\lambda T \mathbf{\beta} = \mbox{\boldmath $\beta_0$}$
and $\lambda T \mathbf{\gamma} = \mbox{\boldmath $e_3$}$.
\end{lemma}

\proof
By Lemma~\ref{isotropic-orth-transf}, we may assume
  $\mathbf{\beta} =\mbox{\boldmath $\beta_0$}$.
Write $ \lambda
\mathbf{\gamma} =
\begin{bmatrix}
a \\
b \\
c
\end{bmatrix}$.
Then $a+bi =0$ and $c^2 = a^2 + b^2 + c^2 = 1$.
Depending on whether $c=\pm 1$, we use one of the two forms
of $T$ in (\ref{T-form-orth-fixing-beta})
fixing $\mbox{\boldmath $\beta_0$}$.
If $c=-1$, we set $y=-b$ in (\ref{T-form-orth-fixing-beta}).
If $c=+1$, we set $y=-b$ in the form of $T$ with the negated
third column from (\ref{T-form-orth-fixing-beta}).

The last conclusion follows from what has been
proved  applied to the pair
$\lambda  \mathbf{\beta}$
and $\lambda  \mathbf{\gamma}$.
\qed

\vspace{.2in}

We are now ready to address in what canonical form
each of the three cases in Theorem~\ref{thm:ternary} can take.

We consider each case in turn:

\noindent
$\bullet$
 There exist three vectors $\mbox{\boldmath $\alpha$,
$\beta$, and $\gamma$}$
of dimension 3 such that
  they are mutually orthogonal to each other,
 i.e.
$\langle \mbox{\boldmath $\alpha$, $\beta$} \rangle=0$,
     $\langle \mbox{\boldmath $\alpha$, $\gamma$} \rangle=0$,
 $\langle \mbox{\boldmath $\beta$,  $\gamma$} \rangle=0$,  and
 \[ \mathbf{F}=\mbox{\boldmath $\alpha$}^{\otimes 3} +
                \mbox{\boldmath $\beta$}^{\otimes 3} +
                \mbox{\boldmath $\gamma$}^{\otimes 3}.\]

Let $r = {\rm rank} \{\mbox{\boldmath $\alpha$, $\beta$,  $\gamma$} \}$.
If $r=0$, then $\mathbf{F}= \mathbf{0}$ is the identically zero function.

If $r=1$, and suppose $\mbox{\boldmath $\alpha$} \not =0$ and
$\mbox{\boldmath $\beta$}$ and $\mbox{\boldmath $\gamma$}$ are linear multiples of
$\mbox{\boldmath $\alpha$}$. Then $\mathbf{F}= \mbox{\boldmath $\alpha'$}^{\otimes 3}$
for some $\mbox{\boldmath $\alpha'$}$.
Depending on whether $\mbox{\boldmath $\alpha'$}$ is isotropic, under an
orthogonal transformation, $T \mathbf{F}$ takes the form
\begin{equation}\label{rank=1-case-1}
T \mathbf{F} = \mbox{\boldmath $\beta_0$}^{\otimes 3},  ~~~~\mbox{or}~~~~
\lambda \mbox{\boldmath $e_3$}^{\otimes 3}.
\end{equation}

Let $r=2$ and suppose $\mbox{\boldmath $\alpha$}$ and $\mbox{\boldmath $\beta$}$
are linearly independent. We show that without loss of generality
 we may assume $\mbox{\boldmath $\gamma$} = \mathbf{0}$.
Let $\mbox{\boldmath $\gamma$} = a \mbox{\boldmath $\alpha$} + b
\mbox{\boldmath $\beta$}$.
Then $\langle \mbox{\boldmath $\gamma$}, \mbox{\boldmath $\gamma$} \rangle=0$.
If either $a =0$ or $b=0$, we can combine the term
$\mbox{\boldmath $\gamma$}^{\otimes 3}$ with either
$\mbox{\boldmath $\beta$}^{\otimes 3}$ or $\mbox{\boldmath $\alpha$}^{\otimes 3}$
respectively, and the term $\mbox{\boldmath $\gamma$}^{\otimes 3}$ disappears.
If both $a, b \not = 0$. By
$\langle \mbox{\boldmath $\alpha$, $\beta$} \rangle=0$,
     $\langle \mbox{\boldmath $\alpha$, $\gamma$} \rangle=0$,
 $\langle \mbox{\boldmath $\beta$,  $\gamma$} \rangle=0$,
we get $a \langle  \mbox{\boldmath $\alpha$, $\alpha$} \rangle
= b \langle \mbox{\boldmath $\beta$, $\beta$} \rangle =0$. Hence
$\langle  \mbox{\boldmath $\alpha$, $\alpha$} \rangle
= \langle \mbox{\boldmath $\beta$, $\beta$} \rangle =0$.
This contradicts Lemma~\ref{orth-one-beta0-another-beta},
by linear independence.
Therefore in case $r=2$ we only need to consider
$ \mathbf{F}=\mbox{\boldmath $\alpha$}^{\otimes 3}
+ \mbox{\boldmath $\beta$}^{\otimes 3}$,
and  $\mbox{\boldmath $\alpha$}$ and $\mbox{\boldmath $\beta$}$
are linearly independent.

By Lemma~\ref{orth-one-beta0-another-beta}
 $\mbox{\boldmath $\alpha$}$ and $\mbox{\boldmath $\beta$}$
can not be both  isotropic.
Suppose one of them is isotropic.
By Lemma~\ref{orth-one-beta0-another-e3},
$\mathbf{F}$ takes the form
\begin{equation}\label{rank=2-case-1-iso}
\mbox{\boldmath $\beta_0$}^{\otimes 3}
+  \lambda  \mbox{\boldmath $e_3$}^{\otimes 3}
\end{equation}
under an orthogonal transformation.

If $r=2$ and both $\mbox{\boldmath $\alpha$}$ and $\mbox{\boldmath $\beta$}$ are
not isotropic, then there exists an orthogonal matrix $T$
such that
$T \mbox{\boldmath $\alpha$} = \lambda \mbox{\boldmath $e_1$}$ and
$T \mbox{\boldmath $\beta$} = \mu \mbox{\boldmath $e_2$}$, thus
$\mathbf{F}$ takes the form
\begin{equation}\label{rank=2-case-1-not-iso}
\lambda \mbox{\boldmath $e_1$}^{\otimes 3} +  \mu \mbox{\boldmath $e_2$}^{\otimes 3}
\end{equation}
under an orthogonal transformation.

Now suppose $r=3$. We claim none of
$\mbox{\boldmath $\alpha$,
$\beta$, and $\gamma$}$
can be isotropic.
Otherwise, say $\mbox{\boldmath $\alpha$}$ is isotropic,
then the linearly independent set
$\{\mbox{\boldmath $\alpha$,
$\beta$,  $\gamma$}\}$ spans the conjugate vector
 $\overline{\mbox{\boldmath $\alpha$}}$. Then it follows that
$\langle \mbox{\boldmath $\alpha$}, \overline{\mbox{\boldmath $\alpha$}} \rangle =0$
and $\mbox{\boldmath $\alpha$} = \mathbf{0}$, a contradiction.
Hence,
under an orthogonal transformation $\mathbf{F}$ takes the form
\begin{equation}\label{rank=3-case-1}
\lambda \mbox{\boldmath $e_1$}^{\otimes 3}
+  \mu \mbox{\boldmath $e_2$}^{\otimes 3}
+  \nu \mbox{\boldmath $e_3$}^{\otimes 3}
\end{equation}
%


\noindent
$\bullet$
There exist three vectors
$\mbox{\boldmath $\alpha$, $\beta_1$, and $\beta_2$}$
of dimension 3 such that
  $\langle \mbox{\boldmath $\alpha$, $\beta_1$} \rangle=0$,
  $\langle \mbox{\boldmath $\alpha$, $\beta_2$} \rangle=0$,
  $\langle \mbox{\boldmath $\beta_1$, $\beta_1$} \rangle=0$,
  $\langle \mbox{\boldmath $\beta_2$, $\beta_2$} \rangle=0$
and
  \[ \mathbf{F}=\mbox{\boldmath $\alpha$}^{\otimes 3} +
                 \mbox{\boldmath $\beta_1$}^{\otimes 3} +
                \mbox{\boldmath $\beta_2$}^{\otimes 3}.\]

Let $r = {\rm rank} \{\mbox{\boldmath $\beta_1$}, \mbox{\boldmath $\beta_2$}\}$.
If $r=0$, then $\mathbf{F}=\mbox{\boldmath $\alpha$}^{\otimes 3}$.
%
%
If $r=1$, we can combine the terms
$\mbox{\boldmath $\beta_1$}^{\otimes 3}$
and $\mbox{\boldmath $\beta_2$}^{\otimes 3}$,
and $\mathbf{F}$ takes the form
$\mbox{\boldmath $\alpha$}^{\otimes 3}
+ \mbox{\boldmath $\beta'$}^{\otimes 3}$,
with $\langle \mbox{\boldmath $\alpha$},
\mbox{\boldmath $\beta'$}  \rangle = 0$.
These cases have already been classified in the first form.
$\mathbf{F}$ takes the forms in (\ref{rank=1-case-1}), (\ref{rank=2-case-1-iso})
or (\ref{rank=2-case-1-not-iso}).

 Suppose $r=2$.
By Lemma~\ref{orth-one-beta0-another-beta},
 for a suitable non-zero constant
$\lambda = 1/\sqrt{\langle \mbox{\boldmath $\beta_1$},
\mbox{\boldmath $\beta_2$}  \rangle}$,
there exists an orthogonal matrix $T$
such that
$\lambda T \mbox{\boldmath $\beta_1$} = \mbox{\boldmath $\beta_0$}$
and $\lambda T \mbox{\boldmath $\beta_2$}  =
\overline{\mbox{\boldmath $\beta_0$}}$.
Under this transformation $\lambda T$,
$\mbox{\boldmath $\alpha$}$ is orthogonal to $\mbox{\boldmath $e_1$}$ and
$\mbox{\boldmath $e_2$}$
which are in the linear span of $\mbox{\boldmath $\beta_0$}$
and $\overline{\mbox{\boldmath $\beta_0$}}$.
Hence $\mbox{\boldmath $\alpha$}$ takes the form $c \mbox{\boldmath $e_3$}$.

We have proved that in this case, for some non-zero constant
$\lambda$ and orthogonal matrix $T$,
\begin{equation}\label{beta-beta-bar+e3}
\lambda T \mathbf{F}= \mbox{\boldmath $\beta_0$}^{\otimes 3} +
\overline{\mbox{\boldmath $\beta_0$}}^{\otimes 3} + c
\mbox{\boldmath $e_3$}^{\otimes 3}.
\end{equation}

\noindent
$\bullet$
There exist two vectors $\mbox{\boldmath $\beta$ and $\gamma$}$
of dimension 3 and a (symmetric) function
${\mathbf{F}}_{\beta}$
 of arity three,
such that $\mbox{\boldmath $\beta$} \neq \mathbf{0}$,
$\langle \mbox{\boldmath $\beta$, $\beta$} \rangle=0$,
$\langle \mathbf{F}_{\beta}, \mbox{\boldmath $\beta$}
 \rangle = \mathbf{0}$ and
  \[ \mathbf{F}= \mathbf{F}_{\beta}
+ \mbox{\boldmath $\beta$}^{\otimes 2} \otimes \mbox{\boldmath $\gamma$}
+ \mbox{\boldmath $\beta$} \otimes \mbox{\boldmath $\gamma$} \otimes
\mbox{\boldmath $\beta$}
+ \mbox{\boldmath $\gamma$}  \otimes  \mbox{\boldmath $\beta$}^{\otimes 2}.\]

First we note that $\mbox{\boldmath $\beta$}^{\otimes 3}$ also satisfies
the annihilation condition,
$\langle \mathbf{F}_{\beta}, \mbox{\boldmath $\beta$} \rangle = \mathbf{0}$,
and can be combined to $\mathbf{F}_{\beta}$.
Hence we can replace $\mbox{\boldmath $\gamma$}$ by any
$\mbox{\boldmath $\gamma$} + \lambda \mbox{\boldmath $\beta$}$.

There are  the following cases, depending on whether
$\langle \mbox{\boldmath $\beta$},
\mbox{\boldmath $\gamma$} \rangle =0$
and whether  $\mbox{\boldmath $\gamma$}$
is isotropic.

Suppose
$\langle \mbox{\boldmath $\beta$},
\mbox{\boldmath $\gamma$} \rangle =0$.
Then we can eliminate the terms
$\mbox{\boldmath $\beta$}^{\otimes 2} \otimes
\mbox{\boldmath $\gamma$}
+ \mbox{\boldmath $\beta$} \otimes \mbox{\boldmath $\gamma$} \otimes
\mbox{\boldmath $\beta$}
+ \mbox{\boldmath $\gamma$} \otimes \mbox{\boldmath $\beta$}^{\otimes 2}$
by combining
it  to  $\mathbf{F}_{\beta}$.
We can transform $\mbox{\boldmath $\beta$}$ to
$\mbox{\boldmath $\beta_0$}$.
In this case, $\mathbf{F}$ takes the form
\begin{equation}\label{singleF0}
T \mathbf{F} = \mathbf{F}_{{\beta}_0}
\end{equation}
where $\langle \mathbf{F}_{{\beta}_0}, \mbox{\boldmath $\beta_0$} \rangle
= \mathbf{0}$.

Suppose $\mbox{\boldmath $\gamma$}$ is isotropic and
$\langle \mbox{\boldmath $\beta$},
\mbox{\boldmath $\gamma$} \rangle \not = 0$.
Then $\mbox{\boldmath $\beta$}$ and $\mbox{\boldmath $\gamma$}$
are linearly independent.
By
Lemma~\ref{orth-one-beta0-another-beta}
there exists
an orthogonal matrix $T$
such that
\[T \mathbf{F} = \mathbf{F}_{\mathbf{\beta}_0} +
\lambda (\mbox{\boldmath $\beta_0$}^{\otimes 2} \otimes
\overline{\mbox{\boldmath $\beta_0$}}
+ \mbox{\boldmath $\beta_0$} \otimes  \overline{\mbox{\boldmath $\beta_0$}}
 \otimes  \mbox{\boldmath $\beta_0$}
+  \overline{\mbox{\boldmath $\beta_0$}} \otimes
\mbox{\boldmath $\beta_0$}^{\otimes 2}),\]
where $\lambda = \langle \mbox{\boldmath $\beta$},
\mbox{\boldmath $\gamma$} \rangle \not = 0$, and
 $\langle \mathbf{F}_{\beta_0},  \mbox{\boldmath $\beta_0$}
\rangle  =0$.
Let $T_2 = \begin{bmatrix} c & s \\
                        -s & c
\end{bmatrix}$, where $c = \cos z$ and $s = \sin z$. Then
$T_2$ maps
$\begin{bmatrix} 1 \\ i \end{bmatrix}$
to $(c + si)
\begin{bmatrix} 1 \\ i \end{bmatrix}$
and maps
$\begin{bmatrix} 1 \\ -i \end{bmatrix}$
to $(c -si)
\begin{bmatrix} 1 \\ -i \end{bmatrix}$.
To each term in
\[\mbox{\boldmath $\beta_0$}^{\otimes 2} \otimes
\overline{\mbox{\boldmath $\beta_0$}}
+ \mbox{\boldmath $\beta_0$} \otimes  \overline{\mbox{\boldmath $\beta_0$}}
 \otimes  \mbox{\boldmath $\beta_0$}
+  \overline{\mbox{\boldmath $\beta_0$}} \otimes
\mbox{\boldmath $\beta_0$}^{\otimes 2},\]
${\rm diag}( T_2, 1 )^{\otimes 3}$
contributes a factor $(c + si)^2 (c -si) = c + si = e^z$,
which can be an arbitrarily chosen non-zero complex number.
In particular we can set it to $1/\lambda$.
Also note that ${\rm diag}( T_2, 1 )^{\otimes 3}$
transforms $\mathbf{F}_{\beta_0}$ to another
such function satisfying the annihilation condition
$\langle \mathbf{F}_{\beta_0}, \mbox{\boldmath $\beta_0$}
\rangle = \mathbf{0}$.
Thus we obtain the form of $\mathbf{F}$ under an orthogonal transformation
\begin{equation}\label{F0plus}
\mathbf{F}_{\beta_0} +
\mbox{\boldmath $\beta_0$}^{\otimes 2} \otimes
\overline{\mbox{\boldmath $\beta_0$}}
+ \mbox{\boldmath $\beta_0$} \otimes
\overline{\mbox{\boldmath $\beta_0$}}
 \otimes  \mbox{\boldmath $\beta_0$}
+  \overline{\mbox{\boldmath $\beta_0$}}
\otimes  \mbox{\boldmath $\beta_0$}^{\otimes 2}.
\end{equation}


Suppose $\mbox{\boldmath $\gamma$}$ is not isotropic and
$\langle \mbox{\boldmath $\beta$},
\mbox{\boldmath $\gamma$} \rangle \not = 0$.
Then we replace $\mbox{\boldmath $\gamma$}$
by
$\mbox{\boldmath $\gamma$} - c \mbox{\boldmath $\beta$}$,
where $c = \langle \mbox{\boldmath $\gamma$}, \mbox{\boldmath $\gamma$}
\rangle/
 (2 \langle \mbox{\boldmath $\beta$}, \mbox{\boldmath $\gamma$} \rangle)$.
Then $\mbox{\boldmath $\gamma$} - c \mbox{\boldmath $\beta$}$ is isotropic
and we have reduced to the previous case.

Summarizing, we note that  (\ref{rank=1-case-1})
and (\ref{rank=2-case-1-iso}) are special cases of
(\ref{singleF0}).
(\ref{rank=2-case-1-not-iso}) is a special case of
(\ref{rank=3-case-1}).
Then it is clear that Theorem~\ref{thm:normal-forms}
is equivalent to Theorem~\ref{thm:ternary}.

%

\section{Tractability}\label{section:tractability}
Suppose $\mathbf{F} = [3; 1, 1; 5, 1, 3; 7, 5, 1, 1]$.
Is Holant$^*(\mathbf{F})$ computable in polynomial time?
It turns out that there are three pairwise orthogonal
vectors
$(1, -1, 1)^{\tt T}, (1, 0, -1)^{\tt T}$  and $(1, 2, 1)^{\tt T}$
such that
\( \mathbf{F}
= {\scriptsize \begin{bmatrix} 1 \\ -1 \\ 1 \end{bmatrix}^{\otimes 3} +
   \begin{bmatrix} 1 \\ 0 \\ -1 \end{bmatrix}^{\otimes 3} +
    \begin{bmatrix} 1 \\ 2 \\ 1 \end{bmatrix}^{\otimes 3}}.\)
By Theorem~\ref{thm:ternary}, Holant$^*(\mathbf{F})$ is tractable.
If we take $T = \frac{1}{\sqrt{6}} {\scriptsize
\begin{bmatrix} \sqrt{2} & \sqrt{3} & 1 \\
                - \sqrt{2} & 0      & 2\\
                \sqrt{2} & -\sqrt{3} & 1
\end{bmatrix}}$, then $T$ is orthogonal, and
$\mathbf{F} = T^{\otimes 3} \mathbf{F}'$,
where $\mathbf{F}' = \sqrt{27} \mbox{\boldmath $e_1$}^{\otimes 3} +
               \sqrt{8} \mbox{\boldmath $e_2$}^{\otimes 3} +
               \sqrt{216} \mbox{\boldmath $e_3$}^{\otimes 3}$.
Hence we can perform an orthogonal transformation by $T$,
then the  problem Holant$^*(\mathbf{F})$
is transformed to Holant$^*(\mathbf{F}')$.
For $\mathbf{F}'$ the polynomial time algorithm on any input graph $\Gamma$ is
simple:  In each connected component of $\Gamma$, any color from $\{B, G, R\}$
at a vertex $v$ uniquely determines the same color at all its neighbors,
and the vertex contributes a factor $\sqrt{27}$ or $\sqrt{8}$ or $\sqrt{216}$
respectively.  These values are multiplied over the connected component.
Thus, if $G$ has connected components $C_1, C_2, \ldots, C_k$,
and $C_j$ has $n_j$ vertices, then the Holant values is
$\prod_{1 \le j \le k} (\sqrt{27}^{n_j} + \sqrt{8}^{n_j} + \sqrt{216}^{n_j})$.

We believe for countless such questions, not only the problem
is very natural, but also the answer is not obvious without the underlying theory.
Note that even though the function $\mathbf{F}$
above takes only positive values,
the vectors can have negative entries.
Armed with the dichotomy theorem, any interested reader can find
many more examples.

In this section we prove that
Holant$^*(\mathbf{F})$ is computable in polynomial time,
for any symmetric ternary function
 $\mathbf{F}$ given in the three forms of Theorem~\ref{thm:ternary},
or equivalently Theorem~\ref{thm:normal-forms}.

For any $3 \times  3$ orthogonal matrix $T$, it keeps the binary equality
$(=_2)$ over $\{B, G, R\}$ unchanged,
namely $T^{\tt T} I_3 T = I_3$ in matrix notation.
Hence Holant$^*(\mathbf{F})$ is tractable iff
Holant$^*(T^{\otimes 3} \mathbf{F})$ is
 tractable.

The above argument
proves that Holant$^*(\mathbf{F})$ is computable in polynomial time
if $\mathbf{F}$
 has form {\it 1.}
\[ a \mbox{\boldmath $e_1$}^{\otimes 3} +
               b \mbox{\boldmath $e_2$}^{\otimes 3} +
               c \mbox{\boldmath $e_3$}^{\otimes 3}.\]

In form {\it 2.}, let $\mathbf{F}$ be
\[\mbox{\boldmath $\beta_0$}^{\otimes 3} +
\overline{\mbox{\boldmath $\beta_0$}}^{\otimes 3} +
 \lambda  \mbox{\boldmath $e_3$}^{\otimes 3}.\]

Under the
matrix
$M = {\scriptsize \begin{bmatrix} Z^{-1}  & 0 \\
                 0  & 1 \end{bmatrix}}$, where
${\scriptsize Z=\frac{1}{\sqrt{2}}\begin{bmatrix} 1 & 1 \\ i & -i \end{bmatrix}}$,
${\scriptsize Z^{-1} = \frac{1}{\sqrt{2}}\begin{bmatrix} 1 & -i \\ 1 & i \end{bmatrix}}$,
the function $\mathbf{F}$ is transformed to
\[ M^{\otimes 3} \mathbf{F}
= \mbox{\boldmath $e_1$}^{\otimes 3} +
                \mbox{\boldmath $e_2$}^{\otimes 3} +
                 \lambda \mbox{\boldmath $e_3$}^{\otimes 3}.\]
Meanwhile the covariant transformation on the binary equality is
$ (=_2) (M^{-1})^{\otimes 2}$, which has the matrix form
$(M^{-1})^{\tt T} I M^{-1} = {\tiny \begin{bmatrix}  0 & 1 & 0 \\
                  1 & 0 & 0 \\
                  0 & 0 & 1\end{bmatrix}}$.
This  can be viewed as  a  Disequality on $\{B, G\}$ and
Equality on $\{R\}$, with a separated domain.
%
Now it is clear that Holant$^*(\mathbf{F})$ is computable in polynomial time
by a connectivity argument. Within each connected component,
any assignment of $R$ will be uniquely propagated as $R$; any
assignment of $B$ or $G$  will be exchanged to  $G$ or $B$
along every edge.

The proof of tractability for form {\it 3.} is more involved.
We refer to the more generic expression of form 3 in
Theorem~\ref{thm:ternary}.
First, under an  orthogonal transformation we may assume
$\mbox{\boldmath $\beta$} = \begin{bmatrix} 1 & i & 0 \end{bmatrix}^{\tt T}$.
The function $\mathbf{F}$ is expressed as a sum
$S +
\mbox{\boldmath $\beta$}^{\otimes 2} \otimes \mbox{\boldmath $\gamma$}
+ \mbox{\boldmath $\beta$} \otimes \mbox{\boldmath $\gamma$} \otimes
\mbox{\boldmath $\beta$}
+ \mbox{\boldmath $\gamma$}  \otimes  \mbox{\boldmath $\beta$}^{\otimes 2}$,
where $\langle S, \mbox{\boldmath $\beta$}
 \rangle = \mathbf{0}$.
We denote by $T_0 = S$, and $T_j$ for the remaining
three terms respectively, $1 \le j \le 3$.
The value Holant$^*(\mathbf{F})$ is
the sum over all $\{B,G,R\}$ edge assignments,
$\sum_\sigma \prod_v  f_v (\sigma \mid_{E(v)})$,
where $E(v)$ are the edges incident to $v$,  and  all $f_v$ are the
function $\mathbf{F}$, or some unary function.

Without loss of generality, we can assume the input graph is
connected. In the first step, we handle all vertices of degree one.
Such a vertex $v$ is connected to another vertex $p$ of degree $d$.
We can calculate a function of arity $d-1$ by combining
the unary function at  $v$ with the function at $p$.
This is a symmetric function and we can replace
the vertex $p$ together with $v$
 by a vertex $q$ of degree $d-1$ and given this function.
If $d=1$, since the graph is connected, there is no
vertex left and we have computed the value of the problem.
If $d=2$,
the new function at $q$  is a unary function.
If $d=3$, then $f_p$ is $\mathbf{F}$.
We may repeat this process until all vertices are of degree 2 or 3
and given either $\mathbf{F}$ or
$\langle \mathbf{u}, \mathbf{F} \rangle
= \sum_{j=0}^3 T'_j$ for some unary $\mathbf{u}$, where
$T'_j =  \langle \mathbf{u}, T_j \rangle$.

For every vertex $v$ of degree 2 or 3, we can express the function  $f_v$
as $\sum_{j=0}^3 T'_j$ or $\sum_{j=0}^{3} T_j$ with
the incident edges assigned as
(ordered) input variables to each $T'_j$ or $T_j$.
(Note that $T'_j$  and $T_j$ are in general
not symmetric, for $1 \le j \le 3$.)
Then Holant$^*(\mathbf{F})
 =\sum_\tau \sum_\sigma \prod_{v}
f_{v,\tau(v)} (\sigma \mid_{E(v)})$,
where the first summation
is over all assignments $\tau$ from all vertices $v \in V$
to some  $j = \tau(v) \in \{0,1,2,3\}$ which assigns a copy of $T'_j$ or $T_j$
as $f_{v,\tau(v)}$ at $v$.

We are given that
$\langle \mbox{\boldmath $\beta$}, T_0
 \rangle = \mathbf{0}$, then
$\langle \mbox{\boldmath $\beta$}, T'_0 \rangle = \mathbf{0}$ as well.
Meanwhile
$T'_1 = c_1 \mbox{\boldmath $\beta$}^{\otimes 2}$,
$T'_2 = c_2 \mbox{\boldmath $\beta$} \otimes \mbox{\boldmath $\gamma$}$,
and
$T'_3 = c_3 \mbox{\boldmath $\gamma$} \otimes \mbox{\boldmath $\beta$}$,
where the constants
$c_1 = \langle \mathbf{u},  \mbox{\boldmath $\gamma$} \rangle$,
and $c_2 = c_3 = \langle \mathbf{u},  \mbox{\boldmath $\beta$} \rangle$.
Note that $T'_j$ and $T_j$, for $1 \le j \le 3$,
are all degenerate functions, and can be decomposed as unary functions.
We also note that  they all have at least as many copies of
$\mbox{\boldmath $\beta$}$ as $\mbox{\boldmath $\gamma$}$.

Fix any $\tau$, let $\mathcal{S}$ (resp. $\mathcal{T}$)
denote the set of vertices which are assigned  the function $T_0$ or $T'_0$
(resp. $T_j$ or $T'_j$, with $1 \le j \le 3$) by $\tau$.
Suppose neither $\mathcal{S}$ nor $\mathcal{T}$ is empty.
Then by connectedness,
there are edges between  $\mathcal{S}$ and $\mathcal{T}$.
All functions in $\mathcal{T}$ are decomposed into
unary functions. There are at least as many copies of
$\mbox{\boldmath $\beta$}$ as $\mbox{\boldmath $\gamma$}$.
Some of these functions may be paired up by edges inside
$\mathcal{T}$. If any two copies of $\mbox{\boldmath $\beta$}$
are paired up, the product is zero.
If every copy of  $\mbox{\boldmath $\beta$}$
is paired up with some $\mbox{\boldmath $\gamma$}$ within $\mathcal{T}$,
then at least
one copy of $\mbox{\boldmath $\beta$}$
is connected to some vertex in $\mathcal{S}$.
But every function in $\mathcal{S}$ is annihilated by
$\mbox{\boldmath $\beta$}$.  Hence the total contribution for such $\tau$
to Holant$^*(\mathbf{F})$ is zero  when
$\mathcal{S}$ and  $\mathcal{T}$ are both non-empty.

Now consider
$\sum_\sigma \prod_v f_{v,\tau(v)}(\sigma \mid_{E(v)})$
for those $\tau$ such that either $\mathcal{S}$ or $\mathcal{T}$ is
empty. Suppose $\mathcal{S} = \emptyset$.
Again we decompose every function in $\mathcal{T}$
into unary functions. Then in order to be non-zero,
the number of $\mbox{\boldmath $\beta$}$ and
$\mbox{\boldmath $\gamma$}$ must be exactly equal.
Hence if there is any vertex of degree 3, the contribution
is 0. We only need to
consider a connected graph such that all vertices have degree 2, which is a
cycle. Because each $\mbox{\boldmath $\beta$}$  must be
paired up exactly with $\mbox{\boldmath $\gamma$}$,
We only
need to calculate the  sum
$\sum_\sigma \prod_v f_{v,\tau(v)}(\sigma \mid_{E(v)})$ for two $\tau$,
which is tractable,
since the graph is just a cycle.

Finally suppose $\mathcal{T} = \emptyset$.
Then there is only one assignment $\tau$ which assigns
$T_0$ and $T'_0$ to every vertex of degree 3 and 2 respectively.
Consider all edge assignments
$\sigma$. Suppose $E=\{e_1,e_2, \ldots, e_m\}$ is the
edge set, and
$e_1=(p,q)$. All assignments $\sigma$ are divided into 3 sets $\Sigma_B$,
$\Sigma_G$ or $\Sigma_R$, according to the value $\sigma(e_1) = B$,  $G$
or $R$, respectively.
There is a natural one-to-one mapping $\phi$ from $\Sigma_B$ to
$\Sigma_G$, such that $(\phi(\sigma))(e_j)= \sigma (e_j)$ for
$j=2,\ldots,m$. Let $\theta (\sigma)$ denote $\prod_v
f_{v,\tau(v)}(\sigma \mid_{E(v)})$,
where $E(v)$ are the edges incident to $v$.
Notice that at all $v \not = p, q$,
the value of $f_{v,\tau(v)}$ is the same for $\sigma$ and $\phi(\sigma)$.
But at $v = p, q$,
$f_{v,\tau(v)}(\phi(\sigma) \mid_{E(v)})
= i f_{v,\tau(v)}(\sigma \mid_{E(v)})$.
This can be directly verified.
%
Hence $\theta (\phi(\sigma))=-\theta (\sigma)$.
Therefore we only need to
calculate $\theta (\sigma)$ for $\sigma$ in $\Sigma_R$. We can use
$\sigma(e_2)$ to divide $\Sigma_R$ into 3 sets, to repeat this
process. At last, we only need to calculate $\theta (\sigma)$ for
the single $\sigma$ mapping every edge to $R$. This
concludes the proof of tractability.

\section{\#P-hardness}\label{section:hardness}

The starting point of our hardness proof is the dichotomy for
 Holant$^*(\mathbf{F})$ problems on the Boolean domain.
A natural hope is that Holant$^*(\mathbf{F})$ is \#P-hard if the Boolean
domain  Holant$^*$ problem for the function $\mathbf{F}^{*\rightarrow \{G,R\}}$, which is the restriction of the function $\mathbf{F}$ to the two-element subdomain $\{G,R\}$,   is already \#P-hard.
But this statement is false when stated in such full generality,
 as we can easily construct an $\mathbf{F}$ such that Holant$^*(\mathbf{F})$ is tractable while Holant$^*(\mathbf{F}^{*\rightarrow \{G,R\}})$ is \#P-hard
(e.g., the first example in Section~\ref{section:tractability}).
However, this would be true if we have another
special binary function $(=_{G,R})
   =
{\tiny \begin{bmatrix}
   0 & 0 & 0 \\
   0 & 1 & 0 \\
   0 & 0 & 1
   \end{bmatrix}}
 $.
 The reduction is straightforward: Given an instance $G$ of Holant$^*(\mathbf{F}^{*\rightarrow \{G,R\}})$, we construct an instance of Holant$^*(\mathbf{F})$ by inserting a vertex into each edge of $G$ and assigning the binary function $=_{G,R}$ to these vertices.  The binary function $=_{G,R}$ in each edge acts as an equality function in the Boolean subdomain  $\{G,R\}$
while any assignment of $B$ anywhere produces a zero.

Therefore, our first main step (from Section \ref{sec:rank2} to \ref{Interpolation-subsection}) is to construct the function  $=_{G,R}$.  If we can construct a non-degenerate binary function with the form ${\tiny
\begin{bmatrix}
   0 & 0 & 0 \\
   0 & * & * \\
   0 & * & *
   \end{bmatrix}
}$, we can use interpolation to interpolate $=_{G,R}$ by a chain of copies
of the above binary function as showed in Section \ref{Interpolation-subsection}. The remaining task is to realize such a binary function.

However we find that it is difficult or impossible
 to realize it directly by gadget construction in most cases. Here we use the idea of holographic reduction. As shown
 in the tractability part, holographic reduction plays an essential role  there
in developing polynomial algorithms. It also plays an
important role in the hardness proof part as a method to normalize functions. We can always apply an orthogonal holographic transformation to
a signature function without changing its complexity as shown in Theorem \ref{thm:orthogoal}.
 If we can realize a binary function with rank 2, which can be constructed
 directly with the help of unary functions (see Lemma \ref{lemma:just-rank2}), then we can hope
to use a holographic reduction to transform the binary function to the
above form. This fits well with the idea of holographic reduction. A binary function with rank 2 shows that there is a hidden structure with a domain of size 2. The holographic reduction mixes the domain elements in a suitable way
so that this hidden  Boolean subdomain becomes explicit.

There are certain rank 2 matrices such as
  ${\tiny \begin{bmatrix} 0 & 0 & 1 \\
   0 & 0 & i \\
   1 & i & 0
   \end{bmatrix}}$, for which an orthogonal holographic transformation does not exist. The reason is that the eigenvector of this matrix corresponding to the eigenvalue 0 is isotropic.
   We shall handle such cases in Lemma \ref{non-iso-rank2}.
   This is the first place where isotropic vectors present some obstacle to our proof. There are several places throughout the entire proof, where we have to deal with  isotropic vectors separately. There are two reasons: (1) For
an isotropic vector, we cannot normalize it to a unit vector by an orthogonal
 transformation; (2) There are indeed additional tractable functions which are related to isotropic vectors.
Consequently we have to circumvent this obstacle presented by the
isotropic eigenvectors.

Additionally,
there are some exceptional cases where the above process cannot go through.
For these cases, we either prove the hardness result directly or
show that it belongs to one of the three forms in Theorem \ref{thm:ternary}.
In the second main step (from Section \ref{sec:domain2} to \ref{sec:0000}), we assume that we are already given $=_{G,R}$ and we
further prove that Holant$^*(\mathbf{F})$ is \#P-hard
if $\mathbf{F}$ is not of one of the three forms in Theorem \ref{thm:ternary}.

Given $=_{G,R}$, Holant$^*(\mathbf{F})$ is \#P-hard if Holant$^*(\mathbf{F}^{*\rightarrow \{G,R\}})$ is \#P-hard, which we use our previous dichotomy
for Boolean Holant$^*$ to determine. Hence we may assume that $\mathbf{F}^{*\rightarrow \{G,R\}}$ takes a tractable form. At this point, we employ holographic reduction to normalize our function further. But we should be careful here since we do not want the transformation to destroy $=_{G,R}$. We introduce the
idea of a domain separated holographic reduction. A basis for a domain separated holographic transformation is of the form
${\tiny \begin{bmatrix} * & 0 & 0 \\
   0 & * & * \\
   0 & * & *
   \end{bmatrix}}$, which mixes up the subdomain $\{G,R\}$ while keeping $B$ separate. In particular,
such orthogonal holographic transformations preserve $=_{G,R}$.

For example, when $\mathbf{F}^{*\rightarrow \{G,R\}}$ is a non-degenerate Fibonacci signature with two distinct roots (Case 1 in Section~\ref{sec:domain2}), we can apply an orthogonal holographic transformation of this form so that $\mathbf{F}$ is transformed to
\begin{center}
{\scriptsize
\begin{tabular}{*{11}{c}c}
{ }     & { }   & { }   &$F_{BBB}$& { }     & { }   & { }  \\
{ }     & { }   &$F_{BBG}$& { } &$F_{BBR}$& { }     & { }  \\
{ }     &$F_{BGG}$& { } &  $F_{BGR}$& { }   &$F_{BRR}$& { }  \\
 $a$& { } &$0$& { } & $0$& { } &   $b$  \\
\end{tabular}
}
\end{center}
According to the Holant$^*$ dichotomy on domain size 2,
when putting this $\mathbf{F}^{*\rightarrow \{G,R\}}=[a, 0, 0, b]$ and a binary
function together, the problem is \#P-hard unless the binary
function is of the form $[*, 0, *]$, $[0,*,0]$, or degenerate.
We shall prove that we can always construct a binary function which is not of these forms unless the function $\mathbf{F}$ has an {\it uncanny} regularity such that it is one of the forms in Theorem \ref{thm:ternary}.

One idea greatly simplifies our argument in this part. By gadget construction, we can realize some binary functions with some parameters, which we can
set freely to any complex number.
Then we want to prove that we can set these parameters suitably so that the signature escapes from all the known
tractable forms. This is quite difficult since different values may make the signature  belong to different tractable forms. A nice observation here is that the condition that a binary signature belongs a particular form say  $[*, 0, *]$ can be described by the zero set of a polynomial.
Thus these values form an algebraic set.
To escape from a finite union of such sets,
it is sufficient to prove that for every form, we can set these parameters
to escape from this particular form. We call this the {\it polynomial
argument}.

The spirit of the proof for all the other tractable non-degenerate ternary
forms for $\mathbf{F}^{*\rightarrow \{G,R\}}$ is similar although the details are very different (there are three cases in Section \ref{sec:domain2}). In particular, we need to employ
a non-orthogonal holographic transformation ${\scriptsize
\begin{bmatrix} 1 & \mathbf{0} \\ \mathbf{0} & Z \end{bmatrix}}$ where
$Z={\scriptsize \frac{1}{\sqrt{2}}\begin{bmatrix} 1 & 1 \\ i & -i \end{bmatrix}}$. This transformation does not preserve $=_{G,R}$, rather it
transforms $=_{G,R}$ to $(\neq_{G,R})={\tiny \begin{bmatrix}
   0 & 0 & 0 \\
   0 & 0 & 1 \\
   0 & 1 & 0
   \end{bmatrix}}$.

When the ternary signature   $\mathbf{F}^{*\rightarrow \{G,R\}}$ is degenerate, the proof structure is quite different (from Section \ref{bottom1000} to \ref{sec:0000}). The reason is that  any set of binary functions are tractable in the Holant framework. So we have to construct a non-degenerate signature with arity at least three. It is quite difficult to construct a totally symmetric function with high arity except with some simple gadgets such as a star or a triangle. These gadgets work for some signatures but fail
for others. Due to this difficulty, we employ unsymmetric gadgets too. Fortunately, we also have a dichotomy for unsymmetric Holant$^*$ problems in
the Boolean domain~\cite{SODA11}. Since the dichotomy for this
more general Boolean Holant$^*$ is more complicated, we use a different proof strategy here. We only show the existence of  a non-degenerate signature with arity at least three, but do not analyze all possible forms case-by-case. We instead prove that we can always construct some binary signature in
addition to the higher arity one, which makes the problem hard no matter what the high arity signature is, provided that $\mathbf{F}$ is not one of the
tractable cases.

For a particular family of signatures which can be normalized to
the following form:
\begin{center}
{\scriptsize \begin{tabular}{*{11}{c}c}
{ }     & { }   & { }   &$0$& { }     & { }   & { }  \\
{ }     & { }   &$ix$& { } &$x$& { }     & { }  \\
{ }     &$0$& { } &  $0$& { }   &$0$& { }  \\
 $1$& { } &$i$& { } & $-1$& { } &   $-i$.  \\
\end{tabular}
}
\end{center}
where two isotropic vectors $(1, i)$ and $(1, -i)$ interact in an
unfavorable way,
we have to use a different argument (See the last case in Section \ref{Degenerate-Rank1-Isotropic}).
Due to its special structure, we have to use a different hard problem
to reduce from, namely the problem of counting perfect matchings on
3-regular graphs. This problem   is \#P-hard.
(This problem is tractable over
planar graphs by the FKT algorithm, the underlying algorithm for
matchgate based holographic algorithms~\cite{HA_J,AA_FOCS}.
This also indicates that the holographic reduction theory
developed here is distinct from that theory.)
Counting perfect matchings on
3-regular graphs as  a \#P-hard problem is also used
in Section~\ref{sec:0000}
when $\mathbf{F}^{*\rightarrow\{G,R\}} = [0,0,0,0]$ is identically 0.

\subsection{Realize a Rank 2 Binary Function}\label{sec:rank2}
\begin{theorem}\label{theorem6.1-get=_GR}
Let $\mathbf{F}$ be a symmetric ternary function over domain $\{B,G,R\}$.
 Then one of the following is true:
 \begin{enumerate}
   \item $\mathbf{F}$ is of one of the forms in Theorem \ref{thm:ternary}, and Holant$^*(\mathbf{F})$ is in P;
   \item Holant$^*(\mathbf{F})$ is \#P-hard;
   \item There exists an orthogonal $3 \times 3$ matrix $T$ such that Holant$^*(\mathbf{F})$
is polynomial time
 equivalent to
\\
 Holant$^*(\{ T^{\otimes 3}\mathbf{F}, =_{G,R} \})$.
 \end{enumerate}
 \end{theorem}

The proof of  Theorem~\ref{theorem6.1-get=_GR}
is completed in Sections~\ref{sec:rank2} and \ref{Interpolation-subsection}.
In Section~\ref{sec:rank2} we prove that either one of
 the first two alternatives
in Theorem~\ref{theorem6.1-get=_GR}
holds, or we can construct a rank 2 binary symmetric function
$f$ in Holant$^*(\mathbf{F})$, such that
 the matrix form of $f$ has
a non-isotropic  eigenvector corresponding to the eigenvalue $0$.
(The eigenspace has dimension 1, so the eigenvector is essentially
unique.)
In Section~\ref{Interpolation-subsection} we use $f$ to get $=_{G,R}$
by holographic reduction and interpolation.

In Lemma~\ref{lemma:just-rank2} we first get
a rank 2 binary symmetric function
$f$ in Holant$^*(\mathbf{F})$.

%
%
%
%

\begin{lemma}\label{lemma:just-rank2}
If $\mathbf{F}$ does not take one of the three forms in Theorem
\ref{thm:ternary}, then we can either prove that Holant$^*(\mathbf{F})$
is \#P-hard or construct a binary symmetric  function $f$
from $\mathbf{F}$ by connecting
a unary
function to it, such that (the matrix form of) $f$ has  rank 2.
\end{lemma}

\begin{proof}
By connecting $\mathbf{F}$ to a unary $\mathbf{u} = (x,y,z)$, we can
realize  $x \mathbf{F}^{1=B} + y \mathbf{F}^{1=G} +z \mathbf{F}^{1=R}$.
For notational simplicity, we denote the $3 \times 3$ matrices
$X=\mathbf{F}^{1=B}$, $Y=\mathbf{F}^{1=G}$ and $Z=\mathbf{F}^{1=R}$.
First suppose there exists a non-zero unary $\mathbf{u}$
such that $x X + y Y +z Z =0$.
If
$\mathbf{u}$ is isotropic, then $\mathbf{F}$ is in the third form of Theorem \ref{thm:ternary}.
Suppose $\mathbf{u}$ is  not isotropic, we may assume
$\mathbf{u}^{\tt T} \mathbf{u} = 1$. Then
 we can apply an orthogonal transformation by a matrix whose
first vector is $\mathbf{u}$, to reduce the
problem to an equivalent problem in domain size 2.
The dichotomy theorem for Holant$^*$ problems
over domain size 2 completes the proof.
The conclusion is that if $\mathbf{F}$ is not of the three
forms, then Holant$^*(\mathbf{F})$ is \#P-hard.
In the following, we assume that
$X$, $Y$ and $Z$ are linearly independent as  complex matrices.

Now we prove the lemma by analyzing the ranks of $X,Y,Z$.
By linear independence, $X,Y,Z$ all have rank $\ge 1$.
\begin{itemize}
  \item If at least one of $X,Y,Z$ has rank 2, then we are done by choosing the corresponding coefficient to be $1$ and the other two to be $0$.
  \item If there are at least two of them (we assume they are $X$ and $Y$)
 have rank 1, we shall prove that $X+Y$ has rank exactly 2.
Firstly, the rank of $X+Y$ is at most $2$ since  both $X$ and $Y$
have rank $1$. For symmetric matrices of rank 1,
we can write $X=u u^{\tt T}$ and $Y=v v^{\tt T}$.
       We know that $u$ and $v$ are linearly independent,
since   $X$ and $Y$ are linearly independent.
  If $X+Y$ has rank at most $1$, then there exists some $w$
such that
 $u u^{\tt T}+ v v^{\tt T}= w w^{\tt T}$.
There exists a vector $u'$ which is orthogonal to $u$ but not to $v$.
This can be seen by considering the dimensions of the null spaces
of  $u$ and $v$. Then $\langle u', v\rangle v
= \langle u', w\rangle w$. This implies that $v$ is a linear multiple of $w$
since $\langle u', v\rangle \neq 0$. Similarly, $u$ is also a
linear multiple of $w$. This contradicts the linear independence
of $u$ and $v$.
  \item In the remaining case, there are at least two of them (we assume they are $X$ and $Y$)
 have rank 3.  Then $\det (xX+Z)=0$ is not a trivial
equation since the coefficient of $x^3$ is $\det(X)\neq 0$.
Let $x_0$ be a root for the equation. Then the rank of $x_0 X+ Z$ is less than 3.
If the rank is 2, then we are done. Otherwise, the rank is exactly 1;
it cannot be zero since $Z$ is not a linear multiple of $X$.  Similarly,
 there exists a $y_0$ such that the rank of
the non-zero matrix $y_0 Y+ Z$ is less than 3.
Again, if the rank is 2, then we are done. Now we assume that both
 $x_0 X+ Z$ and $y_0 Y+ Z$ have rank 1.
If $x_0 X+ Z$ and $y_0 Y+ Z$ are linearly independent, then
$x_0 X+y_0 Y+ 2Z$ has rank exactly 2, by the proof above, and we are done.
If  $x_0 X+ Z$ and $y_0 Y+ Z$ are linearly dependent,
then a non-trivial combination is the zero matrix
$\lambda (x_0 X+ Z) + \mu (y_0 Y+ Z) = 0$.  Since they are both nonzero
matrices, both $\lambda, \mu \not = 0$.
Since $X, Y, Z$ are linearly independent,
we must  have $x_0 = y_0 =0$, and $Z$ has rank 1.
In this case, we consider $z X + Y$.
Again we have some $z_0$ such that $z_0 X + Y$ has rank at most 2.
If it is 2, we are done. It can't be 0, as $X, Y$ are linearly independent.
So $z_0 X + Y$ has rank exactly 1.
Then $z_0 X + Y + Z$ has rank exactly  2.
\end{itemize}

\end{proof}

\begin{lemma}\label{non-iso-rank2}
If we can realize a rank 2 binary symmetric
function in Holant$^*(\mathbf{F})$,
then we can either prove that
 $\mathbf{F}$ takes one of the forms in Theorem \ref{thm:ternary} and
Holant$^*(\mathbf{F})$ is in P, or realize a rank 2 binary
symmetric function such that its matrix form has a non-isotropic
eigenvector corresponding to the eigenvalue $0$.
\end{lemma}

\begin{proof}

We only need to handle the case that the matrix form
of the constructed rank 2 function
 has an isotropic eigenvector corresponding to $0$.

Suppose $A$ is the $3 \times 3$ matrix representing the
binary function $\langle  \mathbf{u}, \mathbf{F} \rangle$
for some unary function $\mathbf{u}$.
By the canonical form in~\cite{Scott},
there exists an
orthogonal matrix $T$, such that
\[
TA T^{\tt T}=\left [ \begin{matrix} 0 & 0 & 1 \\
   0 & 0 & i \\
   1 & i & 0
   \end{matrix} \right ] .
\]

We may
consider $T^{\otimes 3} \mathbf{F}$ instead of $\mathbf{F}$. Because
$TAT^{\tt T}$ is the matrix form for
$\langle T \mathbf{u}, T^{\otimes 3} \mathbf{F} \rangle$,
to reuse the notation, we can assume there exists a $\mathbf{u}$, such that
$\langle  \mathbf{u}, \mathbf{F} \rangle$
has the matrix form $\left [ \begin{matrix} 0 & 0 & 1 \\
   0 & 0 & i \\
   1 & i & 0
   \end{matrix} \right ]$.
We will rename this matrix $A$.

Given any unary function $\mathbf{v}$ and a complex number $x$, we can
realize the binary function
$\langle  x \mathbf{u} + \mathbf{v}, \mathbf{F} \rangle$
which has the matrix form $C = x A + \widetilde{A}$,
where $\widetilde{A}$ is the matrix form of
$\langle  \mathbf{v}, \mathbf{F} \rangle$.
If there exist some unary function $\mathbf{v}$
and  a complex number
$x$,
such that $C$ is nonsingular, and $\gamma= C^{-1} \left [ \begin{matrix} 1 \\
  i \\   0    \end{matrix} \right ]$ is not isotropic, then we can
realize the binary symmetric function $CAC$ of rank 2 as a chain  of
three binary symmetric functions, whose
eigenvector corresponding to $0$ is $\gamma$, and the conclusion
holds.

Now, we prove that if there does not exist such $\mathbf{v}$ and $x$,
then either Holant$^*({\mathbf{F}})$ is in P, or we can realize a required
binary function directly. We calculate the two conditions,
$C$ is singular and $\gamma= C^{-1} \left [ \begin{matrix} 1 \\
  i \\   0    \end{matrix} \right ]$ is isotropic, individually.

Suppose $\widetilde{A}=\left [ \begin{matrix} a & b & c \\
   b & d & e \\
   c & e & f
   \end{matrix} \right ]$.
Then $C = x A + \widetilde{A} =
\left [ \begin{matrix} a & b & c +x \\
   b & d & e +xi \\
   c+x  & e +xi & f
   \end{matrix} \right ]$.
Let $P(x) = \det(C)$.
As a polynomial in $x$, $P(x)$ has degree  at most
2, and the
coefficient of $x^2$ is $a + 2bi-d$. If $a + 2bi-d \neq 0$, then
for all complex $x$ except at most two values,
$C$ is nonsingular.

Because $C \gamma=(1, i, 0)^{\tt T}$, $\gamma$ is orthogonal to
$\mu=(c+x,e+xi,f)$ and $\nu=(b-ai, d-bi, e-ci)$.
Consider the
 cross-product vector $\theta=\left( \left | \begin{matrix}  e+xi & f \\
d-bi & e-ci \end{matrix} \right | ,
\left | \begin{matrix}   f & c+x \\
e-ci & b-ai \end{matrix} \right |, \left | \begin{matrix} c+x & e+xi
\\ b-ai & d-bi
\end{matrix}\right |  \right)^{\tt T}$, which is orthogonal to $\mu$ and
$\nu$. Calculation shows that
the inner product $\theta^{\tt T}  \theta$ is a
polynomial $Q(x)$ of degree at most 2,
and the coefficient of $x^2$ is $(a + 2bi-d)^2$.

Assume $a + 2bi-d \neq 0$. Then, neither $P(x)$ nor $Q(x)$ is the zero
polynomial. There exists an $x$ such that $C$ is nonsingular,
which implies $\gamma \not = {\bf 0}$ in particular, and
$\theta^{\tt T} \theta \neq 0$. If $\mu$ and $\nu$ were linearly dependent,
then $\theta=\mathbf{0}$ by the definition  of
$\theta$, and $\theta^{\tt T} \theta =0$,
a contradiction. Hence,   $\mu$
and $\nu$ are linearly independent. So $\gamma$ is a nonzero linear multiple of
$\theta$, since they both belong to the 1-dimensional
subspace orthogonal to $\mu$
and $\nu$. Then
$\gamma^{\tt T} \gamma$ is a nonzero multiple of
$\theta^{\tt T} \theta \neq 0$, i.e.,  $\gamma$ is not isotropic.
Then $CAC$ is the required function.

Now we assume that for any $\mathbf{v}$, $\widetilde{A}=
\langle \mathbf{v}, \mathbf{F} \rangle$ satisfies $a + 2bi-d =
0$.

Substitute  $d$ by $a + 2bi$, we get
$P(x)=2(b-ai)(e-ci)x-a(e-ci)^2-f(b-ai)^2 + 2c(b-ai)(e -ci)$,
and the coefficient of $x$ in
$Q(x)$ is $2i(e-ci)^3$.

For any fixed $\widetilde{A}$,  either $e-ci =0$, or $e-ci \neq 0$. If $e-ci
\neq 0$, $Q(x)$ is not the zero polynomial. If $P(x)$ is not the zero
polynomial as well, then by the same argument as above, we get a required
function. Hence we assume $P(x)$ is the zero polynomial.
Then by the expression for $P(x)$, it follows that  $b-ai =0$, and
$a=0$. Because we also have  $a+ 2bi-d=0$, we get $a=b=d=0$.

In this case $\widetilde{A}$ has the form
$\widetilde{A}=\left [ \begin{matrix} 0 & 0 & c \\
   0 & 0 & e \\
   c & e & f
   \end{matrix} \right ]$.
It has rank $\le 2$. If it has rank $\le 1$, then $c=e=0$. This is
a contradiction to $e-ci \neq 0$.
Hence it has rank 2.
It is easy to check that the eigenvector
corresponding to the eigenvalue 0 is a multiple of
$(-e, c, 0)^{\tt T}$. If $c^2 + e^2 \not = 0$, then this
eigenvector is
non-isotropic and we are done.
Since $e-ci \neq 0$, the only possibility of $c^2 + e^2 = 0$
is $e = -ci \not = 0$. In this case
it is easy to check that $cA +\widetilde{A}$
has the form
$\left [ \begin{matrix} 0 & 0 & 2c \\
   0 & 0 & 0 \\
   2c & 0 & f
   \end{matrix} \right ]$.
It has rank 2, and a non-isotropic  eigenvector
$(0, 1, 0)^{\tt T}$
corresponding to the eigenvalue 0.

Finally we have for any
$\widetilde{A}$,  $e-ci =0$, in addition to $d = a + 2b i$.

Consider the possible choices of $\bf v$ in
$\widetilde{A}= \langle \mathbf{v}, \mathbf{F} \rangle$. We can set it
to be $\mathbf{F}^{1=B}$,
$\mathbf{F}^{1=G}$ or $\mathbf{F}^{1=R}$.
Considering what entries $a, b, c, d, e$ correspond to
in the table (\ref{domain-3-sig}) for these three cases of $\widetilde{A}$,
we get the following: If $w \neq 0$, then
$\mathbf{F}_{u,v,w}=i \mathbf{F}_{u+1,v-1,w}$ for $v \ge 1$ and
$u+v+w = 3$. If $w=0$, then $\mathbf{F}_{u,v,w} = \mathbf{F}_{u,v,0}
=s i^v + t v
i^{v-1}$ for some coefficients $s$ and $t$,
where $u, v \ge 0$ and $u+v = 3$. This follows from
$e = ci$ and
$d = a +2bi$ for $\widetilde{A}$.  E.g., $e = ci$ in
(\ref{domain-3-sig-1-is-set-to-B}) gives a linear recurrence
$F_{BGR} = i F_{BBR}$, and $d = a +2bi$ in
(\ref{domain-3-sig-1-is-set-to-B}) gives a linear recurrence
$F_{BGG} = 2i F_{BBG} + F_{BBB}$.
Hence, $\mathbf{F}=S+T$ is the summation of two functions $S$ and
$T$, where $S_{u,v,w}=i S_{u+1,v-1,w}$, and $T(u,v,w)=0$, if $w \neq
0$, and $T(u,v,0)=t v i^{v-1}$, where $u+v+w = 3$.
This $T$ can be expressed as the symmetrization
of simple tensor products,
\begin{eqnarray*}
T&=& T_1+T_2+T_3\\
&=& t\left [ \begin{matrix} 0 \\
  1 \\ 0 \end{matrix} \right ] \otimes \left [ \begin{matrix} 1 \\
   i \\ 0 \end{matrix} \right ] \otimes \left [ \begin{matrix} 1 \\
   i \\ 0 \end{matrix} \right ]+t \left [ \begin{matrix} 1 \\
   i \\ 0 \end{matrix} \right ] \otimes \left [ \begin{matrix} 0 \\
   1 \\ 0 \end{matrix} \right ] \otimes \left [ \begin{matrix} 1 \\
   i \\ 0 \end{matrix} \right ]+t \left [ \begin{matrix} 1 \\
   i \\ 0 \end{matrix} \right ] \otimes \left [ \begin{matrix} 1 \\
   i \\ 0 \end{matrix} \right ] \otimes \left [ \begin{matrix} 0 \\
   1 \\ 0 \end{matrix} \right ]\\
&=&
   \frac {t}{2} {\rm Sym} (\left [ \begin{matrix} 1 \\
   i \\ 0 \end{matrix} \right ] \otimes \left [ \begin{matrix} 1 \\
   i \\ 0 \end{matrix} \right ] \otimes \left [ \begin{matrix} 0 \\
   1 \\ 0 \end{matrix} \right ] ).
\end{eqnarray*}

This is in form 3 given in Theorem~\ref{thm:ternary}
and we have shown that in this case Holant$^*(\mathbf{F})$
is tractable in Section~\ref{section:tractability}.

\end{proof}

We summarize Lemma~\ref{lemma:just-rank2} and \ref{non-iso-rank2} as follows:
\begin{corollary}
If $\mathbf{F}$ does not take one of the three forms in Theorem
\ref{thm:ternary}, then we can either prove that Holant$^*(\mathbf{F})$
is \#P-hard or construct a rank 2
 binary symmetric function $f$ from $\mathbf{F}$ by connecting
a unary
function to it, such that its
eigenvector corresponding to the eigenvalue $0$ is not isotropic.
\end{corollary}

\subsection{An Interpolation Lemma}\label{Interpolation-subsection}
Finally we use a holographic transformation and interpolation
to get $=_{G,R}$ from the binary function obtained in Lemma~\ref{non-iso-rank2}.
This will complete the proof of Theorem~\ref{theorem6.1-get=_GR}.

Let $\mathbf{v}$ be a non-isotropic eigenvector corresponding to the
eigenvalue $0$ of the binary function $A$ constructed from $\mathbf{F}$.
We may assume $\langle \mathbf{v}, \mathbf{v} \rangle =1$.
We can extend $\mathbf{v}$ to an orthogonal matrix $T$, such that $\mathbf{v}$ is
the first column vector of $T$.
Then the matrix form of the binary function
 after the holographic
transformation by $T^{-1} = T^{\tt T}$  takes the form
\begin{equation}\label{rank-2-form}
T^{\tt T}A T=\left [ \begin{matrix} 0 & 0 & 0 \\
   0 & a & b \\
   0 & b & c
   \end{matrix} \right ]
\end{equation}
with rank 2.

The next lemma shows that given this, we can interpolate
$=_{G,R}$.

%
\begin{lemma}\label{lem:interpolatoin}
Let $H:\{B, G, R\}^2\rightarrow \mathbb{C}$ be a rank 2 binary
function of the form (\ref{rank-2-form}).
Then for any ${\cal F}$ containing $H$,
we have
\[{\mbox {\rm Holant}}({\cal F}\cup \{=_{G,R}\})\leq_T
{\mbox {\rm Holant}}({\cal F}) .\]
\end{lemma}

\begin{proof}
Consider the Jordan normal form of $H$.
There are two cases:
  there exist a non-singular $M = {\rm diag}(1, M_2)$, and either
 $\Lambda={\scriptsize \left [
\begin{array}{ccc}
0 & 0 & 0\\
0 & \lambda & 0\\
0 & 0 & \mu
\end{array}\right ]}$,
or
$\Lambda' ={\scriptsize \left [
\begin{array}{ccc}
0 & 0 & 0\\
0 & \lambda & 1\\
0 & 0 & \lambda
\end{array}\right ]}$,
such that $H=M\Lambda M^{-1}$,
or $H=M\Lambda' M^{-1}$.

For the first case $H=M\Lambda M^{-1}$,
consider an instance $I$ of
Holant$(\mathcal{F}\cup \{=_{G,R}\})$. Suppose the function $=_{G,R}$ appears
$m$ times. Replace each occurrence of $=_{G,R}$ by a chain of $M$, $=_{G,R}$,
$M^{-1}$. More precisely,
we replace any occurrence of $=_{G,R}(x,y)$ by
$M(x,z) \cdot (=_{G,R})(z, w) \cdot M^{-1}(w,y)$,
where $z, w$ are new variables.
This defines a new instance $I'$.
Since $M {\rm diag}(0, I_2) M^{-1}={\rm diag}(0, I_2)$, where $I_2$ denotes the
$2 \times 2$ identity matrix,  the Holant value of
the instance $I$ and $I'$ are the same.
To have a non-zero contribution to the Holant sum,
the assignments given to any occurrence of the new {\sc Equality} constraints
of the form $(=_{G,R})(z, w)$ must be $(G,G)$ or $(R,R)$.
 We can stratify the Holant sum
defining the value on $I'$ according to how many $(G,G)$
and $(R,R)$  assignments are given to these
occurrences of $(=_{G,R})(z, w)$.
Let $\rho_j$ denote the sum,
over all assignments with $j$ many times  $(G,G)$ and
$m-j$ many times  $(R,R)$, of the evaluation on $I'$,
including those of $M(x,z)$ and $M^{-1}(w,y)$.
Then the Holant value on the instance $I'$
 can be written as $\sum_{j=0}^m \rho_j$.

Now we construct from $I$ a sequence of instances $I'_k$ indexed by $k$:
Replace each occurrence of $(=_{G,R})(x, y)$  by
a chain of $k$ copies of the function $H$ to get an
instance  $I'_k$ of Holant$(\mathcal{F})$.
More precisely, each occurrence of
$(=_{G,R})(x,y)$ is replaced by $H(x,x_1)H(x_1,x_2)\ldots H(x_{k-1},y)$,
where $x_1, x_2, \ldots, x_{k-1}$ are new variables specific for  this
occurrence of $(=_{G,R})(x,y)$.
The function of this chain is $H^k=M\Lambda^k M^{-1}$.
A moment of reflection shows that the value of
the instance $I'_k$ is
\[\sum_{j=0}^m \rho_j \lambda^{kj}
\mu^{k(m-j)}=\mu^{mk} \sum_{j=0}^m \rho_j (\lambda/
\mu)^{kj}.\]

If $\lambda/\mu$ is a root of unity, then take a $k$ such that
$(\lambda/\mu)^k=1$. (Input size is measured by
the number of variables and constraints. The functions in $\mathcal{F}$
are considered constants. Thus this $k$ is a constant.) We have
the value
$\sum_{j=0}^m \rho_j \lambda^{kj} \mu^{k(m-j)}
=\mu^{mk}\sum_{j=0}^m \rho_j$.
As $H$ has rank 2, $\mu \not = 0$,  we can compute the
value of $I$ from the value of  $I'_k$.

If $\lambda/\mu$ is not a root of unity,
$(\lambda/\mu)^j$ are all distinct for $j \ge 1$. We can take
$k=1,\ldots,m+1$ and get a system of linear equations about $\rho_j$.
Because the coefficient matrix is Vandermonde in
$(\lambda/\mu)^j, j=0,1,\ldots m$, we can solve
$\rho_j$ and get the value of $I$.

For the second case $H=M\Lambda' M^{-1}$,
the construction is the same, so we only show
the difference with the proof in the first case. Again we can
stratify the Holant sum for $I'$ according to how many different
types of assignments are given to the
$m$ occurrences of the new {\sc Equality} constraints of the form
$(=_{G,R})(z, w)$.
Any assignment other than assigning only $(G,G)$ or $(R,R)$
will produce a 0 contribution for $I'$.
However,
this time we cluster all assignments according to exactly $j$ many
times $(G,G)$ {\it or} $(R,R)$, and the rest $m-j$  are $(G,R)$'s,
on all $m$ occurrences of these $(=_{G,R})(z, w)$.
Note that any assignment with a non-zero number of $(R,G)$'s
in the corresponding $m$ signatures in $I'_k$, after the
substitution of each $(=_{G,R})(x,y)$  in $I$ by $H(x,x_1)H(x_1,x_2)\ldots
H(x_{k-1},y)$, will
produce a 0 contribution in the Holant value for $I'_k$.
 This is because, by this substitution, effectively
each $(=_{G,R})(z, w)$ in $I'$ is replaced by
 $\Lambda^k= {\scriptsize \left [
\begin{array}{ccc}
0 & 0 & 0 \\
0 & \lambda^k & k \lambda^{k-1}\\
0 & 0 & \lambda^k
\end{array}\right ]}$.
Let $\rho_j$ be  the sum over all
assignments with $j$ many
$(G,G)$ {\it or} $(R,R)$, and $m-j$ many $(G,R)$
 of the evaluation (including those of $M(x,z)$ and $M^{-1}(w,y)$)
on $I'$. Then the Holant value on the instance $I'$  (and on $I$)
 is just  $\rho_m$.


The value of $I'_k$  is
\[\sum_{j=0}^m \rho_j \lambda^{kj}
(k\lambda^{k-1})^{m-j}=\lambda^{(k-1)m} \sum_{j=0}^m (\lambda^j \rho_j)
k^{m-j}.\]

We can take $k=1,\ldots,m+1$ and get a system of linear equations
on $\lambda^j \rho_j$. Because the coefficient matrix is a Vandermonde
matrix, we can solve $\lambda^j \rho_j$ and (since $\lambda \not = 0$
as $H$ has rank 2) we can get the value of
$\rho_m$, which is the value of $I$.
\end{proof}

\subsection{Reductions From Domain Size 2}\label{sec:domain2}

\begin{lemma}\label{lemma:separate}
If a ternary function $\mathbf{F}$ has a separated domain then
Holant$^*( \mathbf{F})$ is either \#P-hard or is in one of
the tractable forms of Theorem \ref{thm:ternary},
and it is determined by the Holant$^*$ problem defined
by the restriction of $\mathbf{F}$
to the separated subdomain of size two.
\end{lemma}

\begin{proof}
Suppose $B$ is separated from $G$-$R$ in $\mathbf{F}$.
Given any connected signature grid for
Holant$^*(\mathbf{F})$, any assignment of $B$ will be
uniquely propagated as $B$. Hence the tractability or \#P-hardness
of the problem is
determined by the Holant$^*$ problem defined
by $\mathbf{F}$ restricted to $\{G, R\}$. Then
the dichotomy Theorem \ref{thm:dich-sym-Boolean}
shows that Holant$^*( \mathbf{F})$ is either \#P-hard or is in one of
the tractable forms of Theorem \ref{thm:ternary}.
More specifically,
a degenerate signature or a  generalized
Fibonacci gate ($a x_{k+2} - b x_{k+1} - a x_k = 0$)
 on $\{G, R\}$  with
$b \not = \pm 2i a$ lead to form {\it 1}. A Fibonacci gate with
$b = \pm 2i a$ leads to form {\it 3}, where we take $\mathbf{F}_{\beta}
= F_{BBB} \mathbf{e_1}^{\otimes 3}$. Finally the tractable
form $[x, y, -x, -y]$
for $\mathbf{F}^{*\rightarrow\{G,R\}}$
leads to form {\it 2}.
\end{proof}

\begin{theorem}\label{thm:F-with-=_GR}
Let $\mathbf{F}$ be a symmetric ternary function over domain $\{B,G,R\}$, which is not of one of the forms in Theorem \ref{thm:ternary}.
Then Holant$^*(\{ \mathbf{F}, =_{G,R} \})$ is \#P-hard.
 \end{theorem}

Theorem~\ref{theorem6.1-get=_GR} and \ref{thm:F-with-=_GR} imply our main
Theorem~\ref{thm:ternary}. The rest of this paper is devoted to
the proof of Theorem~\ref{thm:F-with-=_GR}.

Using $=_{G,R}$
we can realize signatures over domain $\{G,R\}$ from $\mathbf{F}$
such as $\mathbf{F}^{*\rightarrow\{G,R\}}$.
If Holant$^*(\mathbf{F}^{*\rightarrow\{G,R\}})$ is
already \#P-hard as a problem over size 2 domain $\{G,R\}$,
then Holant$^*(\{ \mathbf{F}, =_{G,R} \})$ is \#P-hard and we are done.
Therefore, we only need
 to deal with the cases when Holant$^*(\mathbf{F}^{*\rightarrow\{G,R\}})$
is tractable. They are listed as follows.

 \begin{enumerate}
   \item  $\mathbf{F}^{*\rightarrow\{G,R\}}=H [a, 0, 0, b]^{\tt T}$, where $H$ is a $2 \times 2$ orthogonal matrix, $a b \neq 0$.
   \item  $\mathbf{F}^{*\rightarrow\{G,R\}}=Z [a, 0, 0, b]^{\tt T}$,
where $Z=\frac{1}{\sqrt{2}}\begin{bmatrix} 1 & 1 \\ i & -i \end{bmatrix}$, $a b \neq 0$.
   \item $\mathbf{F}^{*\rightarrow\{G,R\}}=Z [a, b, 0, 0]^{\tt T}$,
where $Z=\frac{1}{\sqrt{2}}\begin{bmatrix} 1 & 1 \\ i & -i \end{bmatrix}$ or  $Z=\frac{1}{\sqrt{2}}\begin{bmatrix} 1 & 1 \\ -i & i \end{bmatrix}$, $b\neq 0$.
   \item $\mathbf{F}^{*\rightarrow\{G,R\}}$ is degenerate.
 \end{enumerate}

 We will prove Theorem~\ref{thm:F-with-=_GR}
by considering  these four cases one by one.
The overall proof approach for the first three cases
 is to construct a binary function over the domain $\{G,R\}$
such that, together with $\mathbf{F}^{*\rightarrow\{G,R\}}$
it is already \#P-hard according to the dichotomy theorem
 for Holant$^*$ over domain size 2, Theorem \ref{thm:dich-sym-Boolean}.
For some functions $\mathbf{F}$, we fail to do this;
 and  whenever this happens,
we show that $\mathbf{F}$ is indeed among
the tractable cases in Theorem \ref{thm:ternary}.
For the fourth case, where $\mathbf{F}^{*\rightarrow\{G,R\}}$
is degenerate  on $\{G,R\}$, our approach is different,
where we need to construct gadgets with a larger arity,
and will be dealt with in later subsections.

\vspace{.1 in}

\noindent {\bf Case 1: $\mathbf{F}^{*\rightarrow\{G,R\}}=
H [a, 0, 0, b]^{\tt T}, ~~~ab \not =0$.}

After a \methsepahr under the
orthogonal matrix $\begin{bmatrix} 1 & \mathbf{0} \\ \mathbf{0} & H \end{bmatrix}$, we can assume that
$\mathbf{F}^{*\rightarrow\{G,R\}}=[a, 0, 0, b]$,
where we are given $ab \not =0$. We note that this transformation does not change $=_{G,R}$.
According to Theorem \ref{thm:dich-sym-Boolean},
when putting this $[a, 0, 0, b]$ and a binary
function together, the problem is \#P-hard unless the binary
function is of the form $[*, 0, *]$, $[0,*,0]$
or degenerate.
Now $\mathbf{F}$ has the form
\begin{center}
\begin{tabular}{*{11}{c}c}
{ }     & { }   & { }   &$F_{BBB}$& { }     & { }   & { }  \\
{ }     & { }   &$F_{BBG}$& { } &$F_{BBR}$& { }     & { }  \\
{ }     &$F_{BGG}$& { } &  $F_{BGR}$& { }   &$F_{BRR}$& { }  \\
 $a$& { } &$0$& { } & $0$& { } &   $b$  \\
\end{tabular}
\end{center}

Suppose $F_{BGR}\neq 0$.
We can realized a binary function $[F_{BGG}+a t, F_{BGR}, F_{BRR}]$ over domain $\{G,R\}$ by connecting this ternary
function to
a unary function $(1,t,0)$,
namely $\langle (1,t,0), \mathbf{F} \rangle$,
 and then putting $=_{G,R}$ on the other two dangling edges.
Since $a\neq 0$ and we can choose any $t$, we can make the first entry of
$[F_{BGG}+a t, F_{BGR}, F_{BRR}]$ arbitrary  and
 the function is out of all three tractable binary forms. Therefore the problem is \#P-hard.

Now we can assume
that $F_{BGR}=0$. To simplify notations,
we use variables to denote the function entries as follows
\begin{equation}\label{variable-form-for-F-center-all-0}
\begin{tabular}{*{11}{c}c}
{ }     & { }   & { }   &$g$& { }     & { }   & { }  \\
{ }     & { }   &$y$& { } &$w$& { }     & { }  \\
{ }     &$x$& { } &  $0$& { }   &$z$& { }  \\
 $a$& { } &$0$& { } & $0$& { } &   $b$  \\
\end{tabular}
\end{equation}

Then we use the gadget as depicted in Figure \ref{binary-1-2}   to construct another binary function.
\begin{figure}[hbtp]
	\begin{center}
		\includegraphics[width=3 in]{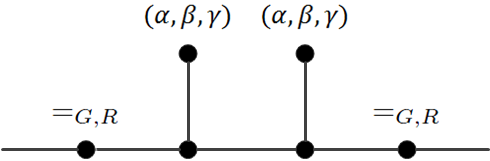}
	\caption{A binary gadget.}
	\label{binary-1-2}
	\end{center}
\end{figure}
The signature of this binary function has been calculated in
Section~\ref{section:calculus}
(see (\ref{sig-2-weed-gadget})), and is
\[[f_0,f_1,f_2]=[(\alpha x +\beta a)^2+(\alpha y + \beta x)^2, (\alpha y + \beta x)(\alpha w + \gamma z), (\alpha z + \gamma b)^2+(\alpha w + \gamma z)^2].\]
If there exists
some $(\alpha, \beta, \gamma)$ such that this  $[f_0,f_1,f_2]$ is not of
the form $[*, 0, *]$, $[0,*,0]$, or degenerate, then the problem is \#P-hard
and we are done.

All conditions are polynomials (1) $f_0=f_2=0$, or (2) $f_1=0$, or (3) $f_1^2=f_0 f_2$.
By \methpoly, we only need to deal with cases that one of them is the zero
polynomial.





If statement (1) $f_0=f_2=0$ holds for all
$(\alpha, \beta, \gamma)$, we have
\[ (x^2+y^2) \alpha^2 + 2(a x+ x y) \alpha \beta + (a^2+x^2) \beta^2 =  (z^2+w^2) \alpha^2 + 2(b z+ z w) \alpha \gamma + (b^2+z^2) \gamma^2  =0,\]
as identically zero polynomials in $(\alpha, \beta, \gamma)$.
Therefore we have
\[ x^2+y^2 =a x+ x y=a^2+x^2=z^2+w^2=b z+ z w=b^2+z^2=0.\]
Since $a  \neq 0$, we have $x  \neq 0$ from $a^2+x^2=0$. Similarly, we have $z \neq 0$. Then the conclusion is $x=\epsilon_1 a, ~~ y=-a,~~ z=\epsilon_2 b,~~ w=-b$,
where $\epsilon_1, \epsilon_2 \in \{i, -i\}$.  Then we rewrite our function as follows
\begin{center}
\begin{tabular}{*{11}{c}c}
{ }     & { }   & { }   &$g$& { }     & { }   & { }  \\
{ }     & { }   &$-a$& { } &$-b$& { }     & { }  \\
{ }     &$\epsilon_1 a$& { } &  $0$& { }   &$\epsilon_2 b$& { }  \\
 $a$& { } &$0$& { } & $0$& { } &   $b$  \\
\end{tabular}
\end{center}

\begin{figure}[hbtp]
	\begin{center}
		\includegraphics[width=3 in]{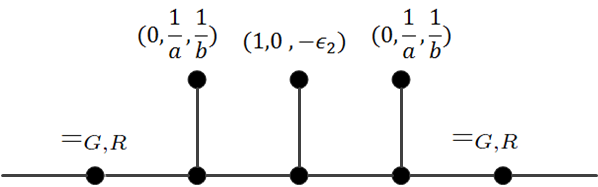}
	\caption{A binary gadget.}
	\label{fig:binary-2}
	\end{center}
\end{figure}
Next
 we use the gadget depicted
 in Figure \ref{fig:binary-2} to construct another binary function over
domain $\{G,R\}$, whose signature is calculated with the
techniques of Section~\ref{section:calculus}
\[\begin{bmatrix} \epsilon_1 & 1 & 0 \\ \epsilon_2 & 0 & 1 \end{bmatrix}
\begin{bmatrix} g+ \epsilon_2 b & -a & 0 \\ -a & \epsilon_1 a & 0 \\ 0 & 0 & 0 \end{bmatrix}
\begin{bmatrix} \epsilon_1 & \epsilon_2  \\ 1 & 0 \\ 0 & 1 \end{bmatrix}
=\begin{bmatrix} - g -\epsilon_1 a -\epsilon_2 b  &
 \epsilon_1 \epsilon_2 (g + \epsilon_1 a + \epsilon_2 b)  \\
\epsilon_1 \epsilon_2 (g + \epsilon_1 a + \epsilon_2 b)
& -g-\epsilon_2 b  \end{bmatrix}.\]
If $g  +\epsilon_1 a +\epsilon_2 b \neq 0$, this symmetric binary signature can not be of the
form $[*, 0, *]$ or $[0,*,0]$, and it is not degenerate
as its determinant is nonzero.
Therefore the problem
is \#P-hard.

If  $g + \epsilon_1 a + \epsilon_2 b = 0$,
we show that this is indeed a tractable case in Theorem \ref{thm:ternary}.
It is of the second form in Theorem \ref{thm:ternary}
where $\mbox{\boldmath $\alpha$} =(0,0,0)^{\tt T},
\mbox{\boldmath $\beta_1$}=\sqrt[3]{a}(\epsilon_1, 1, 0)^{\tt T}$
and $\mbox{\boldmath $\beta_2$}=\sqrt[3]{b}(\epsilon_2, 0, 1)^{\tt T}$.

If statement (2) $f_1=0$ holds for all $(\alpha, \beta, \gamma)$, we have
 $x=y=0$ or $z=w=0$. If $x=y=0$, the ternary
function (\ref{variable-form-for-F-center-all-0}) is as follows
\begin{center}
\begin{tabular}{*{11}{c}c}
{ }     & { }   & { }   &$g$& { }     & { }   & { }  \\
{ }     & { }   &$0$& { } &$w$& { }     & { }  \\
{ }     &$0$& { } &  $0$& { }   &$z$& { }  \\
 $a$& { } &$0$& { } & $0$& { } &   $b$  \\
\end{tabular}
\end{center}
Then $G$ is separated from $B$-$R$,
and by Lemma \ref{lemma:separate}, we are done.
The case $z=w=0$ is similar.

If statement (3) $f_1^2=f_0 f_2$ holds for all $(\alpha, \beta, \gamma)$,
we have
\begin{equation}\label{eqn:f1^2=f0f2}
(\alpha x +\beta a)^2 (\alpha z + \gamma b)^2 + (\alpha x +\beta a)^2 (\alpha w + \gamma z)^2 +(\alpha y + \beta x)^2 (\alpha z + \gamma b)^2  = 0.
\end{equation}
Let $\alpha=a$ and $\beta=-x$, we have $(a y - x^2)^2 (a z + \gamma b)^2=0$ holds for all $\gamma$. Since $b\neq 0$, we can choose $\gamma$ such that $a z + \gamma b \neq 0$
and conclude that $a y - x^2=0$. Similarly, let  $\alpha=b$ and $\gamma=-z$, we can get $bw -z^2=0$. Then let $\beta=\gamma=1$ and $\alpha=0$ in (\ref{eqn:f1^2=f0f2}), we have
\[ a^2b^2 +a^2z^2+b^2x^2=0.\]
Denote by $p=\frac{x}{a}$ and $q=\frac{z}{b}$, we have $p^2+q^2+1=0$ and the ternary signature in (\ref{variable-form-for-F-center-all-0}) has the following
form
\begin{center}
\begin{tabular}{*{11}{c}c}
{ }     & { }   & { }   &$g$& { }     & { }   & { }  \\
{ }     & { }   &$a p^2$& { } &$b q^2$& { }     & { }  \\
{ }     &$a p$& { } &  $0$& { }   &$ b q$& { }  \\
 $a$& { } &$0$& { } & $0$& { } &   $b$  \\
\end{tabular}
\end{center}
If $p=0$ or $q=0$, then the function is separable and we are done by Lemma \ref{lemma:separate}. In the
following, we assume that $pq\neq 0$.

\begin{figure}[hbtp]
	\begin{center}
		\includegraphics[width=3 in]{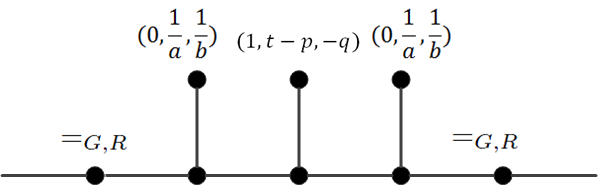}
	\caption{A binary gadget.}
	\label{fig:binary-3}
	\end{center}
\end{figure}
Then we use the gadget  in Figure \ref{fig:binary-3} to construct another binary function over domain $\{G,R\}$, whose signature is
\[\begin{bmatrix} p & 1 & 0 \\ q & 0 & 1 \end{bmatrix}
\begin{bmatrix} g-  b q^3-a p^3 + a p^2 t & a p t & 0 \\ a p t &  a t & 0 \\ 0 & 0 & 0 \end{bmatrix}
\begin{bmatrix} p & q  \\ 1 & 0 \\ 0 & 1 \end{bmatrix}=\begin{bmatrix}p^2 \delta + at (p^2 +1)^2	& 	pq \delta + a pqt (p^2 +1) \\
pq \delta + a pqt (p^2 +1)	&	q^2 \delta + a p^2 q^2 t
    \end{bmatrix} ,\]
where $\delta=g-a p^3 - b q^3$. We denote this symmetric binary function as $[g_0, g_1, g_2]$.

If $\delta=0$, one
can verify  that this is indeed a tractable case of Theorem \ref{thm:ternary}.
 This is of the
third form of Theorem \ref{thm:ternary}, where
$\mbox{\boldmath $\beta$}=(1,-p,-q)^{\tt T}, \mbox{\boldmath $\gamma$}
=(0,0,0)^{\tt T}$,
and $\mathbf{F}_\beta$ is the given function $\mathbf{F}$.

Now we assume that $\delta \neq 0$. If there exists
 some $t$ such that this binary function is not of
the form $[*, 0, *]$, $[0,*,0]$, or degenerate, then the problem is \#P-hard and we are done.
Otherwise, by the same argument as above,  at least one of the three
statements (i) $g_0=g_2=0$, (ii) $g_1=0$, or (iii) $g_1^2=g_0 g_2$
holds for all $t$. Choose $t=0$, we have all three
$g_0, g_1, g_2 \neq 0$. Therefore, the
only possibility is  that $g_1^2=g_0 g_2$ holds for all $t$. However, this is also impossible which can be seen by choosing $t=\frac{1}{a}$.
One can calculate the determinant
$\det \begin{bmatrix} g_0 & g_1 \\ g_1 & g_2  \end{bmatrix} =
\delta q^2 \not =0$.
This completes the proof for the case  $\mathbf{F}^{*{}=\{G,R\}}=H [a, 0, 0, b]^{\tt T}$.

\vspace{.1 in}

\noindent {\bf Case 2: $\mathbf{F}^{*\rightarrow\{G,R\}}
=Z [a, 0, 0, b]^{\tt T},
~~~ab \not = 0$.}

The  problem
Holant$^*(\{\mathbf{F},=_{G,R}\})$
 can be written as Holant$^*(=_2 | \{\mathbf{F},=_{G,R}\})$, where $*$ means that both sides
can use all unary functions.
After  a holographic transformation under the matrix $\tilde{Z}=\begin{bmatrix} 1 & \mathbf{0} \\ \mathbf{0} & Z \end{bmatrix}$, we can get an
equivalent problem  Holant$^*(\neq_{B;G,R} | \{\tilde{Z}^{-1}\mathbf{F},\neq_{G,R}\})$, where the two binary
functions are, respectively,
\begin{equation}\label{neq-BGR-def}
(\neq_{B;G,R}) = \tilde{Z}^{\tt T} (=_2)
\tilde{Z} =
 \begin{bmatrix} 1 & 0 & 0 \\ 0 & 0 & 1 \\ 0 & 1 & 0 \end{bmatrix},~~~
\mbox{and}~~~  (\neq_{G,R}) =
\tilde{Z}^{-1} (=_{G,R}) (\tilde{Z}^{-1})^{\tt T}
= \begin{bmatrix} 0 & 0 & 0 \\ 0 & 0 & 1 \\ 0 & 1 & 0 \end{bmatrix}.
\end{equation}
We use $\mathbf{\tilde{F}}$ to denote the ternary function
$\tilde{Z}^{-1}\mathbf{F}$ after
the transformation.
Then we have  $\mathbf{\tilde{F}}^{*\rightarrow\{G,R\}}= [a, 0, 0, b]$.
By connecting $\neq_{B;G,R}$ to both sides of $\neq_{G,R}$, we can get
the function  $\neq_{G,R}$ on the LHS.
For a bipartite holant problem
Holant$^*([f_0,f_1, f_2]| [a, 0, 0, b])$
 over domain size 2,
the problem is \#P-hard unless the binary function $[f_0,f_1, f_2]$ is of
the form $[*, 0, *]$, $[0,*,0]$, or degenerate~\cite{holant}.
Therefore, we will try to construct binary functions
 in the LHS of Holant$^*(\{\neq_{B;G,R}, \neq_{G,R}\}
| \{\mathbf{\tilde{F}}, \neq_{G,R}\})$
over domain $\{G,R\}$.

Our ternary function $\mathbf{\tilde{F}}$ is as follows
\begin{center}
\begin{tabular}{*{11}{c}c}
{ }     & { }   & { }   &$\tilde{F}_{BBB}$& { }     & { }   & { }  \\
{ }     & { }   &$\tilde{F}_{BBG}$& { } &$\tilde{F}_{BBR}$& { }     & { }  \\
{ }     &$\tilde{F}_{BGG}$& { } &  $\tilde{F}_{BGR}$& { }   &$\tilde{F}_{BRR}$& { }  \\
 $a$& { } &$0$& { } & $0$& { } &   $b$  \\
\end{tabular}
\end{center}

If $\tilde{F}_{BGR}\neq 0$, we can realized a binary function $[\tilde{F}_{BRR}, , \tilde{F}_{BGR}, \tilde{F}_{BGG}+a t]$ over domain $\{G,R\}$ by connecting this ternary function to
a unary function $(1,t,0)$ and putting $\neq_{G,R}$ on
the other two dangling edges. Since $a\neq 0$ and we can choose any $t$, we can make the third entry of
$[\tilde{F}_{BRR}, , \tilde{F}_{BGR}, \tilde{F}_{BGG}+a t]$ arbitrary  and
 the function is not in all three tractable binary forms.
Therefore the problem is \#P-hard. Now we can assume
that $\tilde{F}_{BGR}=0$. To simplify notations,
we use variables to denote the function entries as follows
\begin{center}
\begin{tabular}{*{11}{c}c}
{ }     & { }   & { }   &$g$& { }     & { }   & { }  \\
{ }     & { }   &$y$& { } &$w$& { }     & { }  \\
{ }     &$x$& { } &  $0$& { }   &$z$& { }  \\
 $a$& { } &$0$& { } & $0$& { } &   $b$  \\
\end{tabular}
\end{center}

Then we use the gadget depicted in Figure \ref{fig:binary-1-z}   to construct another binary function in the LHS.
\begin{figure}[hbtp]
	\begin{center}
		\includegraphics[width=3 in]{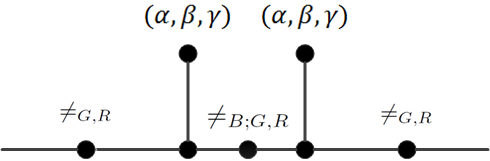}
	\caption{A binary gadget.}
	\label{fig:binary-1-z}
	\end{center}
\end{figure}
The signature of this binary function is
\[[f_0,f_1,f_2]=[(\alpha w + \gamma z)^2, (\alpha y + \beta x)(\alpha w + \gamma z) + (\alpha x +\beta a)(\alpha z + \gamma b) , (\alpha y + \beta x)^2].\]
If there exists some $(\alpha, \beta, \gamma)$ such that this  $[f_0,f_1,f_2]$ is not of
the form $[*, 0, *]$, $[0,*,0]$, or degenerate, then the problem is \#P-hard
and we are done. Otherwise, for all $(\alpha, \beta, \gamma)$, we have (1)
$f_0=f_2=0$, (2) $f_1=0$, or (3) $f_1^2=f_0 f_2$. Since all the conditions are
polynomials of  $(\alpha, \beta, \gamma)$, we can conclude that at least one of the three conditions (1), (2), or (3) holds for all $(\alpha, \beta, \gamma)$.

If condition (1) $f_0=f_2=0$ holds for all $(\alpha, \beta, \gamma)$, we have $x=y=z=w=0$ and the problem is separable and therefore tractable. And it can be
easily verified that this is the second form of Theorem \ref{thm:ternary}, where $\mbox{\boldmath $\alpha$}=\sqrt[3]{g}(1,0,0)^{\tt T}$,
 $\mbox{\boldmath $\beta_1$}=\frac{\sqrt[3]{a}}{\sqrt{2}}(0,1,i)^{\tt T}$ and
$\mbox{\boldmath $\beta_2$}=\frac{\sqrt[3]{b}}{\sqrt{2}}(0,1,-i)^{\tt T}$.

If condition (2) $f_1=0$ holds for all $(\alpha, \beta, \gamma)$, we have that
\[xz+yw=xb+yz=az + x w= ab +xz=0.\]
Since $xz=-ab\neq 0$, we can conclude from above that
\[\frac{x}{a}=\frac{y}{x}=p,~~~\frac{z}{b}=\frac{w}{z}=q,~~~\mbox{and}~~~pq=-1.\]
The ternary signature has the form
\begin{center}
\begin{tabular}{*{11}{c}c}
{ }     & { }   & { }   &$g$& { }     & { }   & { }  \\
{ }     & { }   &$a p^2$& { } &$b q^2$& { }     & { }  \\
{ }     &$a p$& { } &  $0$& { }   &$ b q$& { }  \\
 $a$& { } &$0$& { } & $0$& { } &   $b$  \\
\end{tabular}
\end{center}

\begin{figure}[hbtp]
	\begin{center}
		\includegraphics[width=3 in]{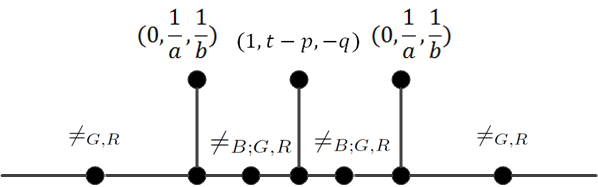}
	\caption{A binary gadget.}
	\label{fig:binary-3-z1}
	\end{center}
\end{figure}
Then we use the gadget in Figure \ref{fig:binary-3-z1} to construct another binary function over domain $\{G,R\}$, whose signature is
\[\begin{bmatrix} q & 0 & 1 \\ p & 1 & 0   \end{bmatrix}
\begin{bmatrix} g-  b q^3-a p^3 + a p^2 t & 0 & a p t  \\ 0 & 0 & 0 \\ a p t & 0 & a t   \end{bmatrix}
\begin{bmatrix} q & p \\ 0 & 1 \\ 1 & 0 \end{bmatrix}=\begin{bmatrix}q^2 \delta + at (pq +1)^2	& 	pq \delta + a p^2 t (pq +1) \\
pq \delta + a p^2 t (pq +1)	&	p^2 \delta + a p^4 t
    \end{bmatrix} ,\]
where $\delta=g-a p^3 - b q^3$. We denote this symmetric binary function as $[g_0, g_1, g_2]$.

If $\delta = 0$, i.e.,  $g=a p^3 + b q^3 $, we show that this is indeed a tractable case of Theorem \ref{thm:ternary} as follows.
The ternary function $\mathbf{\tilde{F}}$ can be written as
\[\mathbf{\tilde{F}}=a \begin{bmatrix} p \\ 1 \\ 0 \end{bmatrix}^{\otimes 3} + b \begin{bmatrix} q \\ 0 \\ 1 \end{bmatrix}^{\otimes 3}.\]
And
\[\mathbf{F}=\tilde{Z}\mathbf{\tilde{F}}
=a \begin{bmatrix} p \\ \frac{1}{\sqrt{2}} \\ \frac{i}{\sqrt{2}} \end{bmatrix}^{\otimes 3} + b \begin{bmatrix} q \\ \frac{1}{\sqrt{2}} \\ -\frac{i}{\sqrt{2}} \end{bmatrix}^{\otimes 3}.\]
This is of the first form of tractable cases in  Theorem \ref{thm:ternary},
where $\mbox{\boldmath $\alpha$}
=\sqrt[3]{a}(p,\frac{1}{\sqrt{2}} , \frac{i}{\sqrt{2}})^{\tt T}$,
$\mbox{\boldmath $\beta$}
=\sqrt[3]{b}(q,\frac{1}{\sqrt{2}} , -\frac{i}{\sqrt{2}})^{\tt T}$ and
$\mbox{\boldmath $\gamma$}=(0,0,0)^{\tt T}$.
We note that the condition
$\langle \mbox{\boldmath $\alpha$},\mbox{\boldmath $\beta$} \rangle=0$
is guaranteed by $pq=-1$.

Now we assume that $\delta \neq 0$. If there exists
 some $t$ such that the binary function $[g_0, g_1, g_2]$ is not of
the form $[*, 0, *]$, $[0,*,0]$, or degenerate, then the problem is \#P-hard and we are done.
Otherwise, by the same argument as above,  at least one of the three
statements  holds for all $t$:
(i) $g_0=g_2=0$, (ii) $g_1=0$, or (iii) $g_1^2=g_0 g_2$. Choose $t=0$, we have $g_0 g_1 g_2 \neq 0$. Therefore, the only possibility is  that $g_1^2=g_0 g_2$ holds for all $t$. However, this is also a contradiction which can be seen by choosing $t=\frac{1}{a}$.
One can calculate the determinant
$\det \begin{bmatrix} g_0 & g_1 \\ g_1 & g_2  \end{bmatrix} =
\delta p^2 \not =0$.

If condition (3)
$f_1^2=f_0 f_2$ holds for all $(\alpha, \beta, \gamma)$, we have
\begin{equation}\label{eqn:f1^2=f0f2-za00b}
0=f_1^2 - f_0 f_2 = 2(\alpha y + \beta x)(\alpha w + \gamma z)(\alpha x +\beta a)(\alpha z + \gamma b) + (\alpha x +\beta a)^2(\alpha z + \gamma b)^2 .
\end{equation}
Let $\alpha=x$ and $\beta=-y$, we have $( x^2-ay)^2 (x z + \gamma b)^2=0$ holds for all $\gamma$. Since $b\neq 0$,
we conclude that $a y - x^2=0$. Similarly, let  $\alpha=z$ and $\gamma=-w$, we can get $bw -z^2=0$. Then let $\beta=\gamma=1$ and $\alpha=0$ in (\ref{eqn:f1^2=f0f2-za00b}), we have
\[ ab + 2 x z=0.\]
Denote by $p=\frac{x}{a}$ and $q=\frac{z}{b}$, we have $pq=-\frac{1}{2}$ and the ternary signature has the form
\begin{center}
\begin{tabular}{*{11}{c}c}
{ }     & { }   & { }   &$g$& { }     & { }   & { }  \\
{ }     & { }   &$a p^2$& { } &$b q^2$& { }     & { }  \\
{ }     &$a p$& { } &  $0$& { }   &$ b q$& { }  \\
 $a$& { } &$0$& { } & $0$& { } &   $b$  \\
\end{tabular}
\end{center}

\begin{figure}[hbtp]
	\begin{center}
		\includegraphics[width=3 in]{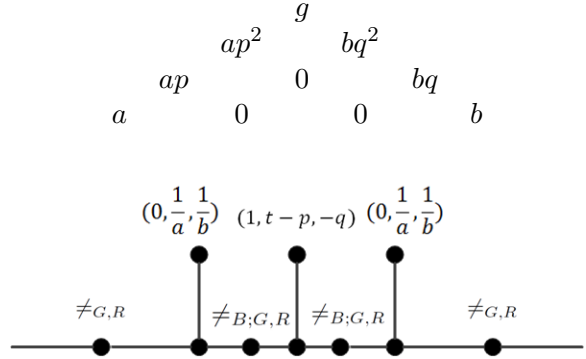}
	\caption{A binary gadget.}
	\label{fig:binary-3-z1-copy}
	\end{center}
\end{figure}
Then we use the gadget in Figure \ref{fig:binary-3-z1-copy} to construct another binary function over domain $\{G,R\}$, whose signature is
\[\begin{bmatrix} q & 0 & 1 \\ p & 1 & 0   \end{bmatrix}
\begin{bmatrix} g-  b q^3-a p^3 + a p^2 t & 0 & a p t  \\ 0 & 0 & 0 \\ a p t & 0 & a t   \end{bmatrix}
\begin{bmatrix} q & p \\ 0 & 1 \\ 1 & 0 \end{bmatrix}=\begin{bmatrix}q^2 \delta + at (pq +1)^2	& 	pq \delta + a p^2 t (pq +1) \\
pq \delta + a p^2 t (pq +1)	&	p^2 \delta + a p^4 t
    \end{bmatrix} ,\]
where $\delta=g-a p^3 - b q^3$. We denote this symmetric binary function as $[g_0, g_1, g_2]$.
(This is the same construction as in Figure~\ref{fig:binary-3-z1},
but $p$ and $q$ have a different meaning.)

If $\delta=0$, we show that this is indeed a tractable case of Theorem \ref{thm:ternary} as follows.
The ternary function $\mathbf{\tilde{F}}$ can be written as
\[\mathbf{\tilde{F}}=a \begin{bmatrix} p \\ 1 \\ 0 \end{bmatrix}^{\otimes 3} + b \begin{bmatrix} q \\ 0 \\ 1 \end{bmatrix}^{\otimes 3}.\]
And
\[\mathbf{F}=\tilde{Z}\mathbf{\tilde{F}}=a \begin{bmatrix} p \\ \frac{1}{\sqrt{2}} \\ \frac{i}{\sqrt{2}} \end{bmatrix}^{\otimes 3} + b \begin{bmatrix} q \\ \frac{1}{\sqrt{2}} \\ -\frac{i}{\sqrt{2}} \end{bmatrix}^{\otimes 3}.\]
This is of the third form of tractable case in  Theorem \ref{thm:ternary},
where $\mathbf{F}_\beta$ is the given function
$\mathbf{F}$ and
$\mbox{\boldmath $\beta$}=(-\sqrt{2}, p+q, -p i + q i)^{\tt T},
\mbox{\boldmath $\gamma$}=(0,0,0)^{\tt T}$.
We note that $pq=-\frac{1}{2}$ implies that
$\langle \mbox{\boldmath $\beta$},\mbox{\boldmath $\beta$} \rangle=0$.

Now we assume that $\delta \neq 0$. If there exists some $t$ such that this binary function is not of the form $[*, 0, *]$, $[0,*,0]$, or degenerate, then the problem is \#P-hard and we are done.
Otherwise, by the same argument as above,  at least one of the three
(i) $g_0=g_2=0$, (ii) $g_1=0$, or (iii) $g_1^2=g_0 g_2$ holds for all $t$. Choose $t=0$, we have $g_0 g_1 g_2 \neq 0$. Therefore, the only
 possibility is  that $g_1^2=g_0 g_2$ holds for all $t$. However, this is also a contradiction which can be seen by choosing $t=\frac{1}{a}$.
One can calculate the determinant
$\det \begin{bmatrix} g_0 & g_1 \\ g_1 & g_2  \end{bmatrix} =
\delta p^2 \not =0$.
This completes the proof for the case  $\mathbf{F}^{*\rightarrow\{G,R\}}=Z [a, 0, 0, b]^{\tt T}$.

\vspace{.1 in}

\noindent {\bf Case 3: $\mathbf{F}^{*\rightarrow\{G,R\}}
=Z [a, b, 0, 0]^{\tt T},
~~~b \not = 0$.}

Here we only prove for the case $Z=\frac{1}{\sqrt{2}}\begin{bmatrix} 1 & 1 \\ i & -i \end{bmatrix}$. The other case is symmetric.
The Holant problem can be written as Holant$^*(=_2 | \{\mathbf{F},=_{G,R}\})$, where $*$ means that both sides
can use unary functions.
After a holographic transformation under the matrix $\tilde{Z}=\begin{bmatrix} 1 & \mathbf{0} \\ \mathbf{0} & Z \end{bmatrix}$, we can get an equivalent problem  Holant$^*(\neq_{B;G,R} | \{\tilde{Z}^{-1} \mathbf{F},\neq_{G,R}\})$, where
the two binary functions $\neq_{B;G,R}$ and $\neq_{G,R}$
are given in (\ref{neq-BGR-def}).
We use $\mathbf{\tilde{F}}$ to denote this ternary function after
the transformation $\tilde{Z}^{-1} \mathbf{F}$.
Then we have  $\mathbf{\tilde{F}}^{*\rightarrow\{G,R\}}= [a, b, 0, 0]$. And after a scaling, we assume that
$\mathbf{\tilde{F}}^{*\rightarrow\{G,R\}}= [a, 1, 0, 0]$.
By connecting $\neq_{B;G,R}$ to both sides of $\neq_{G,R}$, we can get a $\neq_{G,R}$ on the LHS.
For a bipartite holant problem
Holant$^*([f_0,f_1, f_2]| [a, 1, 0, 0])$ over domain size 2,
the problem is \#P-hard unless the binary function $[f_0,f_1, f_2]$ is of
the form $[0, *, *]$ or degenerate.
This can be seen as follows: Clearly for such $[f_0,f_1, f_2]$
it is tractable, as $[0, *, *]$ requires the number of edges assigned 0
to be at most the number assigned 1, while $[a, 1, 0, 0]$
requires the number of edges assigned 0 to be strictly more than
the number assigned 1. Suppose $[f_0,f_1, f_2]$ is nondegenerate
and not of this
form, we may normalize it to $[1, b, c]$ where $c \not = b^2$.
Consider the holographic reduction defined by
$M = \begin{bmatrix} 1 & b \\ 0 & \sqrt{c-b^2} \end{bmatrix}$.
The matrix form for $[1,0,1] M^{\otimes 2}$ is
$M^{\tt T} I_2 M = \begin{bmatrix} 1 & b \\ b & c \end{bmatrix}$,
namely $[1, b, c]$,
while $M^{\otimes 3} [a, 1, 0, 0]$ is
\[
\begin{bmatrix} 1 & b \\ 0 & \sqrt{c-b^2} \end{bmatrix}^{\otimes 3}
\left[ a \begin{bmatrix} 1 \\ 0 \end{bmatrix}^{\otimes 3} +
\begin{bmatrix} 1 \\ 0 \end{bmatrix} \otimes
\begin{bmatrix} 1 \\ 0 \end{bmatrix} \otimes
\begin{bmatrix} 0 \\ 1 \end{bmatrix}
+
\begin{bmatrix} 1 \\ 0 \end{bmatrix} \otimes
\begin{bmatrix} 0 \\ 1 \end{bmatrix} \otimes
\begin{bmatrix} 1 \\ 0 \end{bmatrix}
+
\begin{bmatrix} 0 \\ 1 \end{bmatrix} \otimes
\begin{bmatrix} 1 \\ 0 \end{bmatrix} \otimes
\begin{bmatrix} 1 \\ 0 \end{bmatrix}
\right],\]
which is $[a+ 3b, \sqrt{c-b^2}, 0, 0]$.
By Theorem  \ref{thm:dich-sym-Boolean},
Holant$^*([a+ 3b, \sqrt{c-b^2}, 0, 0])$ is \#P-hard.
Therefore, to show \#P-hardness, we will
construct binary functions in the LHS of Holant$^*(\{\neq_{B;G,R}, \neq_{G,R}\} | \mathbf{\tilde{F}})$ over domain $\{G,R\}$.

Now we have the ternary function $\mathbf{\tilde{F}}$ as follows
\begin{center}
\begin{tabular}{*{11}{c}c}
{ }     & { }   & { }   &$\tilde{F}_{BBB}$& { }     & { }   & { }  \\
{ }     & { }   &$\tilde{F}_{BBG}$& { } &$\tilde{F}_{BBR}$& { }     & { }  \\
{ }     &$\tilde{F}_{BGG}$& { } &  $\tilde{F}_{BGR}$& { }   &$\tilde{F}_{BRR}$& { }  \\
 $a$& { } &$1$& { } & $0$& { } &   $0$  \\
\end{tabular}
\end{center}

If $\tilde{F}_{BRR}\neq 0$, we can realized a binary function $[\tilde{F}_{BRR}, \tilde{F}_{BGR}+ t, \tilde{F}_{BGG}+a t]$ in LHS over domain $\{G,R\}$ by connecting this ternary function to
a unary function $(1,t,0)$ and putting $\neq_{G,R}$ on
the other two dangling edges. It can be easily seen that we can choose some $t$ such that $[\tilde{F}_{BRR}, \tilde{F}_{BGR}+ t, \tilde{F}_{BGG}+a t]$ is not degenerate. And it is not of
the form $[0, *, *]$  since $\tilde{F}_{BRR} \neq 0$.
Therefore the problem is \#P-hard. Now we can assume
that $\tilde{F}_{BRR}=0$. To simplify notations, we use
variables to denote the function entries of $\mathbf{\tilde{F}}$ as follows
\begin{equation}\label{F-tilde-matchingonGR}
\begin{tabular}{*{11}{c}c}
{ }     & { }   & { }   &$g$& { }     & { }   & { }  \\
{ }     & { }   &$z$& { } &$w$& { }     & { }  \\
{ }     &$x$& { } &  $y$& { }   &$0$& { }  \\
 $a$& { } &$1$& { } & $0$& { } &   $0$  \\
\end{tabular}
\end{equation}

Then we use the gadget in Figure \ref{fig:binary-1-zm}   to construct another binary function.
\begin{figure}[hbtp]
	\begin{center}
		\includegraphics[width=3 in]{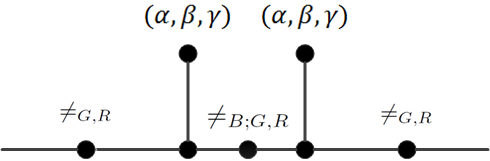}
	\caption{A binary gadget.}
	\label{fig:binary-1-zm}
	\end{center}
\end{figure}
The signature of this binary function is
\[[f_0,f_1,f_2]=[(\alpha w + \beta y)^2, (\alpha w + \beta y)(\alpha z + \beta x + \gamma y) + (\alpha y +\beta)^2 , (\alpha z + \beta x + \gamma y)^2 + 2 (\alpha y +\beta) (\alpha x + \beta a + \gamma)].\]
If there exists some $(\alpha, \beta, \gamma)$ such that this  $[f_0,f_1,f_2]$ is not of
the form $[0, *, *]$ or degenerate, then the problem is \#P-hard
and we are done. Otherwise, for all $(\alpha, \beta, \gamma)$, we have
the equalities (1) $f_0=0$ or (2) $f_1^2=f_0 f_2$.
Since these are
polynomials of  $(\alpha, \beta, \gamma)$,  at least one of  (1) $f_0=0$
 or (2) $f_1^2=f_0 f_2$ holds for all $(\alpha, \beta, \gamma)$.

If  equality (1) $f_0=0$ holds for all $(\alpha, \beta, \gamma)$,
we have $y=w=0$.
We will verify that the problem is tractable, and
 it is of the third form of Theorem \ref{thm:ternary}.
We use $\mathbf{\tilde{F}}_\beta$ to denote the following ternary function
\begin{center}
\begin{tabular}{*{11}{c}c}
{ }     & { }   & { }   &$g$& { }     & { }   & { }  \\
{ }     & { }   &$z$& { } &$0$& { }     & { }  \\
{ }     &$x$& { } &  $0$& { }   &$0$& { }  \\
 $a$& { } &$0$& { } & $0$& { } &   $0$  \\
\end{tabular}
\end{center}
Let $\tilde{\beta}=(0,1,0)^{\tt T}$ and $\tilde{\gamma}=(0,0,1)^{\tt T}$.
Compared to (\ref{F-tilde-matchingonGR}) we have
  \[ \mathbf{\tilde{F}}=\mathbf{\tilde{F}_\beta} + \mathbf{\tilde{\beta}}^{\otimes 2} \otimes \tilde{\gamma}  + \tilde{\beta} \otimes \tilde{\gamma} \otimes \tilde{\beta}+  \tilde{\gamma} \otimes \mathbf{\tilde{\beta}}^{\otimes 2}.\]
  Therefore applying the transformation $\tilde{Z}$, we have
  \[ \mathbf{F}= \tilde{Z}\mathbf{\tilde{F}}=\tilde{Z} \mathbf{\tilde{F}_\beta} + {(\tilde{Z} \tilde{\beta})}^{\otimes 2} \otimes (\tilde{Z} \tilde{\gamma})  + (\tilde{Z} \tilde{\beta}) \otimes (\tilde{Z} \tilde{\gamma}) \otimes (\tilde{Z} \tilde{\beta})+  (\tilde{Z} \tilde{\gamma}) \otimes {(\tilde{Z} \tilde{\beta}) }^{\otimes 2}.\]
  We verify that this is in
 the third form of Theorem \ref{thm:ternary}, with
 $\mathbf{F}_\beta=\tilde{Z} \mathbf{\tilde{F}_\beta}$,
$\mbox{\boldmath $\beta$}=\tilde{Z} \tilde{\beta}$ and
 $\mbox{\boldmath $\gamma$}=\tilde{Z} \tilde{\gamma}$.
  First it is easy to verify that
$\langle \mbox{\boldmath $\beta$},\mbox{\boldmath $\beta$} \rangle =0$.
  Second
$\langle (0,0,1), \mathbf{\tilde{F}_\beta} \rangle =\mathbf{0}$,
and $\tilde{\beta}^{\tt T} \tilde{Z}^T \tilde{Z} = (0,0,1)$, hence
  \[\langle \mbox{\boldmath $\beta$}, \mathbf{{F}_\beta} \rangle = \mathbf{0}.\]

If  equality (2) $f_1^2=f_0 f_2$ holds for
all $(\alpha, \beta, \gamma)$, we have
\[
0=f_0 f_2-f_1^2 = (\alpha y +\beta) (2 (\alpha w + \beta y)^2  (\alpha x + \beta a + \gamma)-2 (\alpha w + \beta y)(\alpha z + \beta x + \gamma y) (\alpha y +\beta) - (\alpha y +\beta)^3)  .
\]
As a product of two polynomials in $(\alpha, \beta, \gamma)$,
to be identically zero,
one of them must be identically zero.
Since $\alpha y +\beta$ is not identically zero, we have
\begin{equation}\label{eqn:f1^2=f0f2-zab00}
2 (\alpha w + \beta y)^2  (\alpha x + \beta a + \gamma)-2 (\alpha w + \beta y)(\alpha z + \beta x + \gamma y) (\alpha y +\beta) - (\alpha y +\beta)^3 = 0,
\end{equation}
for all $(\alpha, \beta, \gamma)$.

Let $\alpha=y$ and $\beta=-w$, we have $ w=y^2$.
Substituting  $ w=y^2$ in (\ref{eqn:f1^2=f0f2-zab00}), we have
\[2y^2 (\alpha x + \beta a + \gamma) -2 y (\alpha z + \beta x + \gamma y) -(\alpha y +\beta) = 0,\]
for all $(\alpha, \beta, \gamma)$.
And we conclude that $y (2 x y -2 z - 1) =0$ and $2 a y^2 - 2 x y -1 =0$.
The second equation implies that $y \not =0$.
So we have $2 x y -2 z - 1 =0$.

Then the ternary signature has the form
\begin{equation}\label{F-form-in-zab00}
\begin{tabular}{*{11}{c}c}
{ }     & { }   & { }   &$g$& { }     & { }   & { }  \\
{ }     & { }   &$a y^2-1$& { } &$y^2$& { }     & { }  \\
{ }     &$a y -\frac{1}{2y}$& { } &  $y$& { }   &$ 0$& { }  \\
 $a$& { } &$1$& { } & $0$& { } &   $0$  \\
\end{tabular}
\end{equation}

If $g=a y^3 -\frac{3}{2} y$, we show that this problem is indeed tractable.
We show that this $\mathbf{F}$ is of the third form in Theorem \ref{thm:ternary}, where $\mathbf{F}_\beta$ is the given function $\mathbf{F}$, and
\[
\mbox{\boldmath $\beta$}=(1, -\frac{y}{\sqrt{2}}+\frac{1}{2 \sqrt{2}y} ,
 \frac{y i}{\sqrt{2}}+\frac{i}{2 \sqrt{2}y})^{\tt T}, ~~~~
\mbox{\boldmath $\gamma$}=(0,0,0)^{\tt T}.\]
First it is easy to verify that $\langle \mbox{\boldmath $\beta$},
\mbox{\boldmath $\beta$} \rangle=0$.  We
also need to verify
$\langle \mbox{\boldmath $\beta$}, \mathbf{F} \rangle = \mathbf{0}$,
or  $\langle \mbox{\boldmath $\beta$}, \tilde{Z} \mathbf{\tilde{F}} \rangle = \mathbf{0}$,
which  is equivalent to $\langle \tilde{Z}^{\tt T}
\mbox{\boldmath $\beta$},   \mathbf{\tilde{F}}
 \rangle = \mathbf{0}$.
This $\tilde{Z}^{\tt T} \mbox{\boldmath $\beta$}$ is $(1, -y, \frac{1}{2y})^{\tt T}$,
and it is easy to see that indeed this unary
function annihilates $\mathbf{\tilde{F}}$ in
(\ref{F-form-in-zab00}), using the ``calculus''
from Section~\ref{section:calculus}.

\begin{figure}[hbtp]
	\begin{center}
		\includegraphics[width=3 in]{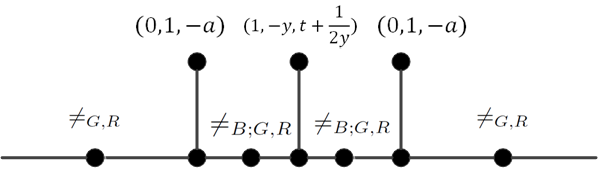}
	\caption{A binary gadget.}
	\label{fig:binary-3-zm}
	\end{center}
\end{figure}
Now we assume that $g \neq a y^3 -\frac{3}{2} y$.
We denote by $\delta=g - a y^3 + \frac{3}{2} y$.
 Then we use the gadget  in Figure \ref{fig:binary-3-zm} to construct another binary function over domain $\{G,R\}$, whose signature is
\[\begin{bmatrix} y & 1 & 0 \\ -\frac{1}{2y} & 0 & 1   \end{bmatrix}
\begin{bmatrix} \delta +y^2 t & 0  &   y t \\ 0 & 0 & 0 \\  y t  &  0 & t \end{bmatrix}
\begin{bmatrix} y & -\frac{1}{2y} \\ 1 & 0 \\ 0 & 1 \end{bmatrix}=\begin{bmatrix}y^2 \delta + y^4 t	& 	 -\frac{\delta}{2} +\frac{y^2 t}{2} \\
-\frac{\delta}{2} +\frac{y^2 t}{2} 	&	\frac{\delta}{4 y^2} +\frac{t}{4}
    \end{bmatrix}.\]
 We denote this symmetric binary function as $[g_0, g_1, g_2]$.
 If there exists some $t$ such that this binary function is not of
the form $[0, *, *]$, or degenerate, then the problem is \#P-hard and
we are done.
Otherwise, by the same argument as above,  at least one of the two
equations (i) $g_0=0$ or (ii) $g_1^2=g_0 g_2$ holds for all $t$.
Choose $t=0$, we have $g_0 = y^2 \delta \neq 0$.  Choose $t=\frac{\delta}{y^2}$, we can verify that the signature is not degenerate.
 This completes the proof for the case  $\mathbf{F}^{*\rightarrow\{G,R\}}=Z [a, b, 0, 0]^{\tt T}$.

We have completed the proof of Theorem~\ref{thm:F-with-=_GR}
when $\mathbf{F}^{*\rightarrow\{G,R\}}$ is non-degenerate.

\subsection{$\mathbf{F}^{*\rightarrow\{G,R\}}$ is $[1,0,0,0]$
After an Orthogonal Transformation}\label{bottom1000}
We have proved our dichotomy theorem Theorem~\ref{thm:ternary}
when $\mathbf{F}^{*\rightarrow\{G,R\}}$ is non-degenerate.
In the rest of this paper we deal with the case when
$\mathbf{F}^{*\rightarrow\{G,R\}}$ is degenerate.
We first suppose it has rank 1, and therefore has the form
$(a, b)^{\otimes 3}$. In this subsection
we assume $(a, b)$ is  non-isotropic.

Thus, possibly after a reversal, which is an orthogonal transformation,
 $\mathbf{F}^{*\rightarrow\{G,R\}}$
is degenerate of the form $c [1,\lambda, \lambda^2,
\lambda^3]^{\tt T}$, where $c \not = 0$
and $\lambda \not \in \{i,-i\}$. As $(1, \lambda)$
is not isotropic, we can perform an orthogonal transformation,
after which, and ignoring a
non-zero scalar multiple,
we may assume   the bottom line $\mathbf{F}^{*\rightarrow\{G,R\}}$
is $[1,0,0,0]^{\tt T}$.

Suppose $\mathbf{F}=[u; t, r; s, p, q; 1, 0, 0, 0]$, namely
\begin{equation}\label{bottom1000F}
\begin{tabular}{*{11}{c}c}
{ }     & { }   & { }   &$u$& { }     & { }   & { }  \\
{ }     & { }   &$t$& { } &$r$& { }     & { }  \\
{ }     &$s$& { } &  $p$& { }   &$q$& { }  \\
 $1$& { } &$0$& { } & $0$& { } &   $0$  \\
\end{tabular}
\end{equation}

Our general method will be to
construct some gadgets which can realize certain  functions, and
under some conditions we can use them
to show  that Holant$^*(\mathbf{F})$ is \#P-hard.
We first construct the following  function
$\mathbf{H}(x_1,x_2,x_3,x_4)=(\sum_{y \in \{B, G, R\}}
\mathbf{F}(x_1,x_2,y)\mathbf{F}(y,x_3,x_4))^{*\rightarrow\{G,R\}}$,
which can be realized by connecting two copies of $\mathbf{F}$
by one edge, and then connecting $=_{B,G}$
on all four external edges.
Denote by $\mathcal{T}$ the set of functions of arity at most 2,
and $\langle \mathcal{T} \rangle$ the tensor product closure
of $\mathcal{T}$.

\begin{lemma}
If $p \neq 0$ or $q
\neq 0$, then $\mathbf{H} \not \in \langle \mathcal{T} \rangle$.
\end{lemma}

\begin{proof}
By the definition of $\mathbf{H}$,
we have
$\mathbf{H}(x_1,x_2,x_3,x_4)=\mathbf{H}(x_2,x_1,x_3,x_4)=\mathbf{H}(x_1,x_2,x_4,x_3)$,
because $\mathbf{F}$ is a symmetric function. We prove the lemma by
a contradiction. Suppose $\mathbf{H} \in \langle \mathcal{T} \rangle$.
Then, either
$\mathbf{H}(x_1,x_2,x_3,x_4)=\mathbf{P}(x_1,x_2)\mathbf{Q}(x_3,x_4)$,
or
$\mathbf{H}(x_1,x_2,x_3,x_4)=\mathbf{P}(x_1,x_3)\mathbf{Q}(x_2,x_4)$
for some binary functions $\mathbf{P}$ and $\mathbf{Q}$.

The matrix form $\mathbf{H}_{x_1x_2,~x_3x_4}$ (whose rows are indexed
by $x_1x_2$ and columns are indexed by  $x_3x_4$, and are in the order
$GG,GR,RG,RR$) of the function $\mathbf{H}$ is
\[\left [
\begin{array}{cccc}  s^2+1 & sp & sp & sq \\ sp & p^2 & p^2 & pq  \\sp & p^2 & p^2 & pq   \\ sq & pq & pq & q^2
\end{array} \right ].\]
If
$\mathbf{H}(x_1,x_2, x_3,x_4)=\mathbf{P}(x_1,x_2)\mathbf{Q}(x_3,x_4)$
for some binary functions $\mathbf{P}$ and $\mathbf{Q}$,
then $\mathbf{H}_{x_1x_2,~x_3x_4}=\mathbf{P}_{4 \times 1}
(\mathbf{Q}_{4 \times 1})^{\tt T}$ has rank at
most 1, where we use the vector form $\mathbf{P}_{4 \times 1}$
and $\mathbf{Q}_{4 \times 1}$ for the functions $\mathbf{P}$ and
$\mathbf{Q}$.  This implies that $p=q=0$, by
taking some $2 \times 2$ determinantal minors.
A contradiction.

The matrix form $\mathbf{H}_{x_1x_3,~x_2x_4}$ (also in index order
$GG,GR,RG,RR$) of the function $\mathbf{H}$ is
\[\left [
\begin{array}{cccc}  s^2+1 & sp & sp & p^2 \\ sp & sq & p^2 & pq  \\sp & p^2 & sq & pq   \\ p^2 & pq & pq & q^2
\end{array} \right ].\]
If
$\mathbf{H}(x_1,x_2,x_3,x_4)=\mathbf{P}(x_1,x_3)\mathbf{Q}(x_2,x_4)$
for some binary functions $\mathbf{P}$ and $\mathbf{Q}$,
by the same argument,
the matrix $\mathbf{H}_{x_1x_3,~x_2x_4}$  has rank at
most 1. If $p \neq 0$, then because the submatrix indexed by $(RG, RR)
\times (GG, GR)$ is singular, we get $sq=p^2$. Similarly if $q \neq 0$, by the
submatrix indexed by $(RG, RR) \times
(RG, RR)$, we also get $sq=p^2$. Hence, $sq=p^2$ holds. We also have
the determinantal minor indexed  by $(GG, GR) \times (GG, GR)$, $\left |
\begin{array}{cc} s^2+1 & sp \\ sp & sq \end{array} \right |=p^2=0$
and the  minor indexed  by $(GG, RR) \times (GG, RR)$,
 $\left | \begin{array}{cc} s^2+1 & p^2 \\ p^2 & q^2 \end{array}
\right |=q^2=0$.  A contradiction.
\end{proof}

We have proved that
$\mathbf{H} \not \in \langle \mathcal{T} \rangle$
 under the condition $p \neq 0$ or $q \neq 0$.

\begin{lemma}
In Holant$^*(\mathbf{F})$, if $p \neq 0$ or $q \neq 0$, either for
each tractable class $\mathcal{P}_{a,b}$ and $\mathcal{P}$, we can
construct a symmetric non-degenerate binary function not in it, or the
domain is separated and the complexity dichotomy holds.
\end{lemma}

\begin{proof}
Define binary functions
$\mathbf{R}_x=\langle (1,x,0), \mathbf{F} \rangle^{*\rightarrow\{G,R\}}$,
which are clearly realizable,
where $x$ is an arbitrary (algebraic) complex number.
In symmetric signature notation on the Boolean domain
$\mathbf{R}_x=[s+x,p,q]$. If $p \neq
0$ or $q \neq 0$, there is at most one value  $x$ such that $\mathbf{R}_x$
is degenerate.
Assume $p \neq 0$. For any $a \neq 0$, obviously there exists an $x$ such that
a non-degenerate $\mathbf{R}_x \not \in \mathcal{P}_{a,b}$, since the
coefficient of $x$ in the linear equation requirement
for $\mathbf{R}_x \in \mathcal{P}_{a,b}$ (for both alternative forms
in $\mathcal{P}_{a,b}$) is not
zero. For $a=0$, then $b \not =0$,
and  because the middle entry of $\mathbf{R}_x$ is
 $p \neq 0$,  there exists an  $x$ such that
a non-degenerate
$\mathbf{R}_x \not \in \mathcal{P}_{0,b}$. By the same reasoning,
 there exists an $x$ such that
a non-degenerate
$\mathbf{R}_x \not \in \mathcal{P}$.
This completes the proof of the lemma for the case $p \not =0$.

If $p=0$, we have $q \neq 0$.
For each tractable class $\mathcal{P}_{a,b}$ and $\mathcal{P}$,
 except for
$\mathcal{P}_{0,b}$, the function $\mathbf{R}_x$
still handles it for all but finitely many values of $x$.
 The exception is
$\mathcal{P}_{0,b}$, which has the normalized form $\mathcal{P}_{0,1}$.
We prove this case according to  $s \neq 0$ or $s=0$.

The matrix form $\mathbf{H}_{x_1x_2,~x_3x_4}$ (in index order
$GG,GR,RG,RR$) of the function $\mathbf{H}$ is
\[\left [
\begin{array}{cccc}  s^2+1 & sp & sp & sq \\ sp & p^2 & p^2 & pq  \\sp & p^2 & p^2 & pq   \\ sq & pq & pq & q^2
\end{array} \right ]=\left [
\begin{array}{cccc}  s^2+1 & 0 & 0 & sq \\ 0 & 0 & 0 & 0  \\0 & 0 & 0 & 0  \\ sq & 0 & 0 & q^2
\end{array} \right ].\]
Since $q \not = 0$,
the following binary function on $x_1, x_3$ is non-degenerate and realizable
\[\sum_{x_2,x_4 \in \{G, R\}}\mathbf{H}(x_1,x_2,x_3,x_4)=[s^2+1,sq,q^2].\]
For $s \neq 0$ (and $q \not = 0$),
it is easy to check that $[s^2+1,sq,q^2] \not \in
\mathcal{P}_{0,1}$.

Now we suppose $s=0$ (and $p=0$), the simpler construction does not work
and we use a slightly more
complicated construction. Suppose the binary function
$\mathbf{P}= \langle (\alpha, \beta, \gamma), \mathbf{F} \rangle
=
\left [ \begin{array}{ccc} a & b & c \\ b & e & d \\
c & d & f \end{array} \right ]$,
where $\alpha, \beta, \gamma$ are to be determined.
 We construct
\[\mathbf{Q}(x_1,x_2,x_3,x_4)=\left(\sum_{y_1,y_2 \in \{B, G, R\}}
\mathbf{F}(x_1,x_2,y_1)\mathbf{P}(y_1,y_2)\mathbf{F}(y_2,x_3,x_4)\right)^{*\rightarrow\{G,R\}}.\]
This is realizable by connecting three copies of $\mathbf{F}$
in a chain with the middle copy connected to a
unary $(\alpha, \beta, \gamma)$ on one edge,
and then connecting $=_{B,G}$ on four external edges.
The matrix form $\mathbf{Q}_{x_1x_2,~x_3x_4}$ (in index order
$GG,GR,RG,RR$) of the function $\mathbf{Q}$ is
\begin{eqnarray*}
\begin{bmatrix}
 s &1 & 0\\
 p & 0 & 0 \\
 p & 0 & 0  \\
q & 0 & 0
\end{bmatrix}
\begin{bmatrix} a & b & c \\ b & e & d \\ c & d & f
\end{bmatrix}
\begin{bmatrix}
 s &p & p & q\\ 1 & 0 & 0 & 0 \\ 0 & 0 & 0 & 0
\end{bmatrix}
=
\begin{bmatrix}
 as^2+2bs+e & asp+bp & asp+bp & asq+bq \\ asp+bp & ap^2 &ap^2 & apq \\
asp+bp & ap^2 & ap^2 & apq \\ asq+bq & apq & apq & aq^2
\end{bmatrix}
=
\begin{bmatrix}
e & 0 & 0 & bq \\ 0 & 0 &0& 0\\
0 & 0 &0& 0\\ bq & 0 & 0 & aq^2
\end{bmatrix}.
\end{eqnarray*}

Let
$\mathbf{S}=\sum_{x_2,x_4 \in \{G, R\}}\mathbf{Q}(x_1,x_2,x_3,x_4)
=[e,bq,aq^2]$, which is realizable.
We want to show that there exists $(\alpha, \beta, \gamma)$ such that
a non-degenerate $\mathbf{S}
\not \in \mathcal{P}_{0,1}$. This means a non-degenerate $\mathbf{S}$
satisfies
$bq \neq 0$ and ($e \neq 0$ or $aq^2 \neq 0$).
The violation of these requirements are
specified by polynomial equations on $(\alpha, \beta, \gamma)$, therefore
we only need to show there exists $(\alpha, \beta, \gamma)$
satisfying each condition separately.

By definition $e = P_{GG}= s \alpha + 1 \beta + 0 \gamma
= \beta$ is not the zero polynomial (in $\alpha, \beta, \gamma$).
Similarly, by the ``calculus'' from Section~\ref{section:calculus},
$a = P_{BB} = u \alpha + t \beta + r \gamma$
and $b= P_{BG} = t \alpha$.
$\mathbf{S}$ is non-degenerate iff $b^2-ae \neq 0$, since $q \not =0$.
If $t \neq 0$, then
$b^2-ae$, $bq$ and $e$ are all non-zero polynomials in $\alpha, \beta, \gamma$.
Hence, if $t \neq 0$, we can get a non-degenerate binary function not in
$\mathcal{P}_{0,1}$.

If  $t=0$  (and $p=s=0$), for the domain of $\mathbf{F}$,
  $\{G\}$ is separated from $\{B, R\}$. The validity of
Theorem~\ref{thm:ternary} for such cases follows from
the complexity dichotomy Theorem  \ref{thm:dich-sym-Boolean}.
\end{proof}

The two lemmas above solve the case $p \neq 0$ or $q \neq 0$. Now, we
consider $p=q=0$, and $\mathbf{F}=[u; t, r; s, 0, 0; 1, 0, 0, 0]$:
\begin{center}
\begin{tabular}{*{11}{c}c}
{ }     & { }   & { }   &$u$& { }     & { }   & { }  \\
{ }     & { }   &$t$& { } &$r$& { }     & { }  \\
{ }     &$s$& { } &  $0$& { }   &$0$& { }  \\
 $1$& { } &$0$& { } & $0$& { } &   $0$  \\
\end{tabular}
\end{center}

If $r=0$, the domain is separated, and this is handled as before.

Now we suppose $r \neq 0$.
\begin{lemma}\label{p=q=0-andr-not-0}
For $\mathbf{F}$ given in (\ref{bottom1000F})
where $p=q=0$ and $r \not =0$, (i.e., $\mathbf{F}$
is given in the table above with $r \not
=0$), the problem
Holant${^*}(\mathbf{F})$ is \#P-hard, unless
$s=t=0$ and the domain is separated (in which case
Theorem~\ref{thm:ternary} holds).
\end{lemma}

\begin{proof}
 Consider
$\langle (0,1,x), \mathbf{F}\rangle =[t+xr;s,0;1,0,0]$, for any complex $x$.
This is
\begin{center}
\begin{tabular}{*{11}{c}c}
{ }     & { }   & { }   &$t+xr$& { }     & { }   & { }  \\
{ }     & { }   &$s$    & { } &$0$& { }     & { }  \\
{ }     &$1$& { } &  $0$& { }   &$0$& { }  \\
\end{tabular}
\end{center}
We can pick an $x$, and use it
to realize $=_{B,G}$ by interpolation. Then we can utilize $=_{B,G}$ to
reduce a \#P-hard problem
on the Boolean domain $\{B, G\}$ to Holant${^*}(\mathbf{F})$.

Over the domain $\{B, G\}$, we try to prove the following:
We construct a function not in $\langle \mathcal{T} \rangle$,
and
 for each tractable class $\mathcal{P}_{a,b}$ and
$\mathcal{P}$, we construct a binary function not in it.
We construct binary
functions first.


\begin{itemize}
\item $s \neq 0$

By choosing an $x$, we can realize a   non-degenerate
 binary function $[t+xr,s,1]$ in domain
$\{B,G\}$ using
$\langle (0,1,x), \mathbf{F}\rangle$ and
$=_{B,G}$. The rest of the proof is the
same:
 For each
tractable class $\mathcal{P}_{a,b}$ and $\mathcal{P}$,
 we find a suitable $x$ such that a non-degenerate $[t+xr,s,1]$
does not belong to $\mathcal{P}_{a,b}$ and
$\mathcal{P}$.

\item $s =0$

\subitem{$\diamond$}   $t=0$

The domain is separated.

\subitem{$\diamond$} $t \neq 0$

For any $x$, we can realize a non-degenerate
 binary function $[u+xr,t,0]$  in domain
$\{B,G\}$ using
$\langle (1,0,x), \mathbf{F}\rangle$  and $=_{B,G}$. For each
tractable class $\mathcal{P}_{a,b}$ and $\mathcal{P}$, there is some
$x$ such that $[u+xr,t,0]$ is non-degenerate
and not in the class.

\end{itemize}

%
%
%
%
%
%
%

Now we construct a suitable ternary function.
Obviously, we can realize the ternary function $[u,t,s,1]$ in domain
$\{B,G\}$, using $=_{B,G}$.
If it is not in $\langle \mathcal{T} \rangle$, we are done.

Suppose  we have $[u,t,s,1] \in \langle \mathcal{T} \rangle$.
A {\it symmetric} ternary signature $[u,t,s,1]$
being decomposable in $\langle \mathcal{T}\rangle$
can only be degenerate, of the form $(s, 1)^{\otimes 3}$.
That is, $\mathbf{F}=[s^3; s^2, r; s, 0, 0; 1, 0, 0, 0]$.
\begin{center}
\begin{tabular}{*{11}{c}c}
{ }     & { }   & { }   &$s^3$& { }     & { }   & { }  \\
{ }     & { }   &$s^2$& { } &$r$& { }     & { }  \\
{ }     &$s$& { } &  $0$& { }   &$0$& { }  \\
 $1$& { } &$0$& { } & $0$& { } &   $0$  \\
\end{tabular}
\end{center}

We construct $\langle (1,-s,1), \mathbf{F}
\rangle =[r;0,r;0,0,0]$, which is
\begin{center}
\begin{tabular}{*{11}{c}c}
{ }     & { }   & { }   &$r$& { }     & { }   & { }  \\
{ }     & { }   &$0$    & { } &$r$& { }     & { }  \\
{ }     &$0$& { } &  $0$& { }   &$0$& { }  \\
\end{tabular}
\end{center}

We can use this
to realize $=_{B,R}$
by interpolation, using the fact that
$r \not =0$.
Then on domain $\{B, R\}$,
the problem
 Holant$^*(\mathbf{F'})$ is \#P-hard
for $\mathbf{F'} = \mathbf{F}^{*\rightarrow\{B,R\}}=[s^3,r,0,0]$.

\end{proof}

\subsection{$\mathbf{F}^{*\rightarrow\{G,R\}}$ Is Degenerate
Of Rank 1 And Isotropic}\label{Degenerate-Rank1-Isotropic}


Suppose
$\mathbf{F}^{*\rightarrow\{G,R\}}$ is degenerate of rank 1, therefore
of the form
$(a, b)^{\otimes 3}$, however in this subsection
 we assume $(a, b)$ is isotropic.
The high level proof strategy is  similar to that of
subsection~\ref{bottom1000}, but  the execution is
considerably more complicated.
We  only need to prove the case when $\mathbf{F}^{*\rightarrow\{G,R\}}$
is  $[1,i,-1,-i]$.
The case with $[1, -i, -1, i]$ is the same, and follows formally
by taking the conjugation.
Let $\mathbf{F}=[u; t, r; s, p, q; 1, i, -1, -i]$, namely
\begin{equation}\label{bottom1000F-1in1ni}
\begin{tabular}{*{11}{c}c}
{ }     & { }   & { }   &$u$& { }     & { }   & { }  \\
{ }     & { }   &$t$& { } &$r$& { }     & { }  \\
{ }     &$s$& { } &  $p$& { }   &$q$& { }  \\
 $1$& { } &$i$& { } & $-1$& { } &   $-i$  \\
\end{tabular}
\end{equation}

Suppose $\mathbf{T}=\langle (\alpha,\beta,\gamma),\mathbf{F} \rangle=\alpha \mathbf{A}+\beta \mathbf{B} + \gamma \mathbf{C}$,
where $\mathbf{A} = \mathbf{F}^{1=B}$, $\mathbf{B} = \mathbf{F}^{1=G}$ and $
\mathbf{C}
= \mathbf{F}^{1=R}$. Construct
\[ \mathbf{H}(x_1,x_2,x_3,x_4)=
\left( \sum_{y_1,y_2 \in \{B, G, R\}}
\mathbf{F}(x_1,x_2,y_1) \mathbf{T}(y_1,y_2)  \mathbf{F}(y_2,x_3,x_4)
\right)^{*\rightarrow\{G,R\}}.\]
Because $\mathbf{F}$ is symmetric, $\mathbf{H}$ satisfies the condition in Fact \ref{fact-half-sym-cut2}, and we can use \methps  to prove it is not in $\langle \mathcal{T} \rangle$, by showing two decompositions are impossible.

Let
$\mathbf{S}=\left [
\begin{array}{ccc}  s &1 & i\\ p & i & -1 \\ p & i & -1  \\   q & -1 & -i \end{array} \right ]$,
indexed by $\{G, R\}^2 \times \{B, G, R\}$
in lexicographic order. Then
the arity 4 function $\mathbf{H}$ has a matrix form
 $\mathbf{H}_1=\mathbf{S} \mathbf{T} \mathbf{S}^{\tt T}$,
where the rows are indexed by $(x_1,x_2) \in \{G, R\}^2$
and the columns are  indexed by $(x_3,x_4) \in \{G, R\}^2$.
The other matrix form of $\mathbf{H}$ is $\mathbf{H}_2$ indexed by $(x_1,x_3)$ and $(x_2,x_4)$.
$\mathbf{H}$ has decomposition form $\mathbf{K}(x_1,x_2)\mathbf{L}(x_3,x_4)$ (resp. $\mathbf{K}(x_1,x_3)\mathbf{L}(x_2,x_4)$) iff $\mathbf{H}_1$ (resp. $\mathbf{H}_2$) has rank at most 1.

Let $\mathbf{P}=\left [
\begin{array}{cc}  s &1\\ p & i \\ p & i \\   q & -1 \end{array} \right ]$ and $\mathbf{Q}=\left (
\begin{array}{ccc}  1 &0  & 0\\0 & 1 & i \end{array} \right )$. Then, $\mathbf{S}=\mathbf{P} \mathbf{Q}$.
By associativity, we can multiply $\mathbf{Q}\mathbf{T} \mathbf{Q}^{\tt T}$
first in $\mathbf{H}_1=\mathbf{P} \mathbf{Q} \mathbf{T} \mathbf{Q}^{\tt T}
\mathbf{P}^{\tt T}$.
We have
\[\mathbf{Q} \mathbf{A} \mathbf{Q}^{\tt T}= \left (
\begin{array}{cc}  u & t+ir\\t+ir & s+2ip-q \end{array} \right ),
~~~
\mathbf{Q} \mathbf{B} \mathbf{Q}^{\tt T}= \left (
\begin{array}{cc}  t & s+ip\\s+ip & 0 \end{array} \right ),
~~~
\mathbf{Q} \mathbf{C} \mathbf{Q}^{\tt T}= \left (
\begin{array}{cc}  r & p+iq\\p+iq & 0 \end{array} \right ),\]
and
\[\mathbf{Q} \mathbf{T} \mathbf{Q}^{\tt T}= \left (
\begin{array}{cc} u  \alpha + t \beta + r \gamma  & (t+ir) \alpha
+(s+ip)  \beta +(p+iq) \gamma \\(t+ir)\alpha +(s+ip)\beta +(p+iq)  \gamma & (s+2ip-q)\alpha  \end{array} \right ).\]

\begin{lemma} \label{lem: four-nary-H-not-in-T}
If $p \neq is$ or $q \neq ip$, then
there exist some $\alpha,\beta,\gamma$,
such that $\mathbf{H} \not \in \langle \mathcal{T} \rangle$.
\end{lemma}

\begin{proof}
The proofs under both conditions are the same. W.l.o.g. we assume $p \neq is$.
The proof is composed of three steps. We will use different
matrix or vector representations of $\mathbf{H}$.
The goal is to show that there are some $\alpha,\beta,\gamma$,
such that the two matrix forms of $\mathbf{H}$ both have rank at least two.

In the first step, we use the matrix form $\mathbf{H}_1=\mathbf{P} (\mathbf{Q} \mathbf{T} \mathbf{Q}^{\tt T}) \mathbf{P}^{\tt T}$  of $\mathbf{H}$, and show that for some $\alpha,\beta,\gamma$, this matrix has rank at least 2.

The submatrix $\left ( \begin{array}{cc}  s &1\\ p & i \end{array} \right )$ of $\mathbf{P}$ has full rank.
We only need to show that $\left ( \begin{array}{cc}  s &1\\ p & i \end{array} \right ) (\mathbf{Q} \mathbf{T} \mathbf{Q}^{\tt T})  \left ( \begin{array}{cc}  s &1\\ p & i \end{array} \right )^{\tt T}$, a $2 \times 2$ submatrix of $\mathbf{H}_1$, is of full rank. $\det ( \mathbf{Q} \mathbf{T} \mathbf{Q}^{\tt T} )$
is a polynomial whose coefficient of $\beta^2$ is the nonzero number
$-(s+ip)^2$.
For any fixed $\alpha$ and $\gamma$, there are 3 different values $c_1,c_2,
(c_1+c_2)/2$ (these may depend on $\alpha,\gamma$), such that when $\beta$
takes any one of
 these values, $\det ( \mathbf{Q} \mathbf{T} \mathbf{Q}^{\tt T} ) \neq 0$
and consequently $\mathbf{H}_1$ has rank at least 2.

In the second and third steps we mainly consider the rank of
$\mathbf{H}_2$. This is done in a two-step process.
Either we establish that $\mathbf{H}_2$ has rank least 2
for some setting $c_1,c_2,
(c_1+c_2)/2$ for which  $\mathbf{H}_1$ also has rank at least 2,
or we get some additional condition.
Then in the third step we establish the existence of the required
 $\alpha,\beta,\gamma$ for both $\mathbf{H}_1$ and  $\mathbf{H}_2$  under
the additional condition.

We start our second step by considering the matrix
form $\mathbf{H}_2=\mathbf{H}_{x_1x_3,~x_2x_4}$. If
at least one of the three matrices $\mathbf{H}_2$
given by $(\alpha,c_1,\gamma)$, $(\alpha,c_2,\gamma)$
and $(\alpha,(c_1+c_2)/2,\gamma)$ has rank at least two,
by \methps, $\mathbf{H} \not \in \langle \mathcal{T} \rangle$.

Now suppose all
three matrices $\mathbf{H}_2$ given by $(\alpha,c_1,\gamma)$,
$(\alpha,c_2,\gamma)$ and $(\alpha,(c_1+c_2)/2,\gamma)$ have rank at most 1.
By the conclusion $\det ( \mathbf{Q} \mathbf{T} \mathbf{Q}^{\tt T} ) \neq 0$ in the first step, $\mathbf{H}_2 \not = \mathbf{0}$ for all three settings,
 their ranks are exactly 1.

Let the matrices $\mathbf{H}_2$ given by
$(\alpha,c_1,\gamma)$ and $(\alpha,c_2,\gamma)$
be $\mathbf{u} \mathbf{u}^{\tt T}$
and $\mathbf{v} \mathbf{v}^{\tt T}$ for some
column vectors $\mathbf{u}$ and $\mathbf{v}$.
Then the matrix  $\mathbf{H}_2$ given by
 $(\alpha,(c_1+c_2)/2,\gamma)$ is
$(\mathbf{u} \mathbf{u}^{\tt T} +  \mathbf{v} \mathbf{v}^{\tt T})/2$.
If  $\mathbf{u}$ and $\mathbf{v}$ are linearly independent,
then $(\mathbf{u} \mathbf{u}^{\tt T} +  \mathbf{v} \mathbf{v}^{\tt T})/2$
has rank 2.
(It certainly has rank at most two, since its image as a linear map is contained in
the span of $\{\mathbf{u}, \mathbf{v}\}$.
By linear independence, there are $\mathbf{w}$ satisfying
$\mathbf{u}^{\tt T} \mathbf{w} =0$ but $\mathbf{v}^{\tt T} \mathbf{w} \neq 0$.
Thus the image contains $\mathbf{v}$,
and similarly it also contains $\mathbf{u}$.)
Hence $\mathbf{u}$ and $\mathbf{v}$ are linearly dependent.
It follows that the matrices  $\mathbf{u} \mathbf{u}^{\tt T}$
and $\mathbf{v} \mathbf{v}^{\tt T}$ are also linearly dependent.
This linear dependence remains the same when we write these two matrices as vectors.
%


We use the vector form $\mathbf{H}_3$ of $\mathbf{H}$
to show the consequence of this observation.
This form helps to explain getting rid of $\mathbf{P}$ and $\mathbf{P}^{\tt T}$.
Let $\widetilde{\mathbf{A}}$ denote
the column vector  form of $\mathbf{Q} \mathbf{A} \mathbf{Q}^{\tt T}$,
namely $\widetilde{\mathbf{A}} = (u, t+ir, t+ir, s+2ip-q)^{\tt T}$.
Similarly, let $\widetilde{\mathbf{B}} = (t,s+ip,s+ip,0)^{\tt T}$
and $\widetilde{\mathbf{C}} = (r,p+iq,p+iq,0)^{\tt T}$
be the column vector  forms of  $\mathbf{Q} \mathbf{B} \mathbf{Q}^{\tt T}$
and  $\mathbf{Q} \mathbf{C} \mathbf{Q}^{\tt T}$, respectively.
Then
$\mathbf{H}_3=\mathbf{P}^{\otimes 2} (\alpha \widetilde{\mathbf{A}}
+\beta \widetilde{\mathbf{B}} + \gamma \widetilde{\mathbf{C}})$,
which lists all entries of $\mathbf{H}$,
and therefore also  all entries of  $\mathbf{H}_2$, in some order.
Notice that
the submatrix $\left ( \begin{array}{cc}  s &1\\ p & i \end{array} \right )^{\otimes 2}$ of $\mathbf{P}^{\otimes 2}$ is of
full rank.
 Let $\alpha=1$ and $\gamma=0$, we get
 $\widetilde{\mathbf{A}}+c_1 \widetilde{\mathbf{B}}$ and
 $\widetilde{\mathbf{A}} +c_2 \widetilde{\mathbf{B}}$ are linearly dependent,
where $c_1 \not = c_2$.
It follows that
 $\widetilde{\mathbf{A}}$ and $\widetilde{\mathbf{B}}$ are linearly dependent.
Because the entry $s+ip$ in $\widetilde{\mathbf{B}}$ is nonzero,
$\widetilde{\mathbf{A}}$ is a multiple of $\widetilde{\mathbf{B}}$,
and  $s+2ip-q=0$ as the corresponding entry in $\widetilde{\mathbf{B}}$
is 0.  This is just $(s+ip)+i(p+iq)=0$. Hence we have $p+iq = i(s+ip) \neq 0$.



In the third step, we fix $\alpha=0,\beta=1,\gamma=0$.
Obviously, $\mathbf{H}_1$ has rank at least 2, since
$\det(\mathbf{Q} \mathbf{B} \mathbf{Q}^{\tt T}) = -(s+ip)^2
 \not =0$.
We consider $\mathbf{H}_2$.
Since the matrix
$\left ( \begin{array}{cc}  s &1\\ p & i \end{array} \right )$
has rank 2, and
$\begin{pmatrix} t \\ s+ip \end{pmatrix}$ is a nonzero vector,
we have
either $\left ( \begin{array}{cc}  s & 1 \end{array} \right ) \left ( \begin{array}{c}  t \\ s+ip \end{array} \right ) \neq 0$ or $\left ( \begin{array}{cc}  p & i \end{array} \right ) \left ( \begin{array}{c}  t \\ s+ip \end{array} \right ) \neq 0$.

Suppose the first is not zero.
Consider the $(GG, GR) \times (GG,GR)$
submatrix of $\mathbf{H}_2$, whose row index is by $x_1x_3$
and the column index is by $x_2 x_4$. They are
precisely the entries
in the first row $(GG,GG),(GG,GR), (GG,RG)$ and $(GG,RR)$
of $\mathbf{H}_1$. Recall that
$\mathbf{H}_1=
\mathbf{P} (\mathbf{Q} \mathbf{T} \mathbf{Q}^{\tt T}) \mathbf{P}^{\tt T}
= \mathbf{P} (\mathbf{Q} \mathbf{B} \mathbf{Q}^{\tt T}) \mathbf{P}^{\tt T}
$, after our choice $\alpha=0,\beta=1,\gamma=0$.
The first row of $\mathbf{P}$ is
$\left ( \begin{array}{cc}  s & 1 \end{array} \right )$.
Let
\[\left ( \begin{array}{cc}  a & b \end{array} \right )
= \left ( \begin{array}{cc}  s & 1 \end{array} \right )   \left (
\begin{array}{cc}  t & s+ip\\s+ip & 0 \end{array} \right ).\]
Then the first row of $\mathbf{H}_1$ is
\[\left ( \begin{array}{cc}  s & 1 \end{array} \right ) \mathbf{Q} \mathbf{B} \mathbf{Q}^{\tt T} \mathbf{P}^{\tt T}=
\left ( \begin{array}{cc}  s & 1 \end{array} \right )   \left (
\begin{array}{cc}  t & s+ip\\s+ip & 0 \end{array} \right )   \left (
\begin{array}{cccc}  s & p & p & q\\1 & i & i & -1 \end{array} \right )=  \left ( \begin{array}{cc}  a & b \end{array} \right )  \left (
\begin{array}{cccc}  s & p & p & q\\1 & i & i & -1 \end{array} \right ).\]
Because $s+2ip-q=0$, which we proved in the second step,  we have the linear
dependence
 $ \left ( \begin{array}{c}s\\1 \end{array} \right )+2i \left ( \begin{array}{c}p\\i \end{array} \right )- \left ( \begin{array}{c}q\\-1 \end{array} \right )=0$. Therefore the four entries
in the first row of $\mathbf{H}_1$  are $(k,l,l,k+2il)$,
where $k = as + b$ and $l = ap + bi$.
If the submatrix of $\mathbf{H}_2$ indexed by
$(GG, GR) \times (GG,GR)$
is not of full rank, then $l^2=k(k+2il)$,
which is $(l-ik)^2=0$.
Hence $l = ik$.
It follows that  $ap=ias$. Notice
that $a = \left ( \begin{array}{cc}  s & 1 \end{array} \right ) \left ( \begin{array}{c}  t \\ s+ip \end{array} \right ) \neq 0$. We get $p=is$, a contradiction.

If $\left ( \begin{array}{cc}  p & i \end{array} \right ) \left ( \begin{array}{c}  t \\ s+ip \end{array} \right ) \neq 0$, the proof is similar.
Consider the $(GG, GR) \times (RG,RR)$ submatrix of $\mathbf{H}_2$
indexed by $x_1x_3$ and $x_2 x_4$.
They are precisely the second row entries
$(GR,GG),(GR,GR), (GR,RG)$ and $(GR,RR)$  of $\mathbf{H}_1$. The rest of
the proof is the same as in the previous case.

\end{proof}

It is straightforward  that $[p = is \text{ and } q = ip]$ iff $[s+q =0
 \text{ and } s + 2ip -q = 0]$.  In the next three lemmas we will
complete the case stipulated in this Section~\ref{Degenerate-Rank1-Isotropic},
namely $\mathbf{F}^{*\rightarrow\{G,R\}}$ is degenerate
of rank 1 and isotropic,
  when the negation $[p \neq is \text{ or } q \neq ip]$
holds. We will show that Holant${^*}(\mathbf{F})$
is \#P-hard in this case.  The
method is to construct a suitable binary signature,
or to show directly the problem is \#P-hard.
After that we will deal with the case $[p = is \text{ and } q = ip]$.




\begin{lemma} \label{lem: some-good-binary-except-2-cases}
If $p \not = is \text{ or } q \not = ip$,
or equivalently,  if
$s+q \neq 0 \text{ or } s+2ip-q \neq 0$, then for any $(a,b) \neq (0,0)$, we can construct a nondegenerate symmetric binary function $\mathbf{W} \not \in \mathcal{P}_{a,b}$, and a nondegenerate symmetric binary function $\mathbf{W} \not \in \mathcal{P}$,
except in two cases where this simple construction does not work:
\begin{itemize}
\item{Case 1.} For $\mathcal{P}_{i,-2}$, when $s+q \neq 0$ and $s+2ip-q=0$;
\item{Case 2.} For $\mathcal{P}$,  when $s+q = 0$ and $s+2ip-q \neq 0$.
\end{itemize}
\end{lemma}


\begin{proof}

For any complex number $x$, we
can construct $\mathbf{W}=\langle F,(1,x,0) \rangle^{* \rightarrow \{G,R\}}= [s,p,q]+ x [1,i,-1]$.
We write
\[\mathbf{W}= [f_0, f_1, f_2] = [s+x, p+xi, q-x].\]
The determinant of the matrix form of the binary signature
 $\mathbf{W}$ is $\det ( \mathbf{W} ) =-(s+2ip -q)x+(sq-p^2)$.

There are 4 requirements related to
the conclusion about this symmetric binary function.
\begin{enumerate}
\item   Nondegenerate, that is, $-(s+2ip -q)x+(sq-p^2) \neq 0$.
\item  $f_0+f_1 \neq 0$, that is,  $s+q \neq 0$.
\item $a(f_0-f_2)+bf_1 \neq 0$, that is,
$\langle (a,b), ((s-q,p) +x (2,i))     \rangle \neq 0$.
\item  $f_0-f_2 \neq 0$, that is,  $2x+s-q \neq 0$.
\end{enumerate}

Each requirement is a polynomial in $x$.
The following conditions can guarantee
 each requirement polynomial in $x$ is not the  zero polynomial respectively.
\begin{enumerate}
\item   $s+q \neq 0 \text{ or } s+2ip-q \neq 0$.

If $s+2ip-q \neq 0$, $\det ( \mathbf{W} )$ is not zero obviously.
If $s+2ip-q =0$ but $s+q \neq 0$, we claim $sq-p^2 \neq 0$.
Assume $p^2 =sq$, then  $(s-q)^2 =-4p^2=-4sq$ which gives $(s+q)^2 =0$. A contradiction.

\item  $s+q  \neq 0$.

\item
\begin{itemize}
\item
Whenever $(a,b)$ is not a multiple of $(i,-2)$, it is always
a nonzero polynomial, since the coefficient of $x$ is $2a+ib$.

\item
$s+2ip -q \neq 0$.

Because $\det \begin{bmatrix} s-q  & p \\  2 & i  \end{bmatrix} =i(s+2ip -q) \neq 0$, for any $(a,b) \neq (0,0)$,
$\begin{bmatrix} a & b \end{bmatrix}
 \begin{bmatrix} s-q  & p \\  2 & i  \end{bmatrix}
\not = \mathbf{0}$ (the zero vector), so that there is some $x$ such that
$\begin{bmatrix} a & b \end{bmatrix} \begin{bmatrix} s-q  & p \\  2 & i  \end{bmatrix} \begin{bmatrix} 1 \\ x \end{bmatrix}
\not = 0$. Hence  $(a,b)$ is not orthogonal to $(s-q,p)+x(2,i)$.
Therefore, for all nonzero $(a,b)$,
the polynomial $\langle (a,b), ((s-q,p) +x (2,i))     \rangle \neq 0$.
\end{itemize}

\item  This is a not a zero polynomial.
\end{enumerate}

Recall the definitions of $\mathcal{P}_{a,b}$ and $\mathcal{P}$ in
(\ref{Fibonacci-P-ab-and-P}). For  $\mathcal{P}_{a,b}$, items 1, 2 and 3 are sufficient conditions, and for $\mathcal{P}$, items 1, 2 and 4 are sufficient.
We have 3 cases according to the values of $s+q$ and $s+2ip-q$.

\begin{enumerate}
\item  $s+q \neq 0$ and $s+2ip-q \neq 0$.

All four conditions are satisfied. For every set
$\mathcal{P}_{a,b}$ and $\mathcal{P}$, we have a proper function not in it.

\item  $s+q \neq 0$ and $s+2ip-q=0$.

All four conditions are satisfied, except for $\mathcal{P}_{i,-2}$ from item 3.

\item  $s+q = 0$ and $s+2ip-q \neq 0$.

All conditions are satisfied, except for $\mathcal{P}$ from item 2.
\end{enumerate}
\end{proof}

For the two exceptional cases of Lemma~\ref{lem: some-good-binary-except-2-cases}
(i.e., exactly one of  $s+q$ and $s+2ip-q$ is 0),
we can not get all the
desired binary functions by the simple construction $\langle \mathbf{F},(\alpha,\beta,\gamma) \rangle$. We go back to analyze the function $\mathbf{H}$ in Lemma \ref{lem: four-nary-H-not-in-T} and some other
more complicated constructions. The proof of
Lemma \ref{lem: four-nary-H-not-in-T} not only establishes
the conclusion   $\mathbf{H} \not \in \langle \mathcal{T} \rangle$.
It uses  \methpoly that we will combine
with additional requirements.



Let $Z_1=\frac{1}{2}
\left ( \begin{array}{cc} 1  & 1 \\ i & -i \end{array} \right )$
and $Z_2= \frac{1}{2}
\left ( \begin{array}{cc} 1  & 1 \\ -i & i \end{array} \right )$.
Then $\mathcal{P}_{i,-2}$
is composed of binary symmetric functions in $\langle Z_1 \mathcal{M} \rangle$.
More concretely,
nondegenerate signatures in  $\mathcal{P}_{i,-2}$ are precisely
functions of the form $Z_1^{\otimes 2} [a, b, 0]$, for $b \not =0$.

\begin{lemma}\label{lm:s+q=not0}
When $\mathbf{F}^{*\rightarrow\{G,R\}}
= [1,i,-1,-i]$
 and $s+q \neq 0$, $s+2ip-q=0$, the problem  Holant${^*}(\mathbf{F})$
 is \#P-hard.
\end{lemma}

\begin{proof}

Let $\mathbf{H}$ be the function  in Lemma \ref{lem: four-nary-H-not-in-T}
with the matrix form
$\mathbf{H}_1
= \mathbf{P} \mathbf{Q} \mathbf{T} \mathbf{Q}^{\tt T} \mathbf{P}^{\tt T}$,
with $\alpha,\beta,\gamma$ to be set as we wish.
By Lemma~\ref{lem: four-nary-H-not-in-T},
we already have some $(\alpha,\beta,\gamma)$, such that
$\mathbf{H} \not \in \langle \mathcal{T} \rangle$.
Under the condition $s + q \not =0$,
we  can further construct a non-degenerate
binary function not in each of
the tractable classes
$\langle H \mathcal{E} \rangle$, or $\langle Z \mathcal{E} \rangle$
for $Z = Z_1$ or $Z_2$,
or $\langle Z_2 \mathcal{M} \rangle$
by Lemma~\ref{lem: some-good-binary-except-2-cases}.
%
Note that when $s + q \not =0$, the only exception to the binary function
construction in Lemma~\ref{lem: some-good-binary-except-2-cases}
is $\mathcal{P}_{i,-2}$
 which corresponds to $\langle Z_1 \mathcal{M} \rangle$.
If we can prove for some $(\alpha,\beta,\gamma)$,
a nondegenerate $\mathbf{H} \not \in
\langle Z_1 \mathcal{M} \rangle$,
then  we will have proved \#P-hardness.

Consider $(Z_1^{-1})^{\otimes 4} \mathbf{H}$, where
 $Z_1^{-1}=\left ( \begin{array}{cc} 1 & -i \\1 & i \end{array} \right )$.
Let $\mathbf{R}=(Z_1^{-1})^{\otimes 2} \mathbf{P}= \left ( \begin{array}{cc} s-2ip-q & 4 \\ s+q & 0 \\ s+q & 0  \\ s+2ip-q & 0 \end{array} \right )
=\begin{pmatrix}
-4ip & 4\\
s+q & 0 \\
s+q & 0  \\
0 & 0
\end{pmatrix}$.
 Then a matrix form
of $(Z_1^{-1})^{\otimes 4} \mathbf{H}$
is  $\mathbf{R} \mathbf{Q} \mathbf{T} \mathbf{Q}^{\tt T} \mathbf{R}^{\tt T}$.
For $(Z_1^{-1})^{\otimes 4} \mathbf{H}$,
we only need to show
that for some $(\alpha,\beta,\gamma)$,
 $(Z_1^{-1})^{\otimes 4} \mathbf{H}
\not \in \langle  \mathcal{M} \rangle$.
This is equivalent to
$\mathbf{H} \not \in  \langle Z_1 \mathcal{M} \rangle$.
Denote $(Z_1^{-1})^{\otimes 4} \mathbf{H}$
by $\widetilde{\mathbf{H}}$.
The domain of $\mathbf{H}$ is $\{G,R\}^{4}$,
and we denote the domain of $\widetilde{\mathbf{H}}$
by $ \{0,1\}^{4}$.

We will show that if at least
one of $u,t,r$ is not zero, then $\widetilde{\mathbf{H}}
\not \in \langle \mathcal{M} \rangle$
for some $(\alpha,\beta,\gamma)$, by \methpoly.  Because $Z_1$ is an invertible matrix,
$\mathbf{H} \in \langle \mathcal{T} \rangle$ iff $\widetilde{\mathbf{H}} \in \langle \mathcal{T} \rangle$. By Lemma~\ref{lem: four-nary-H-not-in-T}, there is a $\widetilde{\mathbf{H}} \not \in \langle \mathcal{T} \rangle$. For this $\widetilde{\mathbf{H}}$,
there are two $2 \times 2$ submatrices of its
matrix forms $\widetilde{\mathbf{H}}_1$ and $\widetilde{\mathbf{H}}_2$
respectively, both are of full  rank.
Thus, the determinants of the two submatrices,
 as polynomials in $(\alpha, \beta, \gamma)$, are nonzero polynomials.
The requirements that their values are nonzero
are the first two conditions in this application of \methpoly.
The last condition is $\widetilde{\mathbf{H}} (0,1,0,1) \not =0$.

Which pairs of submatrices have rank 2 depend on
the specific values, $p, q, r$ etc, of $\mathbf{F}$.
We do not analyze the various cases explicitly.
We know there is always one pair that works. For any case which contains a particular set of
 values for which a particular pair works,
 we use this pair of determinants as the  first two conditions
in this application of the polynomial argument.

The value of $\widetilde{\mathbf{H}}$
 at the input $(0,1,0,1)$ of Hamming weight 2  is
$\left ( \begin{array}{cc} s+q & 0  \end{array} \right ) \mathbf{Q} \mathbf{T} \mathbf{Q}^{\tt T} \left ( \begin{array}{c} s+q \\  0  \end{array} \right )=(s+q)^2(u  \alpha + t \beta + r \gamma)$, which is not the zero polynomial.

By \methpoly, there is an $\widetilde{\mathbf{H}}$
that  satisfies all three polynomial conditions. By Fact \ref{fact-half-sym-cut1} and \ref{fact-half-sym-cut2}, this $\widetilde{\mathbf{H}}$ is indecomposable, because by  \methps every decomposition leads to one of the two special decompositions, that is, rank at most one for $\widetilde{\mathbf{H}}_1$ and $\widetilde{\mathbf{H}}_2$. $\widetilde{\mathbf{H}}$ is indecomposable  implies that if $\widetilde{\mathbf{H}} \in \langle \mathcal{M} \rangle$, then $\widetilde{\mathbf{H}} \in \mathcal{M}$. The third condition says $\widetilde{\mathbf{H}} \not \in \mathcal{M}$.

Now we focus on the case $u=t=r=0$. We apply a  \methsepahr  (see
Fact. \ref{fact-domain-sepa-HR})
by $\mathbf{M}=
\begin{pmatrix} \sqrt{2} & \mathbf{0} \\
               \mathbf{0} & Z_1^{-1} \end{pmatrix}
=  \left ( \begin{array}{ccc}  \sqrt{2} & 0 & 0  \\  0 & 1 & -i  \\  0 & 1 & i  \end{array} \right )$
to the bipartite form of the Holant problem
Holant$^*(=_2 \mid =_{G,R}, \mathbf{F})$. We remark that this
holographic reduction is
 only for the convenience in calculating the signature of
 the gadget to be constructed.

$\mathbf{M}^{\otimes 3} \mathbf{F}$ is given by

\begin{center}
\begin{tabular}{*{11}{c}c}
{ }     & { }   & { }   &$0$& { }     & { }   & { }  \\
{ }     & { }   &$0$& { } &$0$& { }     & { }  \\
{ }     &$-4\sqrt{2}ip$& { } &  $\sqrt{2}(s+q)$& { }   &$0$& { }  \\
 $8$& { } &$0$& { } & $0$& { } &   $0$  \\
\end{tabular}
\end{center}
This calculation can be done per each row in the table for
$\mathbf{F}$.
E.g., the third row in the table is $\mathbf{F}^{1 = B} = [s, p, q]$.
As a column vector it is $\begin{pmatrix} s & p & p & q \end{pmatrix}^{\tt T}$
and is transformed
to $\sqrt{2} (Z_1^{-1})^{\otimes 2}
  [s, p, q]
= \sqrt{2} [s -2ip - q, s+q, s +2ip - q] =
[-4\sqrt{2}ip, \sqrt{2}(s+q), 0]$.

After factoring out the constant 8,
we will write $\mathbf{M}^{\otimes 3} \mathbf{F}$
 as $\widetilde{\mathbf{F}}$:
\begin{center}
\begin{tabular}{*{11}{c}c}
{ }     & { }   & { }   &$0$& { }     & { }   & { }  \\
{ }     & { }   &$0$& { } &$0$& { }     & { }  \\
{ }     &$a$& { } &  $b$& { }   &$0$& { }  \\
 $1$& { } &$0$& { } & $0$& { } &   $0$  \\
\end{tabular}
\end{center}
where $a = -ip/\sqrt{2}$ and $b = \sqrt{2}(s+q)/8$.
Note crucially that $b \neq 0$.
Up to a constant factor,
the problem
Holant$^*(=_2 \mid =_{G,R}, \mathbf{F})$
becomes Holant$^*(\mathbf{L}_1 \mid \mathbf{L}_2,  \widetilde{\mathbf{F}})$,
where
$\mathbf{L}_1=\left ( \begin{array}{ccc}  1 & 0 & 0  \\  0 & 0 & 1 \\  0 & 1 & 0 \end{array} \right )$
and $\mathbf{L}_2=\left ( \begin{array}{ccc}  0 & 0 & 0  \\  0 & 0 & 1 \\  0 & 1 & 0 \end{array} \right )$.

Now we construct the following ternary triangular gadget
 with three
external dangling edges. We denote its signature  as $\mathbf{V}$.
Each of the three vertices incident to the dangling edges
 is assigned $\widetilde{\mathbf{F}}$, and each
of the three  edges of the triangle is
composed of a chain linked by
$\mathbf{L}_1$, $\langle \widetilde{\mathbf{F}},
 (0,1,0) \rangle$ and $\mathbf{L}_1$.
(A simpler triangular gadget does not work here.)
Calculation shows that the function $\mathbf{V}$ of this gadget restricted to $\{G,R\}$ is $\mathbf{V}^{*\rightarrow \{G,R\}}
=[3b^4+16a^3b^3, 8a^2b^4, 4ab^5, 2b^6]=2 b^3 \left ( \begin{array}{c}  2a  \\ b \end{array} \right )^{\otimes 3}+ 3 b^4 \left ( \begin{array}{c}  1  \\ 0 \end{array} \right )^{\otimes 3}$.
We remark that this restriction of $\mathbf{V}$
 in its domain set is only for the purpose of calculation
later on. While $\mathbf{V}$ is realizable in the right hand side
of Holant$^*(\mathbf{L}_1 \mid \mathbf{L}_2,  \widetilde{\mathbf{F}})$,
we do not claim its restriction $\mathbf{V}^{*\rightarrow \{G,R\}}$
is a realizable signature.

The realizability of $\mathbf{V}$ in
Holant$^*(\mathbf{L}_1 \mid \mathbf{L}_2,  \widetilde{\mathbf{F}})$
  means that in problem
Holant$^*(=_2  \mid =_{G,R}, \mathbf{F})$,
 we can realize ${(\mathbf{M}^{-1})}^{\otimes 3} \mathbf{V}$
in the right hand side. As in Lemma \ref{lem: some-good-binary-except-2-cases}, we want to get a nondegenerate symmetric
 binary function $\mathbf{W} \not \in \mathcal{P}_{i,-2}$.
For this purpose, we only need to show that
  $(Z_1^{-1})^{\otimes 2} \mathbf{W} \not \in \mathcal{M}$ and
it is nondegenerate.

The logical process is the following:
 start from ${(\mathbf{M}^{-1})}^{\otimes 3} \mathbf{V}$
in the right hand side of Holant$^*(=_2  \mid =_{G,R}, \mathbf{F})$, we
connect it with an arbitrary unary function $\mathbf{u}$,
and then restrict the input to $\{G,R\}$,
finally perform a holographic transformation by $Z_1^{-1}$.
At this stage we wish to obtain a nondegenerate signature not in
$\langle \mathcal{M} \rangle$ (being nondegenerate
and of arity 2 this is the same as
being not in $\mathcal{M}$),
by setting the unary function appropriately.
The restriction to $\{G,R\}$ is equivalent to
a transformation by $\mathbf{N}_1 = \begin{pmatrix}
0 & 1 & 0\\
0 & 0 & 1
\end{pmatrix}$.
Let $\mathbf{N}= \begin{pmatrix}
 0 & 0 & 0  \\  0 & 1 & 0 \\  0 & 0 & 1
\end{pmatrix}$, then note that
$\mathbf{N}_1 = \mathbf{N}_1 \mathbf{N}$,
and $\mathbf{M} \mathbf{N} = \mathbf{N} \mathbf{M}$.
\begin{eqnarray*}
(Z_1^{-1})^{\otimes 2} \mathbf{N}_1^{\otimes 2}
\langle {(\mathbf{M}^{-1})}^{\otimes 3} \mathbf{V}, \mathbf{u} \rangle
&=&
(Z_1^{-1})^{\otimes 2} \mathbf{N}_1^{\otimes 2}
\mathbf{N}^{\otimes 2}
\langle {(\mathbf{M}^{-1})}^{\otimes 3} \mathbf{V}, \mathbf{u} \rangle \\
&=&
(Z_1^{-1})^{\otimes 2} \mathbf{N}_1^{\otimes 2}
\mathbf{N}^{\otimes 2}
{(\mathbf{M}^{-1})}^{\otimes 2}
\langle \mathbf{V}, \mathbf{u'} \rangle \\
&=&
(Z_1^{-1})^{\otimes 2} \mathbf{N}_1^{\otimes 2}
{(\mathbf{M}^{-1})}^{\otimes 2}
\mathbf{N}^{\otimes 2}
\langle \mathbf{V}, \mathbf{u'} \rangle \\
&=&
(Z_1^{-1})^{\otimes 2} \mathbf{N}_1^{\otimes 2}
{(\mathbf{M}^{-1})}^{\otimes 2}
\langle \mathbf{N}^{\otimes 3} \mathbf{V}, \mathbf{u''} \rangle \\
&=&
(Z_1^{-1})^{\otimes 2}
\begin{pmatrix} \mathbf{0} & Z_1 \end{pmatrix}^{\otimes 2}
\langle \mathbf{N}^{\otimes 3} \mathbf{V}, \mathbf{u''} \rangle \\
&=&
\begin{pmatrix} 0 & 1 & 0 \\ 0 & 0 & 1 \end{pmatrix}^{\otimes 2}
\langle \mathbf{N}^{\otimes 3} \mathbf{V}, \mathbf{u''} \rangle \\
&=&
\langle \mathbf{N}_1^{\otimes 3} \mathbf{V}, \mathbf{u'''} \rangle \\
&=&
\langle \mathbf{V}^{*\rightarrow \{G,R\}}, \mathbf{u'''} \rangle
\end{eqnarray*}
where $\mathbf{u}, \mathbf{u'}, \mathbf{u''}, \mathbf{u'''}$ are
unary signatures, and $\mathbf{u'''}$ is on domain $\{G,R\}$,
and $\mathbf{u'''}$ can be arbitrary.


If $a=0$, take $\mathbf{u'''} =(1,1)$,
we get $\langle \mathbf{V}^{*\rightarrow \{G,R\}}, \mathbf{u'''} \rangle
= 2b^4 \begin{pmatrix} 0 \\ b \end{pmatrix}^{\otimes 2} + 3b^4
\begin{pmatrix} 1 \\ 0 \end{pmatrix}^{\otimes 2} \not \in \langle \mathcal{M} \rangle$.
If $a \neq 0$, take $\mathbf{u'''} =(1,0)$,
we get $\langle \mathbf{V}^{*\rightarrow \{G,R\}}, \mathbf{u'''} \rangle
= 4ab^3 \begin{pmatrix} 2a \\ b \end{pmatrix}^{\otimes 2}
+ 3b^4
\begin{pmatrix} 1 \\ 0 \end{pmatrix}^{\otimes 2}
 \not \in \langle \mathcal{M} \rangle$.
In either case we realized a nondegenerate symmetric
binary signature
$\mathbf{W} =
\mathbf{N}_1^{\otimes 2}
\langle {(\mathbf{M}^{-1})}^{\otimes 3} \mathbf{V}, \mathbf{u} \rangle$
in Holant$^*(\{=_{G,R}, \mathbf{F}\})$,
such that
 $(Z_1^{-1})^{\otimes 2}\mathbf{W} \not \in
\langle  \mathcal{M}  \rangle$,
thus $\mathbf{W} \not \in
\langle  Z_1  \mathcal{M}  \rangle$.


%
%
%
%

\end{proof}

In the next lemma we finish off Case 2
from Lemma~\ref{lem: some-good-binary-except-2-cases}.
We note that
$\mathcal{P}$ is composed of symmetric functions in $\langle Z_1 \mathcal{E} \rangle=\langle Z_2 \mathcal{E} \rangle$.
Note that for $\tau = \begin{pmatrix} 0 & 1 \\ 1 & 0 \end{pmatrix}$,
   $Z_1 = Z_2 \tau$ and $\tau \mathcal{E}
= \mathcal{E}$.

\begin{lemma}\label{lm:s+q=0}
When $\mathbf{F}^{*\rightarrow\{G,R\}}
= [1,i,-1,-i]$
 and $s+q = 0$, $s+2ip-q \neq 0$, the problem
 Holant${^*}(\mathbf{F})$ is \#P-hard.
\end{lemma}

\begin{proof}
We still use the function
$\mathbf{H}$
whose matrix form is
$\mathbf{P} \mathbf{Q} \mathbf{T} \mathbf{Q}^{\tt T} \mathbf{P}^{\tt T}$
from Lemma \ref{lem: four-nary-H-not-in-T}.
We employ the same general approach as in Lemma~\ref{lm:s+q=not0}.
If we can prove for some $(\alpha,\beta,\gamma)$,
a non-degenerate $\mathbf{H} \not \in
\langle Z_1 \mathcal{E} \rangle$, then
we will have proved \#P-hardness.
The existence of $(\alpha,\beta,\gamma)$
such that $\mathbf{H} \not \in \langle \mathcal{T} \rangle$
was already proved by Lemma~\ref{lem: four-nary-H-not-in-T}.
 Note that by Lemma~\ref{lem: some-good-binary-except-2-cases}, we can construct a nondegenerate binary function not in
the other tractable classes $\langle H \mathcal{E} \rangle$, or $\langle Z_1 \mathcal{M} \rangle$, or $\langle Z_2 \mathcal{M} \rangle$.

Consider $\widetilde{\mathbf{H}}
= (Z_1^{-1})^{\otimes 4} \mathbf{H}$, where $Z_1^{-1}=\left ( \begin{array}{cc} 1 & -i \\1 & i \end{array} \right )$.
Let $\mathbf{R}=(Z_1^{-1})^{\otimes 2} \mathbf{P}= \left ( \begin{array}{cc} 2s-2ip & 4 \\ 0 & 0 \\ 0 & 0  \\ 2s+2ip & 0 \end{array} \right )$.
Then a matrix form
of $\widetilde{\mathbf{H}}$ is $\mathbf{R} \mathbf{Q} \mathbf{T} \mathbf{Q}^{\tt T} \mathbf{R}^{\tt T}$. We only need to show
that  for some $(\alpha,\beta,\gamma)$, a non-degenerate
$\widetilde{\mathbf{H}} \not \in \langle \mathcal{E} \rangle$.
This is equivalent to
$\mathbf{H} \not \in \langle Z_1 \mathcal{E} \rangle$.
We will in fact show that  $\widetilde{\mathbf{H}} \not \in
\langle \mathcal{T} \rangle  \cup
\mathcal{E}$ for some $(\alpha,\beta,\gamma)$;
being not in $\langle \mathcal{T} \rangle$ means that $\widetilde{\mathbf{H}}$
is indecomposable, and consequently
$\widetilde{\mathbf{H}} \not \in \langle \mathcal{E} \rangle$
is equivalent to
$\widetilde{\mathbf{H}} \not \in \mathcal{E}$. Utilizing Lemma \ref{lem: four-nary-H-not-in-T} as how it is used in \ref{lm:s+q=not0}, we only need to prove $\widetilde{\mathbf{H}} \not \in \mathcal{E}$.

Obviously, $\widetilde{\mathbf{H}}$ only has
at most 4 nonzero entries
by the form of $\mathbf{R}$.
They are indexed by $\{GG, RR\} \times \{GG, RR\}$, and we
list them as
\[\mathbf{K}=\left ( \begin{array}{cc} a & b \\b & c \end{array} \right )= \mathbf{L}  \mathbf{Q} \mathbf{T} \mathbf{Q}^{\tt T}  \mathbf{L}^{\tt T}  , \]
where $\mathbf{L}= \left ( \begin{array}{cc} 2s-2ip & 4 \\ 2s+2ip & 0 \end{array} \right )$.
The matrix form of $\widetilde{\mathbf{H}}$ with row index $x_1x_2$ and column index $x_3x_4$ is $\left ( \begin{array}{cccc} a & 0 & 0 & b \\0 & 0 & 0  & 0 \\0 & 0 & 0  & 0 \\b & 0 & 0  & c \end{array} \right )$.

Obviously, $\widetilde{\mathbf{H}} \not \in \mathcal{E}$ iff $\mathbf{K} \not \in \mathcal{E}$ iff one column or one row of $\mathbf{K}$ has two nonzero entries.

Recall that, with $s+q=0$,
\[\mathbf{Q} \mathbf{A} \mathbf{Q}^{\tt T}= \left (
\begin{array}{cc}  u & t+ir\\t+ir & 2(s+ip) \end{array} \right ),~~
\mathbf{Q} \mathbf{B} \mathbf{Q}^{\tt T}= \left (
\begin{array}{cc}  t & s+ip\\s+ip & 0 \end{array} \right ),~~
\mathbf{Q} \mathbf{C} \mathbf{Q}^{\tt T}= \left (
\begin{array}{cc}  r & -i(s+ip)\\ -i(s+ip) & 0 \end{array} \right ).\]
\[\mathbf{Q} \mathbf{T} \mathbf{Q}^{\tt T}= \alpha\mathbf{Q} \mathbf{A} \mathbf{Q}^{\tt T}+
\beta \mathbf{Q} \mathbf{B} \mathbf{Q}^{\tt T}+ \gamma
\mathbf{Q} \mathbf{C} \mathbf{Q}^{\tt T}.\]

Note that $s+2ip -q = 2(s+ip) \not =0$.
$\det \begin{pmatrix} 2s-2ip & 4 \\ 2s + 2ip & 0 \end{pmatrix} \not = 0$.

Assume the first rows $(u,t+ir)$, $(t, s+ip)$ and $(r,-i(s+ip))$ are linearly
 independent.
Then the first row of  $\mathbf{Q} \mathbf{T} \mathbf{Q}^{\tt T}$ can be any vector, so that we can pick $(\alpha, \beta, \gamma)$ such that this
row vector when multiplied  to the right by $\mathbf{L}^{\tt T}$
has two nonzero entries.
Consider the second row $\begin{pmatrix} 2s + 2ip & 0 \end{pmatrix}$
 of $\mathbf{L}$.
It follows that the second row of $\mathbf{K}$ has two  nonzero entries.

From now on, we have  $(u,t+ir)$, $(t, s+ip)$ and $(r,-i(s+ip))$ are linearly dependent.
Hence, $r=-it$ and $u(s+ip)=2t^2$.

If $t=0$, we have $r=0$, and $u=0$ since $s+ip \not = 0$.
Then for any $x,y$, $\left (
\begin{array}{cc}  0 & x\\x & y \end{array} \right )$ can be realized by $\mathbf{Q} \mathbf{T} \mathbf{Q}^{\tt T}$.  The first row of $\mathbf{K}$  is   $ (4x, 2(s-ip)x+4y) \mathbf{L}^{\tt T} $, which can be set to
any vector by setting $x$ and $y$, and in particular
we set $x$ and $y$ so that the  first row of $\mathbf{K}$
has two nonzero entries.

From now on, we have $t \neq 0$.
$\mathbf{F}$ becomes
\begin{center}
\begin{tabular}{*{11}{c}c}
{ }     & { }   & { }   &$2 \sqrt{2}u$& { }     & { }   & { }  \\
{ }     & { }   &$0$& { } &$4t$& { }     & { }  \\
{ }     &$2 \sqrt{2}(s-ip)$& { } &  $0$& { }   &$2 \sqrt{2}(s+ip)$& { }  \\
 $8$& { } &$0$& { } & $0$& { } &   $0$  \\
\end{tabular}
\end{center}
 after a \methsepahr $\mathbf{M}=\left ( \begin{array}{ccc}  \sqrt{2} & 0 & 0  \\  0 & 1 & -i  \\  0 & 1 & i  \end{array} \right )$.
Normalizing by the constant factor $1/8$, we have
$\mathbf{G}$, which has the form
\begin{center}
\begin{tabular}{*{11}{c}c}
{ }     & { }   & { }   &$a^2$& { }     & { }   & { }  \\
{ }     & { }   &$0$& { } &$ab$& { }     & { }  \\
{ }     &$c$& { } &  $0$& { }   &$b^2$& { }  \\
 $1$& { } &$0$& { } & $0$& { } &   $0$  \\
\end{tabular}
\end{center}
 where $a = (\sqrt{2} u)^{1/2}/2$, $b = t/(2a)$, and $c = \sqrt{2} (s-ip) /4$.
We have obtained $t, u
\not =0$,  so $a, b \not =0$ and are well-defined.
The fact that the three terms $\sqrt{2}u/4$, $t/2$ and $\sqrt{2}(s+ip)/4$
form a geometric progression $a^2$, $ab$ and $b^2$ follows from
the conditions we have proved.
We can verify that $b^2 =  \sqrt{2} (s+ip) /4$, and $c = - b/a$.
The problem becomes
Holant$^*(\mathbf{L}_1 \mid \mathbf{L}_2,  \mathbf{G})$,
where $\mathbf{L}_1$ and $\mathbf{L}_2$ are as before.

%

We use the triangular gadget again.
If we use the chain
$\mathbf{L}_1, \langle\mathbf{G},(0,0,1) \rangle, \mathbf{L}_1$ at each edge,
when restricted to $\{G, R\}$,
we get  the edge signature
$[- 2 b^9/a^3,   b^8,  -a^3 b^7,  a^6 b^6 ]$.
Up to a nonzero factor, this is $[ -2, d, -d^2, d^3]$,
 where $d = a^3/b \not = 0$.
This is nondegenerate.  We can realize
Holant$^*([0,1,0]|[-2, d, -d^2, d^3])$.
%
Under a holographic transformation,
Holant$^*([0,1,0]|[-2, d, -d^2, d^3])$
is equivalent to
\[ {\rm Holant}^*\left ([1,0,1] \left | \left ( \begin{array}{cc} 1 & 1 \\ i & -i
\end{array}  \right )^{\otimes 3}  [-2, d, -d^2, d^3] \right). \right . \]
Write $[-2, d, -d^2, d^3]
= - \begin{bmatrix} 1 \\ 0 \end{bmatrix}^{\otimes 3}
- \begin{bmatrix} 1 \\ -d \end{bmatrix}^{\otimes 3}$, then
\[\begin{bmatrix} 1 & 1 \\ i & -i \end{bmatrix}^{\otimes 3}
[-2, d, -d^2, d^3]
= - \begin{bmatrix} 1 \\ 1 \end{bmatrix}^{\otimes 3}
- \begin{bmatrix} 1-d \\ i + di \end{bmatrix}^{\otimes 3}.\]
This is not in any tractable classes of  Theorem \ref{thm:dich-sym-Boolean}.

\end{proof}

We summarize the previous few lemmas in this subsection so far:
\begin{corollary}\label{cor:summary-pnotis-or-qnotip}
When $\mathbf{F}^{*\rightarrow\{G,R\}}
= [1,i,-1,-i]$
 and $p \not = is$ or $q \neq ip$, the problem Holant$^*(\mathbf{F})$
 is \#P-hard.
\end{corollary}

If $p = is$ and $q = ip$, the signature $\mathbf{F}$ is of the form

\begin{equation}\label{bottom1000F-1in1ni-sisns}
\begin{tabular}{*{11}{c}c}
{ }     & { }   & { }   &$u$& { }     & { }   & { }  \\
{ }     & { }   &$t$& { } &$r$& { }     & { }  \\
{ }     &$s$& { } &  $is$& { }   &$-s$& { }  \\
 $1$& { } &$i$& { } & $-1$& { } &   $-i$  \\
\end{tabular}
\end{equation}

If  $r=i t$, $\mathbf{F}$ is in the third form of Theorem \ref{thm:ternary},
since $\langle (0, 1, i), \mathbf{F} \rangle = \mathbf{0}$.
Hence Holant$^*(\mathbf{F})$ is tractable.

Now we suppose $r \not =i t$. We shall prove that the problem is
\#P-hard.
We first consider the gadget in Figure~\ref{triangle-in-sec-6.5}.
        \begin{figure}[hbtp]
        \begin{center}
                \includegraphics[width=3 in]{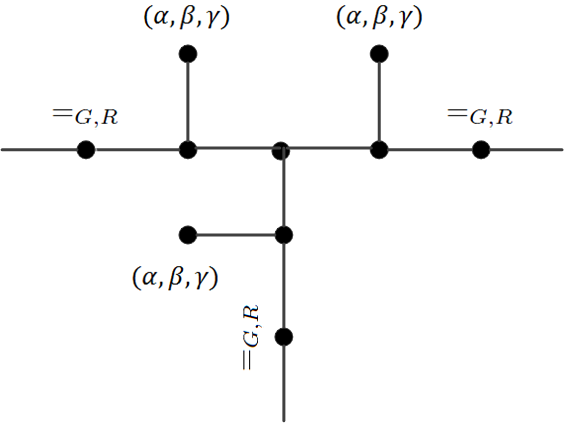}
        \caption{A ternary gadget.}
        \label{triangle-in-sec-6.5}
        \end{center}
\end{figure}
We choose a unary $\mathbf{u}=(\alpha, \beta, 0)$, such that
the matrix form of $\langle \mathbf{u}, \mathbf{F}
\rangle$ is $M =
\begin{bmatrix} w & x & y \\ x & z & iz \\ y & iz & -z \end{bmatrix}$,
where
$w= \alpha u + \beta t$, $x= \alpha t + \beta s$, $y= \alpha r + \beta  i s$
 and $z = \alpha  s +  \beta$.

We wish to compute the signature of this construction, namely the ternary
function $(M^{\otimes 3} \mathbf{F})^{* \rightarrow \{G,R\}}$ on domain
size two.
Consider the $2 \times 3$ matrix
$\begin{bmatrix} x & z & zi \\ y & zi & -z \end{bmatrix}$,
which we decompose to $M_1 M_2$, where
$M_1 = \begin{bmatrix} x & z & 0 \\ y & zi & 0 \end{bmatrix}$,
and
$M_2 =
\begin{bmatrix} 1 & 0 & 0 \\ 0 & 1 & i \\ 0 & 0 & 0 \end{bmatrix}$.
Then
$(M^{\otimes 3} \mathbf{F})^{* \rightarrow \{G,R\}}$ is
$ M_1^{\otimes 3} (M_2^{\otimes 3} \mathbf{F})$.

As $M_2$ has a separated domain form,
$(M_2^{\otimes 3} \mathbf{F})^{* \rightarrow \{G, R\}}$
is identically 0, since these values are combinations of values
from the bottom line $\mathbf{F}^{* \rightarrow \{G, R\}}$.
Formally,
\[(M_2^{\otimes 3} \mathbf{F})^{* \rightarrow \{G, R\}}
= \begin{bmatrix} 0 & 1 & i \\ 0 & 0 & 0 \end{bmatrix}^{\otimes 3}
\mathbf{F}
= \begin{bmatrix} 1 & i \\ 0 & 0 \end{bmatrix}^{\otimes 3}
\begin{bmatrix} 0 & 1 & 0 \\ 0 & 0 & 1 \end{bmatrix}^{\otimes 3}
\mathbf{F}
= \begin{bmatrix} 1 & i \\ 0 & 0 \end{bmatrix}^{\otimes 3}
\mathbf{F}^{* \rightarrow \{G, R\}} =
\begin{bmatrix} 1 & i \\ 0 & 0 \end{bmatrix}^{\otimes 3}
\begin{bmatrix} 1 \\ i \end{bmatrix}^{\otimes 3}
=  \mathbf{0}.\]

To compute the other values of $M_2^{\otimes 3} \mathbf{F}$,
we may set one input of $M_2^{\otimes 3} \mathbf{F}$ to $B$,
which by the form of $M_2$ is the same as
$M_2^{\otimes 2} (\mathbf{F}^{1=B})$.
This can be computed as a matrix product
$M_2 \begin{bmatrix} u & t & r \\ t & s & is \\ r & is & -s \end{bmatrix}
M_2^{\tt T}$, and the result is
$\begin{bmatrix} u & t+ir & 0 \\ t+ir & 0 & 0 \\ 0 & 0 & 0 \end{bmatrix}$.
Thus the signature $M_2^{\otimes 3} \mathbf{F}$ is
\begin{center}
                \begin{tabular}{*{11}{c}c}
                { }     & { }   & { }   &$u$& { }     & { }   & { }  \\
                { }     & { }   &$t+ir$& { } &$0$& { }     & { }  \\
                { }     &$0$& { } &  $0$& { }   &$0$& { }  \\
                 $0$& { } &$0$& { } & $0$& { } &   $0$  \\
                \end{tabular}
                \end{center}
which is
\[u \begin{bmatrix} 1 \\0 \\0 \end{bmatrix}^{\otimes 3}
+ (t+ir) \cdot \frac{1}{2}{\rm Sym} \left[
\begin{bmatrix} 1 \\0 \\0 \end{bmatrix}^{\otimes 2}
\otimes
\begin{bmatrix} 0 \\ 1 \\0 \end{bmatrix}
\right].
\]
(The symmetrization Sym has six terms.)

Now we apply $M_1$, and get
\[u \begin{bmatrix} x \\y \end{bmatrix}^{\otimes 3}
+ (t+ir) \cdot \frac{z}{2} {\rm Sym} \left[
\begin{bmatrix} x \\y  \end{bmatrix}^{\otimes 2}
\otimes
\begin{bmatrix}  1 \\i \end{bmatrix}
\right].
\]

If we can set $y=0$, with  $x \neq  0$ and $z \not = 0$, then
this signature has the form $[a, b, 0, 0]$, with $b \not =0$.
This defines a \#P-hard problem on domain size two by Theorem \ref{thm:dich-sym-Boolean}.
Similarly if we can set $x=0$, with  $y \neq  0$ and $z \not = 0$, then
this signature has the form $[0, 0, b, a]$, with $b \not =0$. This
also defines a \#P-hard problem by Theorem \ref{thm:dich-sym-Boolean}.

We will now show that if $s\not =0$, or if $s=0$ but $r \not = -i t$,
then we can indeed set
$x$, $y$ and $z$ accordingly, and we will have proved the \#P-hardness
of Holant$^*(\mathbf{F})$.

First suppose $s \not =0$.
Since we have $r \not = it$, either  $t \not =s^2$ or $r \not = i s^2$.
Set any $\alpha \not =0$. If $t \not =s^2$,
then set $\beta = - \alpha t/s$ we get $x= \alpha t + \beta s = 0$,
and $y = \alpha r + \beta i s =  \alpha (r - it) \not =0$,
and $z = \alpha s +  \beta =  \alpha ( s^2 - t)/s \not =0$.
Similarly if $r \not = i s^2$, then  set
$\beta =  \alpha i r/s$ we get $y= \alpha r + \beta i s =0$, and
$x = \alpha t + \beta s =  \alpha (t + ir) \not =0$,
and $z = \alpha s +  \beta =  \alpha (s^2 + ir)/s \not =0$.

Now suppose $s =0$. Then $x= \alpha t$, $y=\alpha  r$ and $z=\beta$.
We set $\beta = 1/(t+ir)$ (recall that we have $r \not = it$). Then
$(M^{\otimes 3} \mathbf{F})^{* \rightarrow \{G,R\}}$ is
$\alpha^2 \widehat{f}$,
where
\begin{equation}\label{f-hat-with-single-alpha}
\widehat{f} = \alpha u \begin{bmatrix} t \\r \end{bmatrix}^{\otimes 3}
+  \frac{1}{2} {\rm Sym} \left[
\begin{bmatrix} t \\r  \end{bmatrix}^{\otimes 2}
\otimes
\begin{bmatrix}  1 \\i \end{bmatrix}
\right]
\end{equation}

For $\alpha \not =0$, we can ignore the factor $\alpha^2$.
When $u=0$ we can show the signature
\[ f =
\frac{1}{2} {\rm Sym} \left[
\begin{bmatrix} t \\r  \end{bmatrix}^{\otimes 2}
\otimes
\begin{bmatrix}  1 \\i \end{bmatrix}
\right] = [3 t^2, t^2 i + 2 tr, 2tr i + r^2, 3 r^2 i]
\]
gives a \#P-hard Holant$^*$ problem as follows, by Theorem \ref{thm:dich-sym-Boolean}. First, $f$ is nondegenerate.
 After a holographic
reduction $\begin{bmatrix} t & 1 \\ r & i \end{bmatrix}^{-1}$
this signature is $\tilde{f} = \frac{1}{2} {\rm Sym} \left[
\begin{bmatrix} 1 \\0  \end{bmatrix}^{\otimes 2}
\otimes
\begin{bmatrix}  0 \\1 \end{bmatrix}
\right]$. $f$ is degenerate iff $\tilde{f}$ is,
and if this were the case, then  $\tilde{f}  =
\begin{bmatrix} a \\ b \end{bmatrix}^{\otimes 3}$, for some $a$ and $b$.
Taking $\langle \begin{bmatrix} 1& 0 \end{bmatrix}^{\otimes 2},
\tilde{f} \rangle$, we get  $a=0$, and similarly,
$\langle \begin{bmatrix} 0 & 1 \end{bmatrix}^{\otimes 2},
\tilde{f} \rangle$ gives us $b=0$, a contradiction.
Checking against the tractability criterion of  Theorem \ref{thm:dich-sym-Boolean},
we find that Holant$^*(f)$ is \#P-hard, unless $t^2 + r^2 =0$.
As $t+ir \not =0$, we get the only exceptional case $t - ir =0$.


Now we claim that the proof above also shows that for all $u$
not necessarily 0, the problem is \#P-hard, assuming $t - ir \not = 0$.
This is because the
conditions on degeneracy and on tractability are all expressed
in terms polynomial equations on the entries of the signature.
For any fixed $u$, if the conditions fail to be satisfied for
$\widehat{f}$ at $\alpha =0$
(which is the same as when $u=0$ in (\ref{f-hat-with-single-alpha}),
as has been shown),
then the conditions also fail to be satisfied for
some nonzero $\alpha$ sufficiently small.
This shows that for all $u$ the problem Holant$^*(\widehat{f})$,
for some nonzero $\alpha$, is \#P-hard.
But for any nonzero $\alpha$, the problem Holant$^*(\widehat{f})$
is equivalent to Holant$^*((M^{\otimes 3} \mathbf{F})^{* \rightarrow \{G,R\}})$.
Hence  Holant$^*(\mathbf{F})$ is \#P-hard.

Now we suppose
$s=0$, $t  = i r$ and $u\neq 0$. As $t + ir \not =0$ we have $t \not =0$.
We will use a slightly more complicated
gadget as depicted in Figure~\ref{figure:extended-three-star}, where the outer unary
function is $\mathbf{u}_1 = (1/t, 1, 0)$,
and the inner unary function is $\mathbf{u}_2 = (x, 1, 0)$.
        \begin{figure}[hbtp]
        \begin{center}
                \includegraphics[width=3 in]{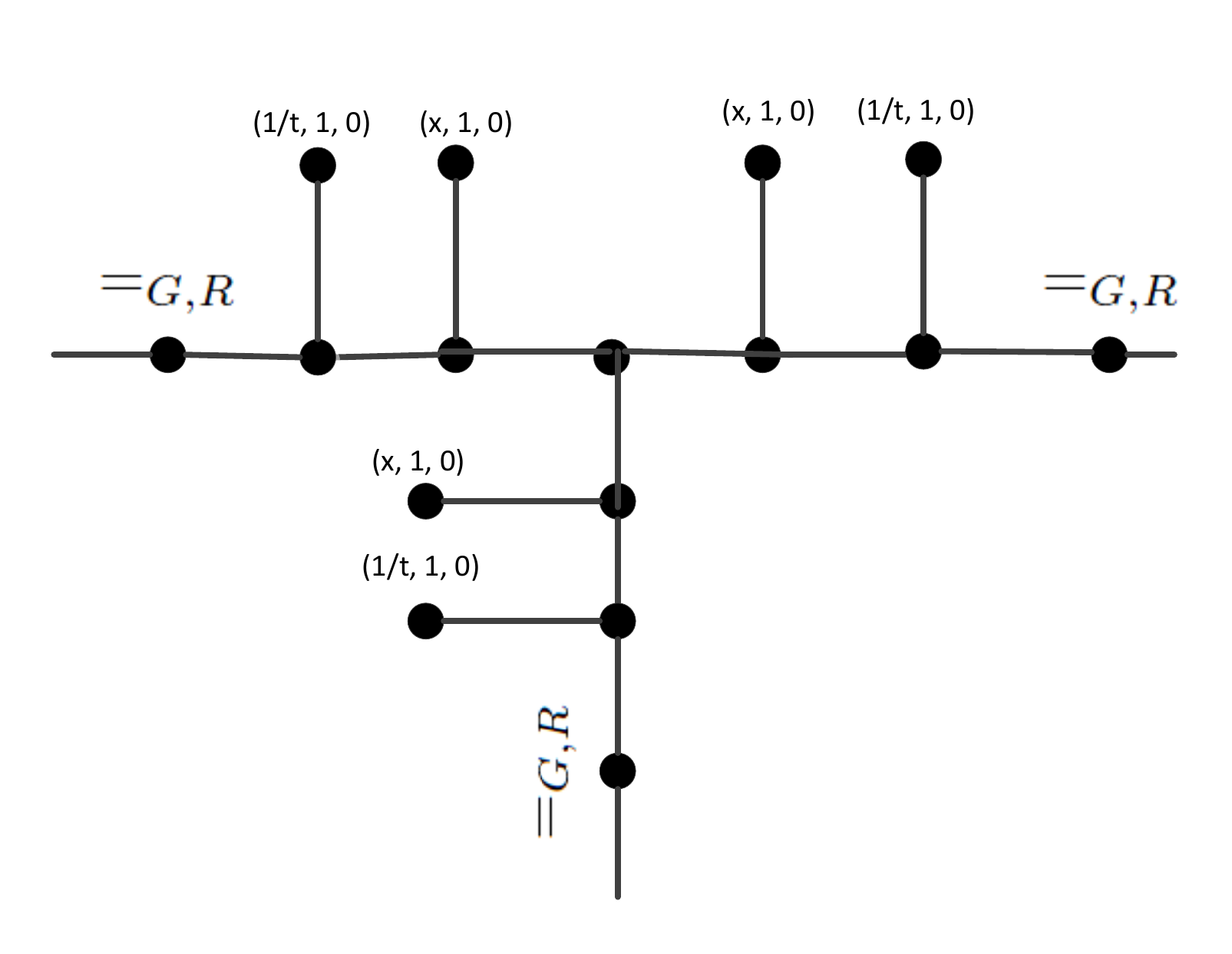}
        \caption{A ternary gadget.}
        \label{figure:extended-three-star}
        \end{center}
\end{figure}
We can calculate the binary function that is the linked chain
of $\langle \mathbf{u}_1, \mathbf{F}
\rangle$ and $\langle \mathbf{u}_2, \mathbf{F}
\rangle$, and in matrix form  it is
\[\begin{bmatrix} * & * & * \\
X & Z & -iZ \\
Y & -iZ & -Z \end{bmatrix}\]
where $X = x (u + 2t) + t$, $Y = - i X + 4ixt$ and $Z = xt$.

We can write $\begin{bmatrix} X & Z & -iZ \\
Y & -iZ & -Z \end{bmatrix}$ as $M_1 M_2$, where
\[ M_1 =
\begin{bmatrix}
X &   Z & 0\\
Y & -iZ & 0
\end{bmatrix}
\mbox{ and } M_2 =
\begin{bmatrix}
1 & 0 & 0\\
0 & 1 & -i \\
0 & 0 & 0
\end{bmatrix}
\]

It can be verified that  $(M_2)^{\otimes 3} \mathbf{F} $ is
\begin{center}
\begin{tabular}{*{11}{c}c}
{ }     & { }   & { }   &$u$& { }     & { }   & { }  \\
{ }     & { }   &$0$& { } &$0$& { }     & { }  \\
{ }     &$0$& { } &  $0$& { }   &$0$& { }  \\
 $8$& { } &$0$& { } & $0$& { } &   $0$  \\
\end{tabular}
\end{center}

So the signature of the gadget is
\[ (M_1)^{\otimes 3} [ (M_2)^{\otimes 3} \mathbf{F} ] = u \begin{bmatrix} X \\ Y \end{bmatrix}^{\otimes 3}  + 8 \begin{bmatrix} Z \\ -iZ \end{bmatrix}^{\otimes 3}.\]

Since $X = x (u + 2t) + t$, $Y = - i X + 4ixt$ and $Z = xt$,  we can
always either set
$X=0$ and $YZ \not =0$, or set $Y=0$ and $XZ \not =0$.
(When $u + 2t \not =0$, we can set $x = -t/(u + 2t) \not =0$.
When  $u + 2t  =0$, then $X = t \not =0$, and we set $x = 1/4$.)
This proves \#P-hardness given $u\neq 0$.

Finally, we suppose
$s=0$, $t  = i r$ and $u= 0$.
As $r \not = it$, we have $r, t \not =0$. $\mathbf{F}$  is
\begin{center}
\begin{tabular}{*{11}{c}c}
{ }     & { }   & { }   &$0$& { }     & { }   & { }  \\
{ }     & { }   &$ir$& { } &$r$& { }     & { }  \\
{ }     &$0$& { } &  $0$& { }   &$0$& { }  \\
$1$& { } &$i$& { } & $-1$& { } &   $-i$  \\
\end{tabular}
\end{center}

After a holographic reduction by
the matrix $T=\left ( \begin{array}{ccc}  \sqrt{2} & 0 & 0  \\  0 & 1 & -i  \\  0 & 1 & i  \end{array} \right )$, $\mathbf{F}$ becomes $\mathbf{H}=T^{\otimes 3} \mathbf{F}=$


\begin{center}
\begin{tabular}{*{11}{c}c}
{ }     & { }   & { }   &$0$& { }     & { }   & { }  \\
{ }     & { }   &$0$& { } &$4ri$& { }     & { }  \\
{ }     &$0$& { } &  $0$& { }   &$0$& { }  \\
$8$& { } &$0$& { } & $0$& { } &   $0$  \\
\end{tabular}
\end{center}

Our problem Holant$^*(\mathbf{F})$ can be restated as
Holant$(\{=_2\} \cup \mathcal{U}  |\{\mathbf{F}\} \cup \mathcal{U})$.
The left hand side $=_2$  becomes ${(T^{-1})}^{\tt T} I_3 T^{-1}$
which is a constant $1/2$ multiplied
by $(\neq_{B;G,R})$ (see equation (\ref{neq-BGR-def})).
Holant$^*(\mathbf{F})$ is holographic equivalent to Holant$(\{\neq_{B;G,R}\} \cup \mathcal{U}  |\{\mathbf{H}\} \cup \mathcal{U})$.

We can realize on the right hand side $\langle \mathbf{H},(x,y,z) \rangle $,
with a unary $(x,y,z)$,
 and   on the left  hand side $(\neq_{B;G,R})^{\otimes 3}  \mathbf{H}=$
\begin{center}
\begin{tabular}{*{11}{c}c}
{ }     & { }   & { }   &$0$& { }     & { }   & { }  \\
{ }     & { }   &$4ri$& { } &$0$& { }     & { }  \\
{ }     &$0$& { } &  $0$& { }   &$0$& { }  \\
$0$& { } &$0$& { } & $0$& { } &   $8$  \\
\end{tabular}
\end{center}

It is easy to see that
Holant$(\{(\neq_{B;G,R})^{\otimes 3}  \mathbf{H}\} \cup \mathcal{U}
\mid \{ \langle \mathbf{H},(x,y,z) \rangle \} \cup \mathcal{U} )$
can be simulated by, and therefore reducible to,
 Holant$(\{\neq_{G,R}\} \cup \mathcal{U}\mid \{\mathbf{H}\} \cup \mathcal{U})$.

Setting $x=0, y= \frac{1}{8}, z = \frac{1}{4 r i} $,
we have the binary function which is an equality on
$\{B, G\}$ and zero elsewhere,
$\langle \mathbf{H},(x,y,z) \rangle= (=_{B,G})$.
Then we can apply $(=_{B,G})$ to restrict
 $(\neq_{G,R})^{\otimes 3}  \mathbf{H}$ to the subdomain $\{B, G\}$.
Notice that
 $[(\neq_{G,R})^{\otimes 3}  \mathbf{H}]^{* \rightarrow \{B,G\}}$  is
$4ri [0, 1, 0, 0]$, where
$[0, 1, 0, 0]$ is the {\sc Perfect Matching} function of 3-regular graphs
on the domain $\{B,G\}$. Hence  we get a \#P-hard problem.

\subsection{$\mathbf{F}^{*\rightarrow\{G,R\}} = [0,0,0,0]$}\label{sec:0000}
We deal with the final case where $\mathbf{F}^{*\rightarrow\{G,R\}}$
is identically 0.
The signature $\mathbf{F}$ is of the form
\begin{center}
\begin{tabular}{*{11}{c}c}
{ }     & { }   & { }   &$u$& { }     & { }   & { }  \\
{ }     & { }   &$t$& { } &$r$& { }     & { }  \\
{ }     &$s$& { } &  $p$& { }   &$q$& { }  \\
 $0$& { } &$0$& { } & $0$& { } &   $0$  \\
\end{tabular}
\end{center}

Let $x$ and $y$ be such that $x^2 + y^2 = 1$,
then $H=\begin{bmatrix} x & y \\ y & -x \end{bmatrix}$ is an orthogonal
matrix. Note that   $y\neq \pm x i$.
We will use $H$ to normalize $(s, p, q)$. This transformation
happens in the domain $\{G, R\}$.
Formally, we perform a transformation in the whole domain
$\{B, G, R\}$ using the orthogonal matrix
$\hat{H} =
 \begin{bmatrix}1 & 0 & 0 \\ 0 & x & y \\ 0 &  y & -x \end{bmatrix}$.
Note that in $\hat{H}$ the domain $\{G, R\}$ is separated from
$\{B\}$.
$(\hat{H}^{\otimes 3} \mathbf{F})^{*\rightarrow\{G,R\}}$
is still [0,0,0,0].
This is because a value of $\hat{H}^{\otimes 3} \mathbf{F}$
under any assignment
that is restricted to $\{G, R\}$ only, after the transformation
$\hat{H}$, becomes a combination of values
of $\mathbf{F}^{*\rightarrow\{G,R\}}$ under a tensor
transformation of $H$, hence 0.
To compute the rest of the signature of $\hat{H}^{\otimes 3} \mathbf{F}$,
we may assign one input to $B$, and compute the binary
signature $(\hat{H}^{\otimes 3} \mathbf{F})^{1=B}$
on $\{B,G,R\}$. By the form of $\hat{H}$ this is the same
as $\hat{H}^{\otimes 2} (\mathbf{F}^{1=B})$.  The
matrix form is the matrix product
$\hat{H}  (\mathbf{F}^{1=B}) \hat{H}^{\tt T}$,
where the matrix form of $\mathbf{F}^{1=B}$ is
$\begin{bmatrix}u & t & r \\ t & s & p \\ r &  p & q \end{bmatrix}$.
Thus $\hat{H}^{\otimes 3} \mathbf{F}$ is

\begin{center}
\begin{tabular}{*{11}{c}c}
{ }     & { }   & { }   &$u'$& { }     & { }   & { }  \\
{ }     & { }   &$t'$& { } &$r'$& { }     & { }  \\
{ }     &$s'$& { } &  $p'$& { }   &$q'$& { }  \\
 $0$& { } &$0$& { } & $0$& { } &   $0$  \\
\end{tabular}
\end{center}
where
\begin{eqnarray*}
u' &=& u\\
t' &=& tx + r y\\
r' &=& -rx + ty\\
s' &=& s x^2 + 2p xy + q y^2\\
p' &=& -p x^2 + (s-q) xy + p y^2\\
q' &=& q x^2 - 2 p xy + s y^2
\end{eqnarray*}
We easily verify that $(q' - s' \pm 2 i p')
= (x \mp i y)^2 (q -s \mp 2 i p)$. Thus, for any given $x \not = \pm iy$,
we have
$q -s \mp 2 i p \not = 0$ iff $q' - s' \pm 2 i p' \not = 0$.

\begin{enumerate}
\item
Consider the case $s=p=q=0$. If $r= \pm i t$, then
$\mathbf{F}$ is in the third form of Theorem \ref{thm:ternary},
since the isotropic $(0, 1, \pm i)$ annihilates $\mathbf{F}$,
namely $\langle (0, 1, \pm i), \mathbf{F} \rangle = \mathbf{0}$.
If $r \not = \pm i t$, we can apply an orthogonal
transformation $\hat{H}$ with $rx = ty$,   such that $\mathbf{F}$ becomes
        \begin{center}
                \begin{tabular}{*{11}{c}c}
                { }     & { }   & { }   &$u$& { }     & { }   & { }  \\
                { }     & { }   &$t'$& { } &$0$& { }     & { }  \\
                { }     &$0$& { } &  $0$& { }   &$0$& { }  \\
                 $0$& { } &$0$& { } & $0$& { } &   $0$  \\
                \end{tabular}
                \end{center}
where $t' = tx + ry \not =0$.
                This gives a \#P-hard problem
on the domain $\{B, G\}$.

From now on not all $s, p, q = 0$.

\item Suppose $p = 0$. But either $s \not = 0$ or $q \not = 0$.
Suppose $s \not = 0$, by symmetry.
We get the binary $\mathbf{F}^{1 = G}$
\begin{center}
\begin{tabular}{*{11}{c}c}
{ }     & { }   & { }   &$t$& { }     & { }   & { }  \\
{ }     & { }   &$s$& { } &$0$& { }     & { }  \\
{ }     &$0$& { } &  $0$& { }   &$0$& { }  \\
\end{tabular}
\end{center}
which is effectively a domain two binary signature $[t, s, 0]$.
We can use this to interpolate $=_{B,G}$.
Then it becomes a solved case before, by the substitution of
$\{B,G\}$ for $\{G, R\}$. Note that
$\mathbf{F}^{* \rightarrow \{B,G\}}$ is not identically 0.

From now on $p \not = 0$.

\item
Suppose $q \not = s\pm 2 i p$.
We can set $y/x$ to be a solution to
$Y^2 + \frac{s-q}{p} Y -1 =0$, such that $Y \not = \pm i$.
Then  $\hat{H}^{\otimes 3} \mathbf{F}$ has $p'=0$.
Note that $q' \not = s' \pm 2 i p' = s'$,
thus at least one of $s'$ or $q' \not =0$.
Thus we have reduced this case to the previous case.

From now on $q = s\pm 2 i p$.

\item
Suppose $p \not = 0$, $q = s\pm 2 i p$, and either $s=0$ or $q=0$.
Then either $q = \pm 2 i p$ or $s = \mp 2 i p$.
By symmetry assume $q=0$.
We can get the binary $\mathbf{F}^{1 = R}$
\begin{center}
\begin{tabular}{*{11}{c}c}
{ }     & { }   & { }   &$r$& { }     & { }   & { }  \\
{ }     & { }   &$p$& { } &$0$& { }     & { }  \\
{ }     &$0$& { } &  $0$& { }   &$0$& { }  \\
\end{tabular}
\end{center}
which is effectively a domain two binary signature $[r, p, 0]$
on $\{B, G\}$.
We can use this to interpolate $=_{B,G}$.
Then it becomes a solved case before, by the substitution of
$\{B,G\}$ for $\{G, R\}$, since
$\mathbf{F}^{* \rightarrow \{B,G\}}$ is not identically 0.

From now on  both $s \not =0$ and $q \not = 0$.

\item
$s, p, q \not = 0$, $q = s\pm 2 i p$, but suppose $q \not = -s$.
Setting $y/x = \frac{2p}{s} \mp i$, we can verify that
this gives an orthogonal matrix $H$ such that $q' = 0$.
This reduces to the previous case.

\item
Finally we have $s, p, q \not = 0$ and $q = s\pm 2 i p = -s$.
Then $p = \mp si$ and $(s, p, q) = s (1, \mp i, -1)$.
We may normalize to $s=1$. We consider the case $(1, i, -1)$;
the case $(1, -i, -1)$ is symmetric.

If  $r=i t$, $\mathbf{F}$ is in the third form of Theorem \ref{thm:ternary},
since $\langle (0, 1, i), \mathbf{F} \rangle = \mathbf{0}$.
Now we suppose $r \not =i t$. We shall prove that the problem is
\#P-hard using the gadget in Figure~\ref{triangle-1}.
        \begin{figure}[hbtp]
        \begin{center}
                \includegraphics[width=3 in]{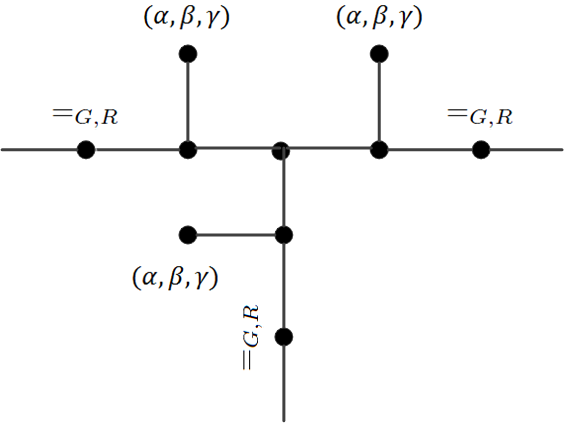}
        \caption{A ternary gadget.}
        \label{triangle-1}
        \end{center}
\end{figure}
The unary $\mathbf{u}=(\alpha, \beta, \gamma)$
with $\alpha=1, \gamma =0$ is chosen such that
the matrix form of $\langle \mathbf{u}, \mathbf{F}
\rangle$ is $M =
\begin{bmatrix} z & x & y \\ x & 1 & i \\ y & i & -1 \end{bmatrix}$,
where $x=t + \beta, y = r+ i \beta$.
Note that the ratio of $x, y$ can be arbitrary, by choosing $\beta$.

We wish to compute the signature of the gadget,
which is the domain two signature of the ternary
function $(M^{\otimes 3} \mathbf{F})^{* \rightarrow \{B,G\}}$.
We can decompose the $2 \times 3$ matrix
$\begin{bmatrix} x & 1 & i \\ y & i & -1 \end{bmatrix}$ as the product
 $M_1 M_2$, where
$M_1 = \begin{bmatrix} x & 1 & 0 \\ y & i & 0 \end{bmatrix}$,
and
$M_2 =
\begin{bmatrix} 1 & 0 & 0 \\ 0 & 1 & i \\ 0 & 0 & 0 \end{bmatrix}$.
Therefore we wish to compute
$ M_1^{\otimes 3} (M_2^{\otimes 3} \mathbf{F})$.

As $M_2$ has a separated domain form,
$(M_2^{\otimes 3} \mathbf{F})^{* \rightarrow \{G, R\}}$
is identically 0, since these values are combinations of values
from the bottom line $\mathbf{F}^{* \rightarrow \{G, R\}}$
by a tensor transformation.
To compute the other values of $M_2^{\otimes 3} \mathbf{F}$,
we may set one input of $M_2^{\otimes 3} \mathbf{F}$ to $B$,
which by the form of $M_2$ is the same as
$M_2^{\otimes 2} (\mathbf{F}^{1=B})$.
This can be computed  by a matrix product and the result is
$\begin{bmatrix} u & t+ir & 0 \\ t+ir & 0 & 0 \\ 0 & 0 & 0 \end{bmatrix}$.
Thus the signature $M_2^{\otimes 3} \mathbf{F}$ is
\begin{center}
                \begin{tabular}{*{11}{c}c}
                { }     & { }   & { }   &$u$& { }     & { }   & { }  \\
                { }     & { }   &$t+ir$& { } &$0$& { }     & { }  \\
                { }     &$0$& { } &  $0$& { }   &$0$& { }  \\
                 $0$& { } &$0$& { } & $0$& { } &   $0$  \\
                \end{tabular}
                \end{center}
which is
\[u \begin{bmatrix} 1 \\0 \\0 \end{bmatrix}^{\otimes 3}
+ (t+ir) \cdot \frac{1}{2}{\rm Sym} \left[
\begin{bmatrix} 1 \\0 \\0 \end{bmatrix}^{\otimes 2}
\otimes
\begin{bmatrix} 0 \\ 1 \\0 \end{bmatrix}
\right].
\]
(The symmetrization has six terms.)

Now we apply $M_1$, and get
\[u \begin{bmatrix} x \\y \end{bmatrix}^{\otimes 3}
+ (t+ir) \cdot \frac{1}{2} {\rm Sym} \left[
\begin{bmatrix} x \\y  \end{bmatrix}^{\otimes 2}
\otimes
\begin{bmatrix}  1 \\i \end{bmatrix}
\right].
\]

If we set $\beta = i r$, then $x = t+ir \not = 0$ and $y=0$.
This signature is $[(u+1) x^3, i x^3, 0, 0]$. Since  $x \not = 0$
this defines a \#P-hard problem on domain size two by Theorem \ref{thm:dich-sym-Boolean}.

\end{enumerate}


\bibliography{refs}

\end{document}